\documentclass[a4paper, 11pt]{article}

\usepackage{amsfonts, amsmath, amssymb, amscd, setspace, subfigure,mathtools,amsthm}
\usepackage{verbatim, geometry,mdframed,hyperref,mathtools}
\usepackage[final]{graphicx}
\usepackage{pictexwd,dcpic,float}
\usepackage[all,cmtip]{xy}
\newcommand{\dchia}{ R \, \mathrm{d} _A \phi _A }
\newcommand{\uu}{ \mathcal{U} }

\newcommand{\hi}{\Phi }

\newcommand{\f}{\mathcal{F} _{1\mathrm{c}  }}
\newcommand{\UU}{ \mathcal{V} }
\renewcommand{\L}{ \Lambda   }
\newcommand{\rrho}{\check{ \rho }   }

\newcommand{\qi}{ \mathbf{i}  }
\newcommand{\qj}{ \mathbf{j}  }
\newcommand{\qk}{ \mathbf{k}  }
\newcommand{\gh}{ \hat{g}  }
\newcommand{\pphi}{ \phi _{ \Alc}}
\newcommand{\ph}{ \phi ^{(3)}}

\newcommand{\cchi}{\vartheta}
\newcommand{\Ans}{ A _\nu ^s }
\newcommand{\gph}{g_{ P ^1 (\mathbb{H}  )  }}
\newcommand{\Aaa}{A ^{ (3) } _\lambda }
\newcommand{\Aa}{A ^{ (3) }  }
\newcommand{\Fff}{F ^{ (3) }}
\newcommand{\MM}{\mathcal{M} }
\newcommand{\MMf}{\mathcal{M} ^{ \mathrm{f} } }

\newcommand{\MMM}{\mathcal{M}_{1\mathrm{c}}  }
\newcommand{\MMMf}{\mathcal{M} _{1\mathrm{c}}^{ \mathrm{f} } }
\newcommand{\gm}{ g _{ \MMMf}  }
\newcommand{\gmm}{ g _{ \MMM}  }

\newcommand{\GG}{\mathcal{N} }
\newcommand{\HH}{\mathcal{H} }

\newcommand{\fp}[2]{\left(#1,#2 \right)_{ \mathbb{H}} }
\newcommand{\pg}[1]{#1 _{\perp} }
\newcommand{\AAA}{ A ^\prime  }
\newcommand{\A}{\hat A _{ \mathrm{std} }}
\newcommand{\Ar}{\hat A _{  R  }}
\newcommand{\Al}{\hat A _{ \lambda  }}
\newcommand{\An}{ A _{ \nu   }}
\newcommand{\All}{ A _{ \lambda    }}

\newcommand{\Ala}{\hat A _{ \lambda, \alpha   }}
\newcommand{\Alaa}{A _{ \lambda, \alpha   }}

\newcommand{\Alc}{\hat {A} ^{ \prime} _{ \lambda  }}

\newcommand{\fs}{  F }
\newcommand{\Agr}{\hat A _g ^{R}}

\newcommand{\Fgr}{\hat F _g ^{R}}

\newcommand{\omr}{ \omega _{ \mathrm{std} } ^R}
\newcommand{\aalpha}{ \varphi }
\renewcommand{\P}{P} 

\newtheorem{thm}{Theorem}
\newtheorem{lem}[thm]{Lemma}
\newtheorem{prop}[thm]{Proposition}
\newtheorem{coro}[thm]{Corollary}

\begin{document}

\title{Adiabatic dynamics of instantons on $S ^4 $ }
\author{
\bf{Guido Franchetti}\\
Institut f\"ur Theoretische Physik, Leibniz Universit\"at\\
Hannover, Germany\\
\tt{franchetti@math.uni-hannover.de}\\  \\
\bf{Bernd J. Schroers }\\
Department of Mathematics,
Heriot-Watt University\\  Edinburgh, United Kingdom\\
\tt{B.J.Schroers@hw.ac.uk}
}

\maketitle
\begin{abstract}
We define and  compute  the   $L^2$ metric  on the framed moduli space of circle invariant 1-instantons on the 4-sphere. This moduli space is four dimensional and our  metric is  $SO(3) \times U(1)$ symmetric. We study the behaviour of generic geodesics and show that the metric is geodesically incomplete. Circle-invariant instantons on the 4-sphere can also be viewed as hyperbolic monopoles, and we interpret our results from this  viewpoint. We  relate our results to work by Habermann on unframed instantons on the 4-sphere  and,  in the limit where the radius of the 4-sphere tends to infinity, to results on instantons on Euclidean 4-space. 
\end{abstract}


\vspace{0.1cm}

\section{Introduction}
\label{intro} 
In field theories with topological solitons, moduli spaces of solitons typically  inherit a natural $L^2$ metric from their embedding  in the infinite-dimensional configuration space of  the field theory \cite{Manton:2004tk}.  The explicit determination of  this metric  is  difficult, and often only possible via indirect methods relying on special properties like K\"ahler or hyperK\"ahler structures.  This is the case for some of the best-studied examples, like the  moduli spaces of abelian Higgs vortices or  of   BPS monopoles. 

There are few cases where  non-trivial metrics on moduli spaces can be obtained by a  parameterisation of  the  moduli space  and  an explicit computation of the integrals which define the $L^2$ metric. Examples include the moduli spaces of  two $P ^1 (\mathbb{C}  )$ lumps  on Euclidean 2-space \cite{ward:1985}, of a single $P ^1 (\mathbb{C}  )$ lump on the 2-sphere \cite{Speight:1997ub},  and of a single instanton on  Euclidean 4-space or on the 4-sphere \cite{Groisser:1987uq,Habermann:1988wv,Habermann:1993ud}.  Both for lumps and instantons, the conformal invariance of the defining equations allows one to switch from Euclidean to spherical geometry without changing the fields (and hence the moduli spaces). However,  the $L^2$ metric depends on the spatial geometry, and  is more particularly interesting in  the spherical case. 

In this paper we study   the framed  moduli space of  circle invariant instantons on the 4-sphere.  One  motivation   stems from the correspondence between circle invariant instantons and hyperbolic monopoles \cite{Atiyah:1987ua}.  
Since the $L^2$ metric on the moduli space of hyperbolic monopoles diverges  it has not yet been possible to determine the adiabatic dynamics of hyperbolic monopoles from the geometry of their moduli space, despite several efforts.  As we shall explain, the metric we compute here can, in a certain sense, be viewed as an answer to this problem.

A second motivation, which is at the same time more straightforward and more fundamental, is that circle invariant instantons on the 4-sphere provide a natural setting for exhibiting and studying  the various  issues associated with the framing of moduli spaces for gauge theories on   a compact manifold.  

Framing essentially amounts to declaring the value of gauge transformations at a  chosen point  as `physically relevant' and including it in the moduli space. It is often natural to do this for moduli spaces of topological solitons  in gauge theories since framing  restores overall  degrees of freedom which are  otherwise only visible as a relative phases in multi-soliton moduli spaces.

In order to compute the $L^2$ metric on the framed moduli space, one  requires the values of   gauge transformations not only at  the chosen point but on the entire manifold. On $\mathbb{R}^4$, the chosen point is usually `infinity' and the gauge transformations on the entire manifold are  determined  via Gauss's law. However, on a compact (oriented, simply connected) 4-manifold, Gauss's law has no non-trivial solutions, so that  the  generator of  `physically relevant' gauge transformations has to be constructed differently. In this paper we propose a natural and explicit construction of this generator as the vertical component, in the decomposition determined by the connection, of the vector field generating the circle action on the total space of the instanton bundle.  

Our framing prescription equips the  moduli space with  an additional circle. This fits well with  the interpretation of circle invariant instantons as hyperbolic monopoles where framing also leads to  an additional circle of gauge transformations,  generated by  the Higgs field of the monopole. In fact we will show how to obtain this Higgs field directly  from our construction.

The moduli space of all 1-instantons  on the 4-sphere, as defined and studied in detail  in  \cite{Habermann:1988wv,Groisser:1987uq}, is 5-dimensional, with one parameter for  the scale  of the instanton and four for  its position on the 4-sphere. The moduli space of circle-invariant instantons  can be identified with a submanifold of this space. In fact, invariance under a circle action forces the   centre of the instanton to lie on a fixed 2-sphere  inside the 4-sphere,  and therefore cuts down the total number of parameters to three.
As we shall see, the restriction to circle-invariant instantons is mathematically very natural. In particular we show that, in  a description of the instanton gauge fields in terms of quaternions, it essentially amounts to replacing quaternionic parameters with complex ones.

The framed moduli space of circle invariant instantons is therefore a 4-dimensional Riemannian manifold, and we shall see that the additional circle allows for a much more interesting geodesic behaviour than that on the unframed moduli space. In fact, the qualitative behaviour of geodesics is remarkably similar to that found for a single lump on a 2-sphere in \cite{Speight:1997ub}. 
 
Throughout this paper we work on a 4-sphere of arbitrary radius $R$. This allows us to obtain  the $L ^2 $ metric of the moduli space of instantons over Euclidean $\mathbb{R}  ^4 $ in the limit $R\rightarrow \infty$. While this limit has been investigated before    \cite{Habermann:1993ud},  our treatment is more direct and explicit.

Our paper is organised as follows.  In Section \ref{s1} we review some basic material about instantons. In particular we introduce two convenient parameterisations of the moduli space of 1-instantons and explicitly describe their relation and geometrical meaning. In Section  \ref{s2} we carefully define the notion of circle invariance and  describe the framed moduli space of circle-invariant 1-instantons, which   turns out to be a trivial circle bundle over an open 3-ball.  In Section \ref{s3},   we define and compute the $L ^2 $-metric on this moduli space.
The resulting metric has an $SO (3) \times U (1) $ symmetry, is non-singular but incomplete and has positive, non-constant scalar curvature. We then discuss some properties of this metric and the behaviour of its geodesics. In the short final Section \ref{s5} we discuss our results and  draw some conclusions. 

\section{Preliminaries}
\label{s1} 
In this section we recall some material about instantons on the 4-sphere. We introduce two convenient parameterisations of the moduli space of 1-instantons, discuss  their geometrical meaning, and describe the framed moduli space. Readers interested in more details on this background material can consult e.g.~\cite{Naber:1500641,Groisser:1987uq,Groisser:1989vl,Atiyah:EHEgiu5y}. 

\subsection{Instantons on  $S ^4 _R   $}
We consider pure $SU (2) $ Yang-Mills theory on $S ^4 _R $, the round 4-sphere of radius $R$. We are keeping $R$ arbitrary as we will be interested in the limit $R \rightarrow \infty $.  It is convenient to identify $ S ^4 _R $ with the quaternionic projective plane $P ^1 (\mathbb{H}  )$. Our conventions for quaternions, which are mostly standard and follow those in \cite{Naber:1500641}, are described in Appendix \ref{qnot}, where we also introduce the quaternionic-valued matrix groups we shall consider in this paper, namely $SL(2, \mathbb{H}  ) $, $Sp (2) $ and $Sp (1) $. The quaternionic projective plane can be defined as the quotient of the 7-sphere by the group $S p (1)  \simeq  SU (2) $ of unit quaternions,
\begin{equation}
\begin{split} 
&P ^1 (\mathbb{H}  )  
=\{ [ q ^1 , q ^2 ]:( q ^1 , q ^2 )\in  S ^7 _{ \sqrt{ R } }\subset \mathbb{H}  ^2 \}, \\
& [q ^1 , q ^2]=[r ^1 , r ^2 ] \Leftrightarrow ( r ^1 , r ^2 ) =( q ^1 h, q ^2 h ) \text{ for some } h\in \mathbb{H} , |h |=1.
 \end{split} 
\end{equation}
Moreover, we identify $SU (2) $ with $Sp (1) $, and Euclidean $\mathbb{R}  ^4 $ with the space $\mathbb{H}  $ of quaternions  carrying the flat metric 
\begin{equation}
\label{flm} 
\gh = \mathrm{d} \bar q \, \mathrm{d} q .
\end{equation}
The manifolds $S ^4 _R $ and $ P ^1 (\mathbb{H}  ) $  are diffeomorphic, but  the  metric
\begin{equation}
\label{gp1h} 
\gph =\frac{R ^4 }{(R ^2  + |q| ^2 ) ^2 }\,  \gh
\end{equation} 
on $P ^1 (\mathbb{H}  )$ inherited via Hopf projection from the standard metric on 
$\mathbb{C}  ^4 $  is one quarter of the round metric on $S ^4 _R $.
The advantage of working with (\ref{gp1h}) is that it converges to $\gh $ in the limit $R \rightarrow \infty $.

Our mathematical framework  is that of  a principal bundle with a connection $\omega$ on it. We will focus on the principal bundle $P$ given by
\begin{equation}
\xymatrix{
Sp (1) \ar@{^{(}->}[r] & S ^7 _{ \sqrt{ R }} \ar[d] ^ \pi \\
& P ^1 (\mathbb{H}  ).
}
\end{equation} 
The total space of $P$ is a 7-sphere of  radius  $\sqrt{ R }$ which is mapped onto $P ^1 (\mathbb{H}  ) $ by the Hopf projection. 
Let 
\begin{equation}
\hat \phi _N : P ^1 (\mathbb{H}  ) \setminus [R,0] \rightarrow \mathbb{H}, \qquad   [q ^1 , q ^2 ] \mapsto R \, q ^1  (q ^2  )^{-1}
\end{equation}
be the map corresponding to stereographic projection from the north pole of $S ^4 _R $, see Appendix \ref{stp} for more details.
We will frequently consider the pullback bundle
$\hat P = (\hat\phi _N ^{-1} ) ^\ast P $.
A connection $\omega$ on $P$ is  uniquely determined by the gauge potential $A =\sigma ^\ast \omega $, where $\sigma$ is a section of the pullback bundle 
$\hat P $. In fact, it follows from the standard theory of principal bundles and connections, see e.g.~\cite{Naber:1500641}, that $\omega$ is determined by the pair $\{ \sigma ^\ast \omega , \tau ^\ast \omega \}$, where $\tau$ is a local section of $P$  over  $P ^1 (\mathbb{H}  ) \setminus \{ [0,R] \}$. On the intersection $ P ^1 (\mathbb{H}  ) \setminus \{[R,0],[0,R]\} $ of the domains of the local sections $\sigma$ and $\tau$, the gauge potential $\tau ^\ast \omega $ can be expressed in terms of $\sigma ^\ast \omega $ and of the transition function between the two sections, and at the remaining point $[0,R] $ it is determined by continuity.
%

If $ \zeta $ is some object defined on $P $, 
we denote by  $ \hat  \zeta  =  (\hat \phi _N ^{-1} )^\ast \zeta  $ the corresponding object on $\hat P $, and vice versa if $\hat\xi $ is some object defined on $ \hat P $, we denote by  $  \xi   =  (\hat \phi _N  )^\ast \hat\xi   $ the corresponding object on $P $.
On $\Lambda ^p (P ^1 (\mathbb{H}  ), \mathrm{ad} (P)  ) $ we take the $L ^2 $ inner product
\begin{equation}
\label{ip} 
\langle \zeta   ,\xi   \rangle 
=  - \frac{1}{2} \int_{P ^1 (\mathbb{H}  )} \mathrm{Tr} \left( \zeta   \wedge * \xi    \right),
\end{equation} 
and denote by $|| \cdot || $ the induced norm.
If $\hat \zeta  , \hat \xi   \in \Lambda ^p  (\mathbb{H}  , \mathfrak{sp} (1)) $, we define
\begin{equation}
\begin{split} 
\fp{ \hat \zeta}{ \hat \xi   }&
=  - \frac{1}{2} \mathrm{Tr} \left( \hat \zeta  \wedge \hat * \hat \xi    \right),
\end{split} 
\end{equation} 
with $\hat * $ the Hodge operator with respect to the metric $\gh $, and by $| \hat \zeta   | ^2 _{ \mathbb{H}  } \, \mathrm{vol} _{ \mathbb{H}  } =  \fp{ \hat \zeta  }{\hat  \zeta   } $.

The curvature of a connection  defines the field strength $\fs \in \Lambda ^2 (P ^1 (\mathbb{H}  ), \mathrm{ad}(\P)) $, with $\mathrm{ad}(\P)$  the linear adjoint bundle $ P \times _{ \mathrm{ad}} \mathfrak{sp} (1) $.  By remarks similar to those made above, 
$\fs $ is uniquely determined by the field strength $ (\hat\phi _N ^{-1}) ^\ast  \fs $  induced by $ \fs $ on  $\hat P $. 
A connection $\omega$ is said to be  a Yang-Mills connection if  $\mathrm{d} _\omega *  F  =0$,  to be anti self-dual if its field strength satisfies the anti self-dual Yang-Mills (ASDYM) equations  $ *  \fs = -  \fs $, and to be self-dual if $*\fs =\fs $, with $ * $  the Hodge operator with respect to the metric (\ref{gp1h})  on $P ^1 (\mathbb{H}  ) $. Because of the Bianchi identities, a self-dual/anti self-dual connection is a Yang-Mills connection.
A non-singular self-dual/anti self-dual connection with finite action is called an instanton.

Principal  $Sp (1) $ bundles over $P ^1 (\mathbb{H}  ) $ are classified by the second Chern number, the integer
\begin{equation}
\label{c2c2c2} 
c _2  (P ) =\frac{1}{8 \pi ^2 }\int _{ P ^1 (\mathbb{H}  ) }\mathrm{Tr} \left( \fs \wedge \fs \right) .
\end{equation} 
The quantity $-c _2 $ is known as instanton number. 
By the usual Bogomolny argument the Yang-Mills action can be written as
\begin{equation}
- \int _{ P ^1 (\mathbb{H}  ) }  \mathrm{Tr}( F \wedge * F )
= - \frac{1}{2} \int \mathrm{Tr} \left[ (F \pm * F )\wedge * (F \pm * F )\right]  \pm 8 \pi ^2 c _2.
\end{equation} 
Therefore, for each value of $c _2 $, instantons are absolute minima of the Yang-Mills action which takes the value $ 8 \pi ^2 |c _2 |$.
Note that the Yang-Mills action is conformally invariant.

The first non-trivial solution of the ADSYM equations was discovered by Belavin, Polyakov, Schwartz and Tyupkin \cite{Belavin:428654} and has instanton number one.
In quaternionic notation this solution, which we call the standard instanton, is given by 
\begin{equation}
\label{stdinst} 
\A 
=  \frac{\Im ( \bar q\, \mathrm{d} q) }{1 + |q |^2 }.
\end{equation} 
We regard it as the pullback to $\hat P $ of the corresponding gauge potential on $P ^1 (\mathbb{H}  ) $.

The moduli space of  Yang-Mills instantons on a principal $Sp (1) $ bundle over $P ^1 (\mathbb{H}  )$ is defined to be $ \MM =   \mathcal{A} ^- / \mathcal{G} $, where $ \mathcal{A} ^- $ is the space of anti self-dual connections and  $\mathcal{G} $ the group of bundle automorphisms,
\begin{equation} 
\label{bauto} 
\mathcal{G} =\{ f: P \rightarrow P\  | \ g\in Sp (1), \,  p\in P \Rightarrow  f(p \cdot g )= f (p) \cdot g ,\,  \pi \circ f = \pi \},
\end{equation} 
acting on a connection by pullback. We say that two connections $ \omega $, $ \omega ^\prime $ are equivalent if  $ \omega ^\prime =f ^\ast \omega $ for $f\in \mathcal{G} $.
We recall that a bundle automorphism $f$ determines and is determined by
 either
  the $\mathrm{ad}$-equivariant map $g _f: P \rightarrow Sp (1) $ defined by 
\begin{equation}
\label{gadeqv} 
f(p ) = p \cdot g _f (p)
\end{equation}
or by the map $h _f : P ^1 (\mathbb{H}  ) \rightarrow P \times _ \mathrm{ad} Sp (1) $ defined for any $p\in \pi ^{-1} (x) $ by 
\begin{equation}
\label{hf} 
h _f (x) =[p, g _f (p) ]  .
\end{equation} 
The moduli space $\MM _ {n} $  of instantons  with instanton number $n$ is a smooth manifold of dimension $8n - 3 $ \cite{Atiyah:EHEgiu5y}. Because of the conformal invariance of the ASDYM equations and of Uhlenbeck's removable singularities theorem \cite{Uhlenbeck:1982bk}, $\MM _n $ is diffeomorphic to the moduli space of Yang-Mills instantons with instanton number $n$ on a principal $Sp (1) $ bundle over $\mathbb{H}  $.\footnote{A Yang-Mills  instanton on a $Sp (1) $ principal bundle over $\mathbb{H}  $ is an anti self-dual connection whose field strength has finite $L  ^2 $ norm. Because of the latter condition, the second Chern number is still defined by (\ref{c2c2c2}).  }

Let us consider the case $n =1 $. Since the ASDYM equations are conformally invariant in 4-dimensions, it is natural to consider the action of the conformal group of the 4-sphere $SO(5,1) $ on solutions of the ASDYM equations. The double cover $\mathrm{Spin}(5,1)= SL(2, \mathbb{H}  )$ of $SO(5,1) $  acts on $ S ^7 _{ \sqrt{ R } } \subset \mathbb{H}  ^2 $ as
\begin{equation}
\label{acts7} 
\left( g ^{-1}  =
\left(\begin{array}{cc} a  & b  \\ c  & d \end{array}\right), \left(  \begin{array}{c}q ^1 \\ q ^2 \end{array} \right) \right) \xmapsto{ \rho }
\frac{ \sqrt{R}}{|g ^{-1} (q ^1 , q ^2 )|}\left( \begin{array}{c} a  q ^1 + b  q ^2 \\ c  q ^1 + d  q ^2 \end{array} \right),
\end{equation}
where $|g ^{-1} (q ^1 , q ^2 )|  = \sqrt{ | a  q ^1 + b  q ^2 |^2 +  | c  q  ^1 +  d q ^2  |^2} $,
and induces an $SL(2 , \mathbb{H}  ) $ action on the space $\mathcal{A} $ of connections given by  
\begin{equation} 
(g, \omega )\mapsto g \cdot \omega =  \rho _{ g ^{-1} }^\ast \omega  ,
\end{equation}
where $ \rho _g =\rho (g, \cdot ) $.
  We pull back by $\rho _{ g ^{-1} } $ in order to obtain a left action. Note that $\rho$ descends to an action $\rrho $ of $SL(2, \mathbb{H}  )$  on $P ^1 (\mathbb{H}  ) $, given by
\begin{equation} 
\label{sl2ha} 
\left(g ^{-1} =\begin{pmatrix} a &b\\ c &d\end{pmatrix},  [ q ^1 , q ^2 ] \right)  \xmapsto{ \rrho }
\left[\sqrt{R} \, \frac{ a  q ^1 + b  q ^2}{ |g ^{-1} (q ^1 , q ^2 )|}, \sqrt{R} \, \frac{cq ^1 + d  q ^2}{ |g ^{-1} (q ^1 , q ^2 )|} \right].
\end{equation} 
Elements $\pm g\in SL(2, \mathbb{H}  ) $ induce the same transformation on $ P ^1 (\mathbb{H}  ) $, whence 
the group homomorphism $g \mapsto \rrho _g $ is surjective with kernel $\pm I $. In other words,  only $SL(2, \mathbb{H}  )/ \mathbb{Z}  _2 $ acts faithfully on $ P ^1 (\mathbb{H}  ) $.
 The isometry group  $SO (5)  $ of the 4-sphere has $ \mathrm{Spin}  (5)  \simeq   Sp (2) \subset SL(2, \mathbb{H}  )$ as its double cover. 


Consider the canonical connection
$\omr $ on $ S ^7 _{ \sqrt{ R } } $ , 
\begin{equation}
\omr=  \frac{1}{R}\Im  \left( \bar q ^1\, \mathrm{d} q ^1 + \bar q ^2 \,  \mathrm{d} q ^2  \right),
\end{equation} 
where we need to divide by $R$ to keep into account the fact that the connection is defined on a sphere of radius $ \sqrt{ R } $ rather than unitary.
The action of $SL(2, \mathbb{H}  )$  on $\omr $ is 
\begin{equation}
\begin{split} 
\label{acto} 
g \cdot \omr &
=\frac{\Im \left[ 
 (|a |^2  + |c |^2 )\bar q ^1 \mathrm{d} q ^1 
 +( |b |^2 + |d|^2 ) \bar q ^2 \mathrm{d} q ^2
 + \bar q ^2 (\bar b a + \bar d c) \mathrm{d} q ^1
 + \bar q ^1 (\bar a b + \bar c d) \mathrm{d} q ^2 
  \right] }{|a q ^1 + b q ^2 |^2 + |c q ^1 + dq ^2  |^2 }.
  \end{split} 
\end{equation}  
It is easy to check that the stabiliser of $\omr $ is $S p (2) $.

For later use, we derive the induced action on gauge potentials. Let $V _N =\{[ q ^1 , q ^2 ]\in P ^1 (\mathbb{H}  ): q ^2 \neq 0 \} $, take the section
\begin{equation}
s _N : V _N \rightarrow \pi ^{-1}  (V _N) , \quad [ q ^1 , q ^2 ] \mapsto \left( q ^1  \left( q ^2  \right) ^{-1}  |q ^2 |, |q ^2 | \right) 
\end{equation}  
of the principal bundle $ \P$ and compose it with $\hat \phi _N ^{-1} $   to obtain a section $s$ 
 of  $\hat P $,
 \begin{equation}
 \label{ours} 
s =s _N \circ \hat\phi ^{-1} _N : \mathbb{H}   \rightarrow S ^7 _{ \sqrt{ R}} \subset \mathbb{H} ^2, \quad
q \mapsto \sqrt{\frac{R}{ R ^2  + |q |^2 }} (q, R ).
\end{equation} 
The pullback via $s$  of the canonical connection $\omr $ is
\begin{equation}
\Ar
= 
 s ^\ast \omr
=\frac{\Im \left( \bar q\, \mathrm{d} q \right)  }{R ^2 +  |q |^2 }.
\end{equation} 
In particular for $R =1 $, we get the standard instanton (\ref{stdinst})  on $S ^4 $.
The pullback via $s$ of (\ref{acto}), for $g ^{-1} $ as in (\ref{acts7}), is
\begin{equation}
\label{atrgen} 
\Agr
=(\rho _{g ^{-1} } \circ s  ) ^\ast \omr 
= \frac{1}{|aq + bR|^2 + |cq + dR  |^2 } \Im  \left[ R( \bar b a +  \bar d c) \,\mathrm{d}q  + \left(  |a |^2 + |c|^2 \right) \bar q \, \mathrm{d} q\right] .
\end{equation} 
We also write  $ g \cdot \Ar $ for the gauge potential obtained pulling back $ g \cdot \omr $ by $s$ as above. 
We write $\Fgr=\mathrm{d} \Agr + \Agr \wedge \Agr$ for the field strength associated to $\Agr$.

\subsection{Two parameterisations of $ \MM _1 $ and their geometrical meaning}
\label{ss2} 

It was shown in \cite{Atiyah:EHEgiu5y} that the action of $SL(2, \mathbb{H}  ) $  on $\MM _1 $, the (unframed) moduli space of 1-instantons, is transitive with stabiliser $Sp (2) $, so that $\MM _1 $ is diffeomorphic to $SL(2, \mathbb{H} )/ Sp (2) $. By choosing a convenient parameterisation of $SL(2, \mathbb{H}  ) $ we obtain a corresponding description of $\MM _1 $.  We  now review two such descriptions, both given in \cite{Naber:1500641},  and interpret them geometrically. 
 
The first corresponds to the Iwasawa decomposition
\begin{equation}
SL (2, \mathbb{H}  ) =N A \, Sp (2) ,
\end{equation}
 where
\begin{align} 
\label{nnn} 
N &  
=\left \{b _n =   \begin{pmatrix}
1 &n/R\\ 0 &1
\end{pmatrix}\in SL(2, \mathbb{H}  )
: n\in \mathbb{H}  
\right\},\\
\label{aaa} 
A &
=\left \{ a _ \nu    
=  \begin{pmatrix}
\sqrt{\nu/R   } &0\\ 0 & \sqrt{R/\nu } 
\end{pmatrix}\in SL(2, \mathbb{H}  )
: \nu    >0
\right\}.
\end{align} 
The factor of $R$ has been inserted in order for the parameters $\nu $, $n$ to retain their usual meaning of scale and centre of an instanton, defined below,
for instantons on a 4-sphere of radius $R$.

The action of elements in $N A $ on $\Ar $  is
\begin{equation}
\label{instnn} 
b _n \,  a _ \nu   \cdot \Ar
=\frac{\Im \left[ (\bar q -  \bar n) \, \mathrm{d} q  \right] }{\nu ^2 + |q -  n |^2 }.
\end{equation} 
By computing the gauge invariant quantity
$| \Fgr |^2 _{ \mathbb{H}  } $ one can verify that as $g = b _n \,  a _\nu   $ varies in $N  A$, the corresponding gauge potentials  are all  inequivalent.
It follows that 
\begin{equation}
\label{nava} 
\MM _1 \simeq   SL(2, \mathbb{H}  )/Sp (2) = NA
\end{equation}
 is diffeomorphic to $\{(n, \nu )\in \mathbb{H}  \times \mathbb{R}   : \nu  >0 \} $, i.e.~the upper half space $\nu  >0 $ in $\mathbb{R}  ^5 $. 
Since $ N A = N \rtimes A $ is a semidirect product,
its action on a point $(n, \nu)$ of the moduli space is given by
\begin{equation}
(  n, \nu  )  \xmapsto{ n ^\prime \nu ^\prime  } (  n ^\prime + \nu^\prime n, \nu ^\prime \nu).
\end{equation}

This parameterisation is well suited to instantons on $\mathbb{H}  $, where the interpretation of the parameters $\nu$ and $n$ is particularly clear. The action density $ - \mathrm{Tr} ( F \wedge * F  ) $ has an absolute maximum for $q =n $. Exactly one half of the total action, $4 \pi ^2 $ for instanton number one, is obtained by integrating the action density over the ball of centre $n$ and radius $\nu$. Therefore, $n$ can be thought as the centre of the instanton and $\nu$ as its scale. 

The second parameterisation is obtained from the decomposition  
\begin{equation}
SL(2, \mathbb{H}  ) =Sp (2)  \, \AAA \, Sp(2) ,
\end{equation} 
where $\AAA $ is the subset of $A$ given by
\begin{equation}
\label{aares} 
\AAA = \{ a _\lambda \in A : \lambda \in(0,R] \}.
\end{equation}
To parametrise $ \mathcal{M} _1 $ it is convenient to write  $Sp (2) $ as the union of its left cosets $g ( S p (1) \times S p (1) )$, $g\in Sp (2) $. 
The action of $Sp (2) $ on $P ^1 (\mathbb{H}  ) $ is transitive with stabiliser $Sp (1) \times Sp (1) $, the double cover of $SO (4) $. Therefore $P ^1 (\mathbb{H}  )\sim  Sp (2) / Sp (1) \times Sp (1) $. 
For $m\in \hat { \mathbb{H}  }$, the one-point compactification of $\mathbb{H}  $, define $g _m \in Sp (2) $ by
\begin{equation}
\label{gm} 
g _m =\begin{cases}
\frac{1}{\sqrt{ R ^2  + |m |^2 }}\left( 
\begin{array}{ccc} R & m \\ - \bar m &R\end{array} 
\right) &\text{if $m\in \mathbb{H}  $},\\
\left( \begin{array}{ccc} 0 & 1 \\ - 1 &0\end{array} 
\right) &\text{if $m =\infty   $}.
\end{cases} 
\end{equation} 
The action of $\{ g _m : m\in \hat{\mathbb{H}} \}$ on $P ^1 (\mathbb{H}  )$ is free and transitive. Therefore
\begin{equation}
\label{cdec} 
SL (2, \mathbb{H}  ) 
=\bigcup_{g\in Sp (2) } g \left(  Sp (1) \times Sp (1) \right) \, \AAA \, Sp (2) 
=\bigcup_{m\in \hat{\mathbb{H}}} g _m \, \AAA\,  Sp (2) ,
\end{equation} 
where we used $[\AAA, Sp (1) \times Sp (1) ]=0$ and $ (Sp (1) \times Sp (1)) Sp (2) =Sp (2) $.
 Note  that $\{ g _m : m\in \hat{\mathbb{H}} \}$ is not a subgroup of $S p (2) $.

The action of an element $g _m a _\lambda\in  \cup_{m} g _m \, \AAA$ on $\Ar $ is 
\begin{equation}
\label{ppp2} 
g _m a _{ \lambda  } \cdot \Ar 
=\frac{\Im \left[\left( \lambda ^2/ R ^2  - 1 \right)\bar m\,   \mathrm{d} q + \left(1 + \lambda ^2  |m |^2/ R ^4    \right) \bar q \, \mathrm{d} q  \right] }{|q -  m  |^2 +  (\lambda ^2/ R ^2 )  |\bar m q/R + R |^2 }.
\end{equation} 
For future reference, the gauge potential obtained pulling back $ a _\lambda \cdot  \omr $ via $s$ is
\begin{equation}
\label{al} 
\Al = a _\lambda \cdot \Ar = \frac{\Im ( \bar q \, \mathrm{d} q )}{\lambda ^2 + |q| ^2 }.
\end{equation} 
By computing $| \hat F ^R _{ g } |^2 _{ \mathbb{H}  } $, with $g =g _m a _\lambda $, it is possible to verify that the instanton with parameters $m\in \hat{ \mathbb{H}  }$, $\lambda\in(R, \infty )$ is equivalent to the one with parameters $ - R ^2 /\bar m$, $ R ^2 / \lambda $.   Hence we obtain all the inequivalent instantons by restricting $\lambda$ to the interval $ (0,R] $.
The standard instanton $\Ar $  is $Sp (2) $-invariant, hence $ g _m \cdot \Ar =\Ar$ for all $m\in \hat{ \mathbb{H}  }$. Therefore
\begin{equation} 
\label{opppp}
\MM _1 
\simeq   SL(2, \mathbb{H}  )/Sp (2) = \bigcup _{m\in \hat{ \mathbb{H}  } } g _m \AAA
\end{equation}
 is  diffeomorphic to the open ball $ \mathring B ^5 _R  $  with radius $R$. The fixed point of the $Sp (2) $ action occurs for $\lambda =R $, which corresponds to the standard instanton  $\Ar $. The boundary of $ \mathring B ^5 _R  $ parametrises  singular ``small'' instantons for which $\lambda =0 $ and is not part of $\MM _1 $. A sketch of the moduli space can be found in Figure \ref{mods}.
\begin{figure}[htbp]
\begin{center}
\includegraphics[scale=0.5]{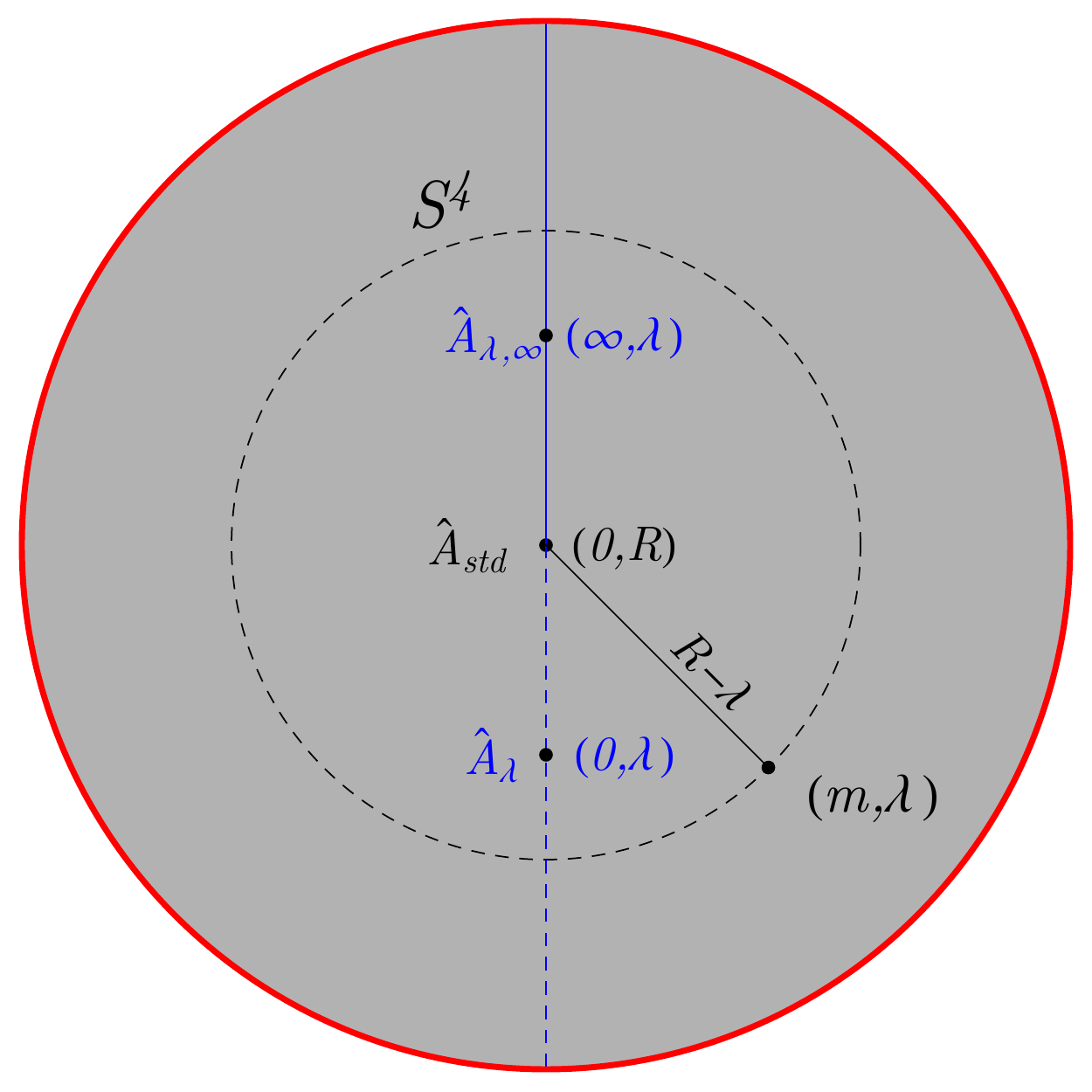} 
\caption{Sketch of $\MM _1  \simeq   \mathring{B} ^5 _R \subset \mathbb{R}  ^5 $. The standard instanton $\Ar $ has $m =0, \lambda =R $ and sits at the centre of the open 5-ball whose boundary, in red, corresponds to singular ``small'' instantons for which $\lambda =0 $. The vertical axis in $\mathbb{R}  ^5$ parametrises instantons with $\lambda \in(0, R ] $ and centre $m =0 $ (dashed vertical line)  or $m =\infty $ (continuous vertical line). }
\label{mods}
\end{center}
\end{figure}


The parameterisation (\ref{opppp})  is better suited to instantons on $P ^1 (\mathbb{H}  )$. To the best of our knowledge, its geometrical meaning has not been described in the literature, so we shall spend a few words doing so. For this purpose, it is convenient to temporarily view $P ^1 (\mathbb{H}  )$ as the 4-sphere $ S ^4 _R $. 

The centre $c$ and scale $l$ of an instanton on $ S ^4 _R $ are defined similarly to the centre and scale of an instanton on $ \mathbb{H}  $, $c$  being the location of the absolute maximum of the action density, and half of the action density being enclosed in the ball $B$ of radius $l$ and centred at $c$ (with respect to the metric on $ S ^4 _R $ induced by Euclidean metric on $\mathbb{R}  ^5 $). Note that the value of $l$ is equal to the length of the shortest great circle arc connecting $c$ to any point on the boundary of $B$, see Figure \ref{geo}.
\begin{figure}[htbp]
\begin{center}
\includegraphics[scale=0.6]{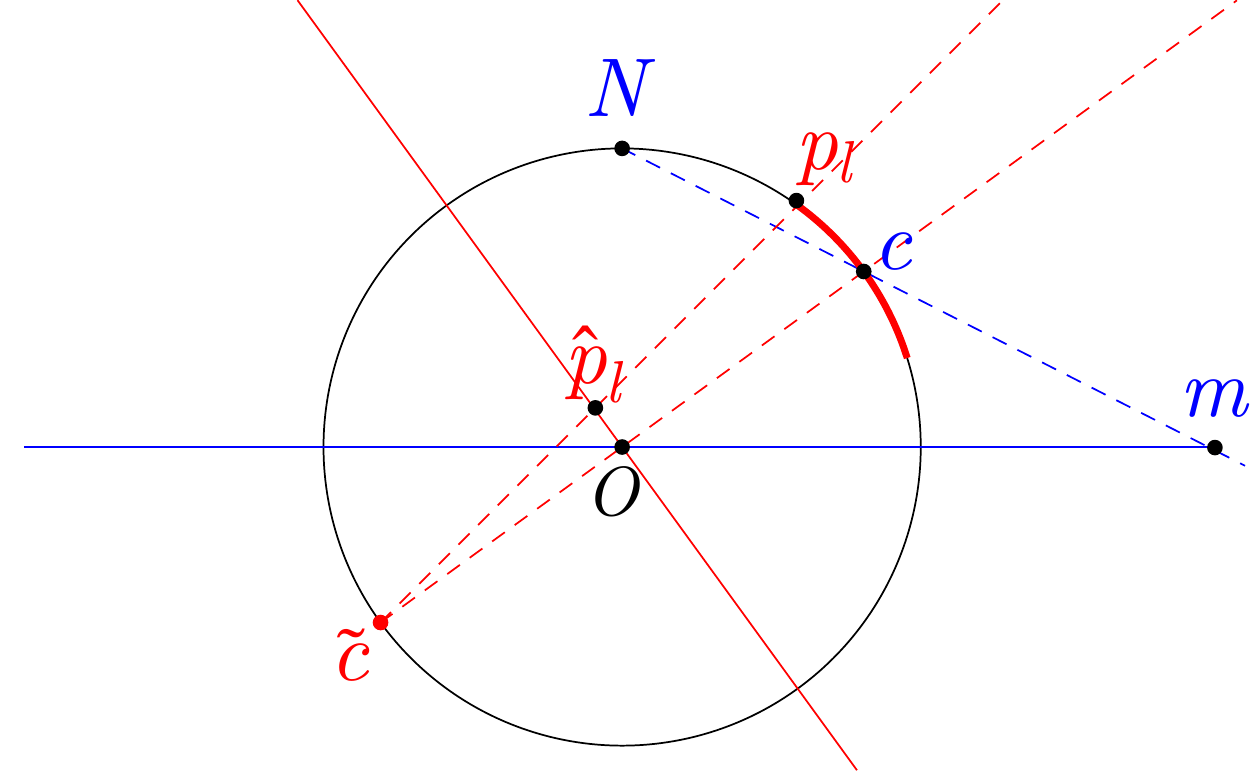}
\caption{The black circle represents the round 4-sphere $S ^4 _R \subset \mathbb{R}  ^5 $ with centre $O $. The red  arc corresponds to the 4-ball $B$ mentioned in the text. The red and blue continuous lines represent the 4-planes through $O$ orthogonal to $ \overline{ O  \tilde c } $ and $ \overline{ O  N } $.  The point $c$ is the centre of an instanton on $S ^4 _R $,  and the instanton scale $l$ is equal to the length of the shortest great circle arc connecting $c$ to $p _l $. The parameter $m$ is equal to $ \phi _N (c) $, while $\lambda$ is equal to the length   of the segment $ \overline{ O  \hat p _l } $, with $\hat p _l =\phi _{ \tilde c }(p _l )$.}
\label{geo}
\end{center}
\end{figure}

For $p\in S ^4 _R $, denote by $\tilde p $ the antipodal point and by $\phi _p: S ^4 _R \rightarrow \mathbb{H}  $ the stereographic projection from $p$. We write $N$ for the north pole $ (0,0,0,0,R) $. 
\begin{prop} 
The centre $c$ and scale $l$ of an instanton on $S ^4 _R $ are related to the parameters $m$ and $\lambda$  by  the relations
\begin{equation}
\label{geomint} 
 m =\phi _N (c), \qquad 
 \lambda = |\phi _{ \tilde c }(p _l  )| ,
\end{equation} 
where $ p _l $  is any point on $ S ^4 _R $ at distance $l$  from  $c$. Equivalently
\begin{equation}
\label{cott} 
\frac{ \lambda}{R} = \tan \left( \frac{l}{2R} \right).
 \end{equation}
 \end{prop} 
\begin{proof} 
Equation (\ref{geomint})  can be easily proved for $c$ the south pole and $ l  \leq R \pi /2 $. Then $ m = 0 $, $ \lambda \leq R $ and using (\ref{al}) we have
\begin{equation}
g _{m}\,  a _{ \lambda  } \cdot \Ar =\Al
= \frac{\Im( \bar q \, \mathrm{d} q)}{\lambda ^2 + |q| ^2 }.
\end{equation} 
Thus in this case $m$ and $\lambda$ coincide with the centre and scale of an instanton on $\mathbb{H}  $. Since stereographic projection is a conformal transformation and the action density $ - \mathrm{Tr} (F \wedge * F )$ is conformally invariant, the stated properties of $c $ and $l$ follows from those of their stereographic projections.
The result generalises to any other value of $c$ and $l$ as the action density  is  invariant under the $Sp (2) $ action.
Equation (\ref{cott}) follows immediately from Figure \ref{geo} and the fact that $l /(2R) $ is the angle $\angle c \, \tilde c\,  p _l $.
\end{proof} 

Unless $m =0 $ and $0<\lambda \leq R$, the (stereographic projection of the) centre and scale on $ S ^4 _R  $ of an instanton are different from the centre and scale on $\mathbb{H}  $ of the same instanton. This is apparent since, unless $c $ is the south pole, stereographic projection of a ball on $S ^4 _R $ with centre $c$  will not be a ball in $\mathbb{H}  $. For example, the projection of any hemisphere not opposite to the north pole is a half space. 

We can now see the geometry behind the condition $ \lambda \leq R $. Suppose an instanton had centre $c$ and scale $l> R \pi /2$.
The  volume of a ball  on $S ^4 _R $ with centre $c$ and radius $l$ is greater than half the surface area of $S ^4 _R $.
The smaller complementary ball, centred at the antipodal point $\tilde c $ and with radius $R \pi - l $, also contains half of the action,
hence the true centre of the instanton is $\tilde c $ and its true scale is $ R \pi - l $. In terms of projected coordinates, if 
$\phi _N  (c) =m$, $|\phi _{ \tilde p } (p _l )| = \lambda $ then $ \phi _N  (\tilde c ) = - R ^2 m/ |m |^2  $, $|\phi _{c} (p _l ) |=R ^2  / \lambda $.
  
To relate the two parameterisations we consider the isometry $\psi: (0, \infty )\times \mathbb{H}  \rightarrow \mathring {B} ^5 _R $,
\begin{equation}
\label{nmhj0} 
\psi (\nu ,n ) = \frac{R }{(R + \nu )^2 + |n |^2 } (\nu ^2 + |n |^2 - R ^2 , 2 R n ),
\end{equation} 
 mapping the half space model of hyperbolic 5-space with sectional curvature $ - 1/R ^2 $ to the open ball model. We denote by 
\begin{equation}
\label{nmhj1} 
\rho ( \nu  , n )  =R \, \sqrt{ \frac{(R - \nu )^2 + |n |^2}{(R  + \nu )^2 + |n |^2} }
\end{equation} 
 the norm of the point $\psi (\nu ,n ) \in\mathring {B} ^5 _R $ with respect to the Euclidean metric on $\mathbb{R}  ^5 $.
  \begin{lem} 
The  parameterisations (\ref{nava}) and (\ref{opppp}) are related by the equations
  \begin{align}
\label{nmn1} 
m( \nu ,n )&
= R \, \phi _N ^1  \left( \frac{\psi (\nu , n )}{\rho (\nu , n )} \right)
=\frac{2 R ^2  \, n}{\sqrt{(R+ \nu ) ^2 + |n| ^2  } \sqrt{(R- \nu ) ^2 + |n| ^2  } + R ^2 - |n| ^2 - \nu ^2 }, \\ 
\label{nmn2} 
\lambda ( \nu ,n ) &
=R \, \frac{R- \rho ( \nu ,n )}{R + \rho( \nu ,n ) }
=\frac{2 R ^2  \, \nu }{\sqrt{(R+ \nu ) ^2 + |n| ^2  } \sqrt{(R- \nu ) ^2 + |n| ^2  } + R ^2 +  |n| ^2 + \nu ^2 },
\end{align} 
with $\phi ^1 _N $  the stereographic projection from the north pole of the \emph{unit} 4-sphere.
\end{lem} 
\begin{proof} 
 Let $g\in SL(2, \mathbb{H}  )$, then for some 
  $ a _{\lambda}\in \AAA, a _{ \nu }\in A $, $b_{n}\in  N $, $u, u ^\prime \in Sp (2) $ we have
  \begin{equation}
  \begin{split} 
  g  &
  =\sqrt{ b _{n} \, a _{\nu} \, (b _{n} \, a _{\nu}) ^\dagger } \, u
  =b _{n} \, a _{\nu} \, u
  =g _{m} \, a _{\lambda } \, u ^\prime 
  =(g _{m} \, a _{\lambda } \, g _{m} ^\dagger) \,  g _{ m}\, u ^\prime.
  \end{split} 
  \end{equation} 
Both $(g _{m} \, a _{\lambda } \, g _{m} ^\dagger) $ and $\sqrt{ b _{n} \, a _{\nu} \, (b _{n} \, a _{\nu}) ^\dagger } $ are Hermitian positive definite matrices, therefore by the uniqueness of polar decomposition we have $u =g _{ m}\, u ^\prime $ (which we are not interested in) and 
\begin{equation}
\label{eq2p} 
 \sqrt{ b _{n} \, a _{\nu} ^2  \, b _{n}  ^\dagger }
 =g _{m} \, a _{\lambda } \, g _{m} ^\dagger.
\end{equation} 
Evaluating both sides of (\ref{eq2p})  and using (\ref{nmhj0}), (\ref{nmhj1}) yields the result.
\end{proof} 
Note how in the $R \rightarrow \infty $ limit $ m \rightarrow n $, $\lambda \rightarrow \nu $ so that the two parameterisations become equivalent.

\section{Circle-invariant instantons}
\label{cinst} 
In this section we define what is meant for an instanton to be circle-invariant and find the  subspace of $\MM _1 $ corresponding to circle-invariant instantons. We also discuss framing and the relation with hyperbolic monopoles.
\label{s2} 
\subsection{The notion of circle invariance}
The group $SO (5) $ naturally acts on $S ^4 _R \subset \mathbb{R}  ^5 $ by  matrix multiplication, and a circle action is obtained by taking an $SO (2) $ subgroup $\UU \subset SO (5) $. The $\mathcal{V} $ action on $S ^4 _R $ induces, through  the diffeomorphism $\hat \pi _H ^{-1}:S ^4 _R  \rightarrow P ^1 (\mathbb{H}  )$, see Appendix \ref{stp},  a $U (1) $ action  on $P ^1 (\mathbb{H}  ) $ by a subgroup $\mathcal{U}\subset  Sp (2) $. The element $  u _{\varphi } \in \mathcal{U} $ corresponding to an element $g _\varphi \in \mathcal{V} $ is  defined up to sign by the requirement
\begin{equation} 
\label{reqreq} 
\phi _N (g _\varphi \, p)= \hat \phi _N ( \check{\rho}(u _{\varphi },\hat \pi _H ^{-1} (p )))  \quad \forall p\in S ^4 _R, 
\end{equation}
 where $\check{\rho}$ is the action \eqref{sl2ha}, and $\phi _N$  and $\hat \phi _N$  are  the stereographic projections from $S ^4 _R$ and   $P ^1(\mathbb{H})$  to $\mathbb{H}$, see again Appendix \ref{stp}.

For definiteness, we consider the following $ SO (2)  $ subgroup of $SO (5) $,
\begin{equation} 
\UU=\left\{ 
\begin{pmatrix}
1 & 0 &0&0 &0\\
0 & 1 &0&0 &0\\
0&0 & \cos  \aalpha & -\sin \aalpha&0\\
0 &0 & \sin  \aalpha &\phantom{-}\cos \aalpha &0\\
0 & 0 &0&0 &1\\
\end{pmatrix} : \aalpha\in[0, 2\pi  ) \right \}.
\end{equation} 
Introducing  global coordinates $(x,y,z,w)\in \mathbb{R}^4 $ on $ \mathbb{H}$  by expanding a quaternion as $q= x + y\qi + z\qj + w \qk$, the transformation induced on $\mathbb{H}  =\phi _N (S ^4 _R \setminus \{(0,0,0,0,R)\} )$  by $g _\varphi \in \UU$  
is a counterclockwise rotation in the $(z,w) $-plane by an angle $ \aalpha$. So, if $p \in S ^4 _R $ and
$q = \phi _N (p) $ then 
\begin{equation} 
\begin{split} 
\label{rotq} 
\phi _N ( g _\varphi \, p ) &= \exp (\qi \aalpha/2) \, q \, \exp ( - \qi \aalpha/2) \\ &
= x + \mathbf{i} y + \mathbf{j}  (z\cos \aalpha   - w\sin \aalpha ) + \mathbf{k} ( w\cos  \aalpha  + z\sin \aalpha ) .
\end{split} 
\end{equation} 
For $p =(b,v) \in S ^4 _R \subset \mathbb{H}  \times \mathbb{R}$,
\begin{equation}
 \hat \phi _N (\hat \pi  _H ^{-1} (p) ) = \frac{R}{R - v} b, \quad 
 \hat \phi _N (\check{\rho}( u_\varphi , \hat \pi  _H ^{-1} (p)) )
= \frac{R}{R - v} \exp (\qi \aalpha/2) \, b \, \exp ( - \qi \aalpha/2) .
\end{equation} 
Therefore  (\ref{reqreq}) gives
\begin{equation} 
\label{liftac} 
 u_{\varphi } = \pm\mathrm{diag} (\exp( \qi \aalpha/2), \exp(\qi \aalpha/2)),
\end{equation} 
so that 
\begin{equation} 
\label{u1group} 
\uu 
= \{\mathrm{diag} (\exp( \qi \aalpha/2), \exp(\qi \aalpha/2)): \aalpha\in[0,  4\pi )\}.
\end{equation} 

We lift the $\mathcal{U} $ action on $P ^1 (\mathbb{H}  )$ to an action on $ S ^7_{\sqrt{ R} } $  by restricting the $SL(2, \mathbb{H}  )$ action $\rho$  given in (\ref{acts7})  to $\mathcal{U} \subset SL(2, \mathbb{H}  )$,
\begin{equation}
\label{cac} 
\left( \exp ( \mathbf{i} \varphi /2 ) , \begin{pmatrix}
 q ^1  \\
q ^2 
\end{pmatrix}  \right) \xmapsto{\rho} \begin{pmatrix}
\exp(\qi \aalpha/2)q ^1\\
\exp(\qi \aalpha/2) q^2
\end{pmatrix} .
\end{equation}
It follows from the explicit matrix  representation of $\qi$ in \eqref{ccccr} that the  weight of the lifted circle action is $1 /2 $.


We call a connection circle invariant if it is strictly invariant under the  action (\ref{cac}),  and
 denote by 
\begin{equation}
\mathcal{A} ^{-} _{ \mathrm{c}} 
= \{ \omega \in \mathcal{A} ^- : u \cdot \omega = \omega  \text{ for all $u \in \mathcal{U} $} \}
\end{equation} 
 the space of circle invariant anti self-dual connections.
 A different, generally weaker, notion of circle invariance is obtained by  requiring a connection to be  invariant up to bundle automorphisms under any lift of the $\mathcal{U}  $ action on $P ^1 (\mathbb{H}  )$ to an action on $ S ^7_{\sqrt{ R} } $. 
However, for the case of 1-instantons on $P ^1 ( \mathbb{H}  )$ it is easy to check (proceeding along the lines of the  proof of Proposition \ref{lem4}) that invariance up to bundle automorphisms reduces to strict invariance.

In the rest of the paper we write $\rho_u$ or  $\rho_\varphi$ for the   action (\ref{cac})  of an element 
\begin{equation}
\label{uconventions}
u =\text{diag}(\exp( \qi \aalpha/2), \exp(\qi \aalpha/2))
\end{equation}
on the total space $S ^7 _{ \sqrt{ R } }$, and $\rrho_u$ or  $\rrho_\varphi$ for the action \eqref{sl2ha}
of the same element on the base $P^1(\mathbb{H})$. We denote the vector fields generating this  circle action on  $S ^7 _{ \sqrt{ R } }$  by $\xi$ and  on  $P^1(\mathbb{H})$ by $\check{\xi}$. 
Acting  on the coordinate functions $q ^1 $, $q ^2 $ on  $S ^7 _{ \sqrt{ R } }$,  we have
\begin{equation}
\label{xiac} 
\xi  ( q ^i ) = \frac{\mathbf{i} }{2} q ^i , \quad i =1,2.
\end{equation} 
Equivalently, with $
q_1=x_1 + y_1\qi + z_1\qj + w_1 \qk$ and  $q_2=x_2 + y_2\qi + z_2\qj + w_2 \qk$,  we obtain the coordinate expression
\begin{equation}
\label{xib}
\xi = \frac{1}{2}\left( x_1\partial_{y_1} - y_1 \partial_{x_1} + z_1\partial_{w_1}-w_1\partial_{z_1}+  x_2\partial_{y_2} - y_2 \partial_{x_2} + z_2\partial_{w_2}-w_2\partial_{z_2} \right),
\end{equation}
while, in terms of the stereographic coordinate $q= x + y\qi +  z\qj + w\qk$ for $P^1(\mathbb{H})$, we deduce from \eqref{rotq} that 
\begin{equation}
\label{xihatb}
\check \xi = z\partial_w-w\partial_z.
\end{equation}

Let us check what circle invariance implies for a gauge potential on $\mathbb{H}  $.
\begin{lem}
Let $\omega$ be a circle-invariant connection, $\sigma$ a global section of $\hat P\simeq \mathbb{H}  \times Sp (1)  $  and $\hat A= \sigma  ^\ast \omega $. Then 
$\hat A $ is invariant up to gauge transformations under the $\uu$ action. 
\end{lem}
\begin{proof} 
As $\omega$ is circle-invariant, for $u\in\uu$,
\begin{equation}
\label{blahh}  
\hat A =\sigma ^\ast \omega = \sigma  ^\ast  \rho _{u ^{-1} } ^\ast \omega .
\end{equation}
Since $\rho _{u ^{-1} } \circ \sigma  $ is not a section of $\hat P$, we rewrite (\ref{blahh})  as 
\begin{equation} 
\hat A = \rrho _{u ^{-1} } ^\ast  \circ ( \rho _{u ^{-1}} \circ \sigma    \circ  \rrho _{u }  ) ^\ast \omega .
\end{equation}
Now $ \rho _{u ^{-1} } \circ \sigma  \circ  \rrho _{u } $ is a section, hence  
\begin{equation}
( \rho _{u ^{-1} } \circ \sigma  \circ  \rrho _{u }  ) ^\ast \omega =h ^{-1} \sigma  ^\ast \omega  h + h ^{-1} \mathrm{d} h 
\end{equation}
for some function $h: \mathbb{H}   \rightarrow Sp (1) $ and 
\begin{equation}
 \rrho  ^\ast _{ u  } \hat A =h ^{-1}  \hat A h + h ^{-1} \mathrm{d} h.
\end{equation} 
\end{proof} 
We will show in Section \ref{mcii} that all circle invariant 1-instantons are of the form $g \cdot \omr $, with $g$ in the group $\mathcal{N} \subset SL(2, \mathbb{H}  )$ defined in Theorem \ref{prop4}. Using (\ref{atrgen}) one can then easily check that for any section $\sigma$ of $\hat P $, $g\in \mathcal{N} $, $\hat A = \sigma ^\ast (g \cdot \omega )$, $u _\varphi = \mathrm{diag} (\exp( \mathbf{i} \varphi /2, \exp( \mathbf{i} \varphi /2 )) $,
 \begin{equation}
 \label{rotrotrot} 
 \rrho  ^\ast _{ u _\varphi  } \hat A 
 =\exp (\qi \, \aalpha /2 )  \hat A \exp (-\qi \, \aalpha /2 ) ,
\end{equation} 
hence $h =\exp( - \mathbf{i} \varphi /2 )$ for all circle invariant 1-instantons.

\subsection{The moduli space of circle invariant 1-instantons}
\label{mcii} 
Denote by $\mathcal{A} ^{-} _{1 \mathrm{c}  } $ the space of circle invariant anti self-dual connections with instanton number one. We say that an automorphism $f\in \mathcal{G} $ preserves circle invariance if $f ^\ast \omega \in \mathcal{A} ^{-} _{1 \mathrm{c}  }  $ for any $\omega \in \mathcal{A} ^{-} _{1 \mathrm{c}  } $.
The moduli space $\MMM  $ of circle invariant 1-instantons is 
\begin{equation}
\MMM  =\mathcal{A} ^{-} _{1\mathrm{c}  }/ \mathcal{G} _ \mathrm{c} ,
\end{equation} 
where $\mathcal{G} _ \mathrm{c} $ the group of bundle automorphisms preserving circle invariance,  which we characterise in Lemma \ref{lem4} below.
The space $\MMM$ can be identified with a subspace of $\MM_{1}  $ since the map
\begin{equation}
[ \omega ] _{ \MMM } \mapsto  [ \omega ] _{ \MM _1}
\end{equation} 
mapping an equivalence class in $\MMM  $ to an equivalence class in $\MM_1 $ is well defined and injective, for if $\omega _1 , \omega _2 \in \mathcal{A} ^- _ \mathrm{c} $,  $f ^\ast \omega _2 =\omega _1  $  then $ f\in\mathcal{G} _{ \mathrm{c}} $.

\begin{prop} 
\label{lem4} 
The group $\mathcal{G} _ \mathrm{c} $  consists of $ \mathcal{U} $-equivariant bundle automorphisms,
\begin{equation}
\mathcal{G} _ \mathrm{c}  
= \{f\in \mathcal{G}  : f \circ  \rho _u =  \rho _u \circ f  \},
\end{equation} 
where $\mathcal{G}$ is the group of bundle automorphisms defined in (\ref{bauto}).
\end{prop} 
\begin{proof} 
For any $\omega\in \mathcal{A} ^- _{ 1 \mathrm{c} }$, $f\in \mathcal{G} $, $f ^\ast \omega $ is circle invariant if and only if, for all  $u\in \mathcal{U} $,
\begin{equation}
\rho _u ^\ast f ^\ast \omega = f ^\ast \rho _u ^\ast \omega.
\end{equation} 
Let $K $ be the intersection of the stabilisers of all the circle invariant 1-instantons. Since $ \mathrm{Stab}(\omr) =Sp (2) $,  $K$ is a subgroup of $Sp (2) $.  For any $p\in P $, $a\in K $,
\begin{equation}
\rho _{ u  } \circ f  \circ \rho _{u ^{-1} } \circ f ^{-1}  (p) = \rho _a (p).
\end{equation} 
For $g _f $ the ad-equivariant map associated to $f$ defined in (\ref{gadeqv}), we can rewrite this as
\begin{equation}
\label{asdfgh} 
\rho _u ( \rho _{u ^{-1}}( p \cdot g _f (p) ^{-1}) \cdot g _f  (\rho _{u ^{-1}}( p \cdot g _f (p) ^{-1}) ) )
= p\cdot  g _f (\rho _{u ^{-1}} (p) )g _f (p) ^{-1} 
= \rho _a( p).
\end{equation} 
By projecting onto the base we see that $a =\pm \mathrm{Id}$, the only elements of $Sp (2) $ which do not move base points. Since the equality $g _f (\rho _{u ^{-1}}( p) ) =\pm g _f (p) $ must hold for all $u\in \mathcal{U} $, it follows $a =\mathrm{Id} $, so that $g _f $ is constant on the $\mathcal{U} $ orbits.  Hence 
\begin{equation}
f(\rho _u (p))=\rho _u ( p \cdot  g_f( \rho _u (p)))=\rho _u( p \cdot g_f(p))=\rho _u (f(p)),
\end{equation} 
which is the claimed equivariance property of $f $.

For  $\omega\in \mathcal{A} ^- _{ 1 \mathrm{c} }$, suppose now that $f\in \mathcal{G} $ is  $\mathcal{U} $-equivariant. Then, for any $u\in \mathcal{U} $,
\begin{equation}
\rho _u ^\ast f ^\ast \omega = (f \circ \rho _u ) ^\ast \omega  =( \rho _u \circ f )^\ast \omega 
 = f ^\ast \rho _u  ^\ast \omega = f ^\ast \omega,
\end{equation} 
hence $f ^\ast \omega $ is circle invariant.
\end{proof} 

Let $f _\epsilon  $, $\epsilon \in \mathbb{R}$ and $f_0=\mathrm{Id} _P$, be a 1-parameter family of  bundle automorphisms,  and $g _{f_\epsilon } $  the ad-equivariant function associated to $f_ \epsilon $ by (\ref{gadeqv}).
Differentiating, we obtain the  ad-equivariant function
\begin{equation}
X_f: 
P \rightarrow  \mathfrak{ g }, \qquad X_f (p)=   \left.\frac{ \mathrm{d} } { \mathrm{d} \epsilon} \right | _{ \epsilon=0 }g_{f_ \epsilon } (p).
\end{equation} 
To any function $X : P \rightarrow \mathfrak{ g } $ is associated  a vertical vector field $ X  ^\sharp $, defined via
\begin{equation}
\label{assvec} 
X  ^\sharp _p = \left. \frac{\mathrm{d} }{\mathrm{d} t} \right | _{ t =0 } p \cdot \exp( t X (p)  ).
\end{equation} 
It then follows from the  ad-equivariance of $X_f$ that  the vector field $X ^\sharp _f $ is  $Sp(1)$ right-invariant.
We call both $ X ^\sharp _f $ and the associated ad-equivariant map $X _f $ infinitesimal bundle automorphisms.

Combining these considerations, and recalling the definition \eqref{xiac} of the generator $\xi$ of the $\uu$ action on $P$, we note:
\begin{coro}
\label{u1coro}
A 1-parameter family  of  bundle automorphism $f_\epsilon $ preserves circle invariance if and only if 
\begin{equation}
\label{bcond} 
[ X _f ^\sharp , \xi ]=0.
\end{equation} 
\end{coro}
\begin{proof}
This follows directly from the observation,  made in the proof of  Proposition \ref{lem4},  that the  ad-equivariant maps $g_{f_\epsilon}$  associated to  circle invariance preserving bundle isomorphisms $f_\epsilon$ are constant on $\uu$ orbits.
\end{proof}


 Since the action of $SL(2, \mathbb{H}  )$ on $\MM _1 $ is transitive with stabiliser $Sp (2) $, it follows that
\begin{equation}
\label{m1cccccc} 
\MM_ {1\mathrm{c} }\simeq \mathcal{N} / \mathcal{N} \cap Sp (2),
\end{equation} 
where $\mathcal{N} $ is the subgroup of $SL(2, \mathbb{H}  )$ acting transitively on $\mathcal{A} ^{-} _{1 \mathrm{c}  } $. We are now going to determine $\mathcal{N}$, but first introduce some notation.

By identifying the unit quaternion $ \mathbf{i} $ with the complex number $i =\sqrt{ -1 }$ we obtain a splitting $\mathbb{H}  = \mathbb{C}  \oplus \mathbb{C}  \mathbf{j} $. We call a quaternion $q$ complex  and write $q\in \mathbb{C}  $ if $q$  is of the form $q = x + y \mathbf{i} $. We write $q\in \mathbb{C}  \mathbf{j} $ if $q =z  \mathbf{j} + w \mathbf{k} $. 


\begin{thm} 
\label{prop4} 
The group 
\begin{equation}
\GG 
=\left\{  \begin{pmatrix} a &b \\ c & d \end{pmatrix} \in SL(2, \mathbb{H}  ) : a,b,c,d\in \mathbb{C}, ad - bc=1  \right\}
\end{equation} 
acts transitively on $\mathcal{A} ^{-} _{1 \mathrm{c}  }  $.
\end{thm} 
\begin{proof} 
For $g\in SL(2, \mathbb{H}  )$, $g \cdot \omr $  is circle invariant if and only if,  for all $u\in \mathcal{U} $,
\begin{equation}
\label{mycond} 
u \cdot g \cdot \omr = g \cdot \omr .
\end{equation}
Let
\begin{equation}
g ^{-1} =\begin{pmatrix}
a &b \\ c &d
\end{pmatrix}\in SL(2, \mathbb{H} ), \qquad 
u ^{-1} =
\begin{pmatrix}
\mathrm{e}^{ \qi \aalpha/2} &0\\ 0 & \mathrm{e}^{ \qi \aalpha/2}
\end{pmatrix}\in \uu.
\end{equation} 
Decompose $a$ into its $\mathbb{C}  $ and $\mathbb{C}  \mathbf{j} $ parts, $a =a _1 + a _2\,  \qj $, $a _1 , a _2 \in \mathbb{C}  $, and similarly for $b,c,d $. Using (\ref{acto}), we find that (\ref{mycond}) implies
\begin{equation}
\label{cinvcond} 
(a _1 =b _1 =0 \text{ or }  a _2 = b_2=0) \quad \text{and} \quad   (c_1 = d_1 =0 \text{ or }    c _2 =d _2 =0).
\end{equation} 
Multiplication by $\qj$ from the right  gives a bijective correspondence between $ \mathbb{C}  $ and $\mathbb{C}  \qj $, and all circle invariant instantons can be obtained by taking $a,b,c,d\in \mathbb{C}  $. In fact, let e.g.~$ \tilde a = a \, \qj $, $\tilde b =b\,  \qj $ with $a, b\in \mathbb{C}  $. Since
$
|\tilde a q ^1  + \tilde b q ^2  |
= |\bar aq ^1  + \bar b q ^2  |
$,
$\overline{\tilde b} \tilde a
=b\bar a$,
it follows from (\ref{acto}) that taking $a ,b \in \mathbb{C}  \mathbf{j}$  results in the same instanton as  taking $a,b \in\mathbb{C} $.
Therefore the group
\begin{equation}
\mathcal{S} =\left\{  \begin{pmatrix} a &b \\ c & d \end{pmatrix}  : a,b,c,d\in \mathbb{C}, |ad-bc|=1  \right\},
\end{equation} 
where the condition $ |ad - b c | =1$ follows from $\mathcal{S} \subset SL(2, \mathbb{H}  )$,
acts transitively on $\mathcal{A} ^{-} _{1 \mathrm{c}  } $.
If $g\in \mathcal{S} $, by factoring out a phase $u\in \mathcal{U} $  we can write $g =h u $,
with $h\in \mathcal{S} $,  $ h _{ 11 } h _{ 22 } -h _{ 12 }h _{ 21 } =1 $. Hence $\mathcal{S} =\mathcal{N} \, \mathcal{U} $.
By definition, circle-invariant instantons are invariant under the action of $ \uu $, hence $ \mathcal{N}  $ also acts transitively on $\mathcal{A} ^{-} _{1 \mathrm{c}  } $.
\end{proof} 

\begin{coro} 
The moduli space $ \MMM$ of circle invariant $Sp (1) $ instantons on $P ^1 (\mathbb{H}  )$ with instanton number one is 
 diffeomorphic to the quotient
 \begin{equation}
 \MMM \simeq  SL(2, \mathbb{C}  )/ SU (2) .
 \end{equation} 
\end{coro} 
\begin{proof} 
The group $\mathcal{N}$ is clearly isomorphic to $SL(2, \mathbb{C}  )$. Also
\begin{equation}
\begin{split} 
\label{wash} 
\mathcal{N} \cap Sp (2)= \Big\{&
\begin{pmatrix} a &b\\ c &d\end{pmatrix} :
a,b,c,d\in \mathbb{C}, \, |a |^2 + |c| ^2 = |b| ^2 + |d |^2 =1,\\ & \bar a b + \bar c d=0, \, ad-bc=1
\Big\},
\end{split} 
\end{equation} 
which is isomorphic to $SU (2) $. The result then follows from (\ref{m1cccccc})  and Theorem \ref{prop4}. 

\end{proof}

Proceeding similarly as in Section \ref{ss2}, we now obtain two useful parameterisations of $\MMM $ in terms of different decompositions of $SL(2, \mathbb{C}  )$.
The analogue of (\ref{cdec}) is
\begin{equation} 
{SL}(2, \mathbb{C}  ) = {SU}(2) \, {\AAA} \,  {SU}(2) =\bigcup _{ m\in\hat{ \mathbb{C}  }} {g _m} \,  {\AAA}\,  {SU}(2),
\end{equation}
 with $\hat { \mathbb{C}   }=\mathbb{C}\cup\{ \infty \}  $, 
$
\AAA = \{ \mathrm{diag}(\sqrt{ \lambda  /R}, \sqrt{ R/ \lambda  }): \lambda \in(0, R ] \}
$, 
 \begin{equation}
\label{gm2} 
g _m =\begin{cases}
\frac{1}{\sqrt{ R ^2  + |m |^2 }}\left( 
\begin{array}{ccc} R & m \\ - \bar m &R\end{array} 
\right) &\text{if $m\in \mathbb{C}  $},\\
\left( \begin{array}{ccc} 0 & 1 \\ - 1 &0\end{array} 
\right) &\text{if $m =\infty   $}.
\end{cases} 
\end{equation} 
It follows
\begin{equation}
\MMM 
\simeq   \GG/\HH
\simeq    \bigcup _{ m\in\hat{ \mathbb{C}  }} g _m  \AAA.
\end{equation}
The moduli space of circle-invariant 1-instantons is therefore diffeomorphic to the open ball $ \hat{\mathbb{C}} \times (0, R]  \sim S ^2 \times (0, R ] =\mathring B ^3 _R$ parameterised by $m$ and $\lambda$. The boundary $\lambda =0 $  is not included in the moduli space and corresponds to singular ``small'' instantons.

For future convenience, we introduce a different parameterisation of $g _m $. Write
\begin{equation}
\label{difpar} 
\frac{R}{\sqrt{ R ^2  + |m | ^2 }}=\cos (\alpha /2 ),  \quad 
\frac{m}{ \sqrt{ R ^2  + |m |^2 }} =  \sin (\alpha /2 ) \exp( \qi \beta  ),
\end{equation} 
 with $\alpha\in[0, \pi ]$, $\beta \in[0, 2\pi  )$, so that
\begin{equation}
g _{m} =
\begin{pmatrix}
\cos (\alpha /2)  & \sin (\alpha /2) \mathrm{e} ^{ \qi \beta }\\
-\sin (\alpha /2) \mathrm{e} ^{ - \qi \beta  } &\cos (\alpha /2)
\end{pmatrix}
=  g_{ \alpha , \beta }.
\end{equation} 
The angles $\alpha , \beta $ parameterise a 2-sphere in the moduli space. Apart from the smaller dimension, the moduli space structure is entirely similar to that of $\MM _1 $ and, with the obvious changes, Figure \ref{mods} still applies. Note that  $ R - \lambda $, rather than $\lambda$, is the radial coordinate on $\mathring B ^3 _R $. We  group the parameters $( \lambda , \alpha , \beta )$ into the vector 
\begin{equation}
\label{veclambda} 
 \boldsymbol{ \lambda } 
 =(R - \lambda) ( \sin  \alpha \cos \beta , \sin \alpha \sin \beta , \cos \alpha  ).
 \end{equation}

Making use  of the Iwasawa decomposition of $SL(2, \mathbb{C}  ) $,
\begin{equation}
{SL}(2, \mathbb{C}  ) ={N} \,{ A} \, {SU} (2), 
\end{equation} 
with $A =\{ \mathrm{diag}(\sqrt{ \nu /R}, \sqrt{ R/ \nu }): \nu >0 \}$, 
\begin{equation} 
N 
=\left \{   \begin{pmatrix}
1 &n/R\\ 0 &1
\end{pmatrix}
: n\in \mathbb{C}  
\right\},
\end{equation} 
we get the alternative description of the moduli space
\begin{equation}
\label{kfdhjg} 
\MMM 
\simeq   N \, A,
\end{equation} 
which is diffeomorphic to the upper half space $\{(n, \nu )\in \mathbb{C}  \times \mathbb{R}  : \nu >0 \}  \subset \mathbb{R}  ^3 $.

Let us pause to interpret these results. Consider a point  $( \lambda , m )\in\mathring B ^5 _R  $ of the moduli space of  (not necessarily circle-invariant) instantons. By computing 
\begin{equation}
u g _m a _\lambda  \cdot \omr,
\end{equation} 
with $u= \mathrm{diag}(\exp(\qi \aalpha/2), \exp(\qi \aalpha/2)) \in\uu $, one can check that the effect of the circle action on $( \lambda , m )$ is, for $m =m _1 + m _2 \, \mathbf{j} $, $m _1, m _2 \in \mathbb{C}  $,
\begin{equation}
\lambda \rightarrow \lambda , \qquad 
m 
\rightarrow \exp(\qi \aalpha/2) m \exp( - \qi \aalpha/2) 
=m _1 + \exp(\qi \aalpha) m _2.
\end{equation} 
It is therefore clear that, for our choice of the circle action, an instanton can be circle-invariant if and only if $m\in \mathbb{C}  $.
With respect to the foliation of $\mathring B ^5 _R \setminus\{0\} $ by 4-spheres, this condition on $m$ selects for each 4-sphere the equatorial 2-sphere which stereographically projects to  $\mathbb{C}  \subset\mathbb{H} =\mathbb{C}  \oplus \mathbb{C}  \mathbf{j}  $. For this reason we obtain  $\MMM  \simeq  \mathring B ^3 _R $ as a subspace of the full moduli space $\MM _1  \simeq   \mathring B ^5 _R $.

The full symmetry group $SL(2, \mathbb{H}  )$  of the Yang-Mills action  acts on $\MM _1  $ (with only $SL(2, \mathbb{H}  ) / \mathbb{Z}  _2  \simeq  SO(5,1) $ acting effectively). If  $\mathbb{H}  $ is the stereographic projection of any of the  4-spheres foliating $\MM _1 $, the circle action  induces a splitting $\mathbb{H}  = \mathbb{C}  \oplus \mathbb{C}  \qj$ and only the subgroup of $SL(2, \mathbb{H}  ) $ preserving this splitting  acts on $\MMM $.
This is (after factoring out  $\uu $ which fixes circle invariant instantons) the subgroup $\GG  \simeq  SL(2, \mathbb{C}  )$ of $SL(2, \mathbb{H}  ) $.  In other words, given a subgroup $SO (2) $ of $ SO(5,1)  \simeq  SL(2, \mathbb{H}  )/ \mathbb{Z}  _2$, invariance with respect to the $SO (2) $ action selects the $SO(3,1)  \simeq  SL(2, \mathbb{C}  )/ \mathbb{Z}  _2  $ subgroup of $SO (5,1) $ commuting with the given $SO (2) $.

Similar remarks can be made for the  moduli space parameterisation (\ref{kfdhjg}). The
 effect of the circle action on a point  $( \nu  , n )\in \mathbb{H}  \times (0, \infty )$ of $\MM _1 $ is
\begin{equation}
\nu  \rightarrow \nu  , \qquad 
n \rightarrow  \exp(\qi \aalpha/2) n \exp( - \qi \aalpha/2) ,
\end{equation} 
so that circle invariance forces $n$ to lie in $\mathbb{C}  $. For an instanton on $\mathbb{H}  $ this is the intuitive condition that the centre of a circle-invariant instanton cannot have any component in the plane which is being rotated. 

\subsection{Framing and monopoles}
\label{monopoles} 
It is  convenient to enlarge the moduli space $\MM _1 $ of $Sp (1) $ Yang-Mills 1-instantons on $P ^1 (\mathbb{H}  )$ by choosing a framing point $x _0 \in P ^1 (\mathbb{H}  )$  and defining the framed moduli space 
\begin{equation}
\MMf _1   =\mathcal{A}_1  ^- / \mathcal{G} _0  ,
\end{equation}
where  $\mathcal{A} _1 ^- $ is the space of anti self-dual 1-instantons and
\begin{equation}
\mathcal{G} _0 =\{ f\in \mathcal{G} : f |_{ \pi ^{-1} (x_0) } = \mathrm{Id}_{P }   \}.
\end{equation} 
The framed moduli space  is fibered over $\MM _1 $ with fibre the structure group modded out by its centre.\footnote{The reason for quotienting out the centre is the fact  that the group of bundle automorphisms always has a subgroup, isomorphic to the centre of the structure group, which stabilises every connection.}
Since $\MM _1 $ is 5-dimensional, $\MMf_1 $ is  8-dimensional with $Sp (1) / \mathbb{Z}  _2\simeq SO (3) $ fibres.  

For circle invariant instantons, we require $x _0 $ to be fixed by the circle action. We shall take $x _0 =[R,0]  $ as our framing point. The framed moduli space of circle invariant 1-instantons is
\begin{equation} 
\MMM
= \mathcal{A}  ^{ - } _{ 1\mathrm{c}} / \mathcal{G} _{ 0c }, 
\end{equation} 
where $\mathcal{G} _{ 0\mathrm{c}  }= \mathcal{G} _0 \cap \mathcal{G} _\mathrm{c} $.
Because of the condition (\ref{bcond}), the framed moduli space $\MMMf $  is fibered over $\MMM $ with $U (1) $, rather than $SO (3) $, fibres and is diffeomorphic to the trivial $U (1) $ bundle over the open 3-ball $\mathring B ^3 _R $,
\begin{equation}
\MMMf
\simeq   \MMM \times U (1)
\simeq   \mathring B ^3 _R \times U (1 ) .
\end{equation}


We will now give an explicit description of the $U(1)$ factor in the framed moduli space of circle invariant instantons. In particular, we will show how to obtain it directly from the  generator $\xi$ \eqref{xib} of the $U(1)$ subgroup $\uu$ which defines our notion of circle invariance. 
 The idea is to view  the vertical component $\xi^\mathrm{v}$  of  $\xi$ as an infinitesimal bundle automorphism, and to show that the bundle automorphism it generates preserves circle invariance and is non-trivial on the fibre over the framing point $x_0$.

 By definition,  the vertical and horizontal components of $\xi$  with respect to a given circle invariant connection $\omega$  satisfy
 \begin{equation}
 \xi= \xi^\mathrm{v} + \xi^\mathrm{h}, \qquad  \omega(\xi^\mathrm{h})=0, \quad \xi^\mathrm{v}  \text{ vertical}. 
 \end{equation}
The vertical component can be calculated according to  
\begin{equation}
\xi^\mathrm{v}= \omega (\xi ) ^\sharp,
\end{equation} 
 with $^\sharp $  defined in (\ref{assvec}).

In order to study the properties of $\xi^\mathrm{v}$,  we  define the Higgs field $\hi: P \rightarrow \mathfrak{ sp }(1)$  via
\begin{equation}
\label{fffphi} 
\hi = \omega (\xi).
\end{equation} 
Since $ \xi $ is left generated, it is $Sp(1)$  right invariant. Thus, for  any $p\in P $, $g\in Sp (1) $,
\begin{equation} 
\hi (p \cdot g ) 
= \omega _{ p \cdot g } ( \xi _{ p \cdot g } )
= \omega _{ p \cdot g } (R _{ g * } \xi _p ) 
= \mathrm{ad} _{ g ^{-1} } \omega _p (\xi _p )
= \mathrm{ad} _{ g ^{-1} } \hi (p).
\end{equation} 
Hence $\hi$ is ad-equivariant, 
 and so is   $\exp( \hi \, \varphi):P \rightarrow Sp (1) $ for any $ \varphi \in[0, 4 \pi )$. Therefore, the Higgs field defines a 1-parameter family of bundle automorphisms $ \{ \Phi _\varphi : \varphi \in[0, 4 \pi )\}$, with $\Phi _\varphi $  the bundle automorphism associated to $\exp(\hi  \, \varphi )$ via (\ref{gadeqv}).
 
 \begin{lem}
\label{lemlem} 
For any $\varphi\in[0, 4 \pi )$ the bundle automorphism $\hi _\varphi =\exp( \hi \, \varphi ) $  preserves circle invariance.  
\end{lem}
\begin{proof} 
For $p\in P $,  $ \Phi _\varphi (p) =  p \cdot \exp( \hi (p)  \, \varphi )=p \cdot \exp( \omega (\xi _p )  \, \varphi )$, so that the infinitesimal generator of the flow  $\Phi _\varphi $ is  $\xi^\mathrm{v}= \omega (\xi ) ^\sharp $.
By (\ref{bcond}),  $\Phi _\varphi $ preserves circle invariance if and only if $\mathcal{L} _\xi (\omega (\xi) ^\sharp )=0 $. 
Write
$
\omega (\xi)  = \omega ( \xi ) _{ \mathbf{i} } \,  \mathbf{i} + \omega ( \xi ) _{ \mathbf{j} } \, \mathbf{j}  + \omega ( \xi ) _{ \mathbf{k} } \, \mathbf{k},
$
so that
\begin{equation}
\mathcal{L} _\xi (\omega (\xi) ^\sharp )
= \mathcal{L} _\xi ( \omega (\xi) _{ \mathbf{i} } \, \mathbf{i} ^\sharp)+
 \mathcal{L} _\xi ( \omega (\xi) _{ \mathbf{j} } \, \mathbf{j} ^\sharp)+ 
  \mathcal{L} _\xi ( \omega (\xi) _{ \mathbf{k} } \, \mathbf{k} ^\sharp).
\end{equation} 
Since $\omega$ is circle invariant, $\mathcal{L} _\xi \omega =0 $, hence 
\begin{equation}
\label{liefider} 
 \mathcal{L} _\xi (\omega ( \xi )) = \iota _\xi ( \mathcal{L} _\xi  - \iota _\xi \mathrm{d} ) \omega = 0 .
 \end{equation}  
 Since $ \xi $ is left generated while $\mathbf{i} ^\sharp $, $ \mathbf{j} ^\sharp $, $ \mathbf{k} ^\sharp $ are right generated, $ [ \xi , \mathbf{i} ^\sharp ] =[ \xi , \mathbf{j} ^\sharp ] =[ \xi , \mathbf{k} ^\sharp ] =0 $. Therefore $[ \xi, \omega (\xi) ^\sharp ]=0 $ so that, by  Corollary  \ref{u1coro},  $\Phi _\varphi $ is an element of $\mathcal{G} _ \mathrm{c} $.
 \end{proof} 
 
Next we would like to show that, for  $\varphi\neq 0 $, $\hi _\varphi $ is a non-trivial element of $ \mathcal{G} _ \mathrm{c } $, that is, does not vanish on the fibre $ \pi ^{-1} ([R,0] ) $. In fact, we will prove the stronger property that $|\hi| $ is constant with value $1/2 $ on  $P |_ \mathcal{S} $, where $\mathcal{S}$  is  the fixed points set of $\xi$. 
 Let us first characterise  $\mathcal{S}$. A point $[q ^1 , q ^2 ] \in \mathcal{S} $ satisfies 
 \begin{equation}
 [\exp( \mathbf{i} \varphi /2 )q ^1 , \exp( \mathbf{i} \varphi /2 )q ^2] =[ q ^1 , q ^2 ],
 \end{equation}
and stereographic projection gives the condition 
\begin{equation} 
\exp( \mathbf{i} \varphi  /2 )Rq ^1 (q ^2) ^{-1}  \exp(- \mathbf{i} \varphi  /2 ) = Rq ^1 ( q ^2 ) ^{-1} ,
\end{equation}
which in turn implies 
\begin{equation}
\label{fp1} 
q ^1 ( q ^2) ^{-1}  \in \mathbb{C} .
\end{equation} 
Therefore, $ \mathcal{S} $ is the 2-sphere 
 \begin{equation}
 \label{ss2234} 
 \mathcal{S}  =\{[q ^1 ,q ^2 ]: Rq ^1  (q ^2) ^{-1} \in \mathbb{C}   \} \cup \{[R,0] \} .
 \end{equation} 
 Writing $ q ^1 = r _1 h _1 $, $ q ^2 =r _2 h _2 $, $r _1 , r _2 \in \mathbb{R}  $, $h _1 , h _2 \in \mathbb{H}  $, $ |h _1 |=|h _2 |=1 $, (\ref{fp1})  becomes $ r _1 r _2 h _1 \bar h _2 \in \mathbb{C}  $, hence $h _1 = z h _2 $ for some $z\in \mathbb{C}  $, $|z |=1 $. Therefore 
 \begin{equation}
 \label{nmmaskdjapsdij} 
P |_{ \mathcal{S} }= \{( \lambda   _1 h, \lambda   _2  h )  : (\lambda  _1  , \lambda   _2) \in S ^3 \subset \mathbb{C}^2, h\in \mathbb{H},   |h |=1 \}/ \sim,
 \end{equation} 
the equivalence relation being left multiplication by a unit complex number.

\begin{prop} 
\label{prop8} 
The Higgs field $\hi$ has constant norm $1/2 $ on $ P | _ \mathcal{S} $. In particular, for  any $\varphi \in(0, 4 \pi )$ the bundle automorphism $\hi _\varphi =\exp( \hi \, \varphi )$ determines  a non-trivial element of $\mathcal{G} _ \mathrm{c} / \mathcal{G} _{ 0 \mathrm{c} } $.
\end{prop} 
\begin{proof}

Let $p\in P $ be such that $\pi (p)\in \mathcal{S}  $. Then $\xi _p $ is purely vertical, 
\begin{equation}
\label{lcond} 
\xi _p = \left. \frac{\mathrm{d} }{\mathrm{d} t} \right | _{ \varphi  =0 } p \cdot \exp ( \varphi  L  _p )
\end{equation}  
for some $L _p \in \mathfrak{ sp }(1) $. The value of the Higgs field at $p$ is then 
\begin{equation}
\hi (p) =\omega _p (\xi _p ) =L _p.
\end{equation}
We now need to determine $L _p $.
The infinitesimal circle action on  a point $p =( q ^1 , q ^2 )\in P $ is  obtained by acting with $\xi$, hence by (\ref{xiac})
\begin{equation}
\label{lkjyuhj} 
( q ^1 , q ^2 ) \mapsto \xi _{ (q ^1 , q ^2 )} (q ^1 , q ^2 )= \frac{\mathbf{i} }{2}( q ^1 , q ^2 ).
\end{equation} 
For $( q ^1 , q ^2 )\in \pi ^{-1} (\mathcal{S} )$, $\xi  _{ ( q ^1 , q ^2 )} $ is purely vertical, hence  the left action of $\xi$ in (\ref{lkjyuhj})  has to be equal to the right action of $ L _{( q ^1 , q ^2 )} $, 
\begin{equation}
\frac{\mathbf{i} }{2} ( q ^1 , q ^2 ) = ( q ^1 , q ^2 )L _{(q ^1 , q ^2 )}.
\end{equation} 
Because of (\ref{nmmaskdjapsdij}), 
\begin{equation}
\frac{\mathbf{i} }{2} q ^i =\frac{\mathbf{i} }{2} \lambda _i h =q ^i  h^{-1} \frac{\mathbf{i} }{2} h, \qquad i =1,  2,
\end{equation} 
so that
\begin{equation}
L _{( q ^1 , q ^2 )} =
h ^{-1} (\mathbf{i}/2) h,
\end{equation} 
and $| \hi (q ^1  , q ^2 )  |=1/2 $.
\end{proof}

In order to link our discussion to  hyperbolic monopoles, we need to  pull back $\Phi$ to a circle invariant field $\ph: P ^1 (\mathbb{H}  )\setminus \mathcal{S} \rightarrow \mathfrak{ sp} (1) $. 
Since $\Phi$ is $\uu$ invariant, pulling back via a local section  $\sigma:  W\subset  P ^1 (\mathbb{H}  )\setminus \mathcal{S}\rightarrow P$ which intertwines the circle actions,
\begin{equation}
\label{sequi} 
\rho _u \circ  \sigma = \sigma \circ \rrho_u, \quad \forall u\in \mathcal{U} 
\end{equation} 
would result in a circle-invariant function on $W\subset  P ^1 (\mathbb{H}  )\setminus \mathcal{S}$. However, it is easy to see that there can be no local section satisfying this requirement, since (with the notation introduced after \eqref{uconventions})  $\rrho_{2\pi}= \mathrm{Id}_{P^1(\mathbb{H}) }$ but 
$\rho_{2\pi} =- \mathrm{Id}_{P } $. However,   such a section does exist for the   $SO (3) $ principal bundle $P / \mathbb{Z}  _2 $, where $\mathbb{Z}  _2 $ acts on $P$ as in (\ref{acts7}), as we will show by  construction in (\ref{sigmasec}) below.  This is sufficient for our purposes since   both $\omega$ and $ \xi $ are invariant under the $\mathbb{Z}  _2 $ action and descend to $P / \mathbb{Z}  _2 $. We sum up the situation as follows.


\begin{lem}
\label{lemsec}
Let  $\omega$ be   a circle invariant connection on $P $, $\hi=\omega(\xi)$ the associated Higgs field and  $\sigma$ a local section of $P / \mathbb{Z}  _2 $ which satisfies \eqref{sequi}. Then 
\begin{equation}
\label{inin}
   \mathcal{L} _{ \check \xi }  ( \sigma ^\ast \hi )  = 0 \quad \text{and}  \quad  \mathcal{L} _{ \check \xi }  ( \sigma ^\ast \omega   ) =0.
   \end{equation}
      Moreover
\begin{equation}
\label{nice} 
\sigma ^\ast \hi =(\sigma ^\ast \omega ) ( \check \xi ).
\end{equation} 
\end{lem}
\begin{proof}
The  equations \eqref{inin} are simply the infinitesimal versions of the circle  invariance of the pull-backs of the connection $\omega$ and the Higgs field. Both follow directly from the intertwining properties of the local section and  the circle  invariance of $\omega$ and $\hi$. 
The relation \eqref{nice}  follows from 
\begin{equation}
(\sigma ^\ast \omega )( \check \xi )
=\mathrm{d} / \mathrm{d} \varphi   |_{ \varphi  =0 } \, \omega ( \sigma ( \rrho _\varphi ))
=\mathrm{d} / \mathrm{d} \varphi  |_{ \varphi   =0 } \, \omega ( \rho _\varphi   ( \sigma  ))
=\omega ( \xi ) \circ \sigma 
 =\sigma ^\ast \hi.
\end{equation} 
\end{proof}

Let us finally come to the relation with hyperbolic monopoles \cite{Atiyah:1987ua}.
There is a conformal isometry $ P ^1 (\mathbb{H}  )  \setminus S ^2    \simeq  H ^3 _R \times S ^1 _R  $, where $H ^3 _R  $ is the hyperbolic 3-space with metric $g _{ H ^3 _R }$ of sectional curvature $-1/ R ^2 $.  In coordinates
\begin{equation}
\begin{split} 
\gph &
= \frac{R ^2 }{(R ^2 + |q| ^2 )^2 } ( \mathrm{d} x ^2 + \mathrm{d} y ^2 + \mathrm{d} z ^2 +  \mathrm{d} w ^2 )
=\frac{R ^2  \rho ^2 }{(R ^2 + |q| ^2  ) ^2 } \left(g _{ H ^3 _R } +  g _{ S ^1 _R   } \right), \\
g _{ H ^3 _R } &
= \frac{R ^2 }{ \rho ^2 }( \mathrm{d} x ^2 + \mathrm{d} y ^2 + \mathrm{d} \rho ^2 ), \qquad g _{ S ^1 _R }= R ^2 \,  \mathrm{d} \cchi ^2.
\end{split} 
\end{equation} 
Here $x,y, z,w $ are coordinates on $\mathbb{H}  $, $ |q| ^2 =x ^2 + y ^2 + z ^2 + w ^2 $ and $(\rho, \vartheta  )$ are polar coordinates on the $(z,w) $-plane. In these coordinates the vector field $\check \xi $ \eqref{xihatb} is simply $\partial / \partial \vartheta $.

The removed  2-sphere  is the set  $\mathcal{S}$ of fixed  points of the $\rrho$ action  of $\UU$, and is characterised by (\ref{ss2234}). It corresponds to the  $(x,y) $-plane $\rho =0 $ and its point at infinity $[R,0] $, and is mapped by the conformal isometry to the asymptotic boundary of $H ^3 _R  $. More precisely, $ \mathcal{S}  \setminus [R,0] $ is mapped to the plane $\rho =0 $ in $H ^3_R $, and $[R,0] $ corresponds to all the points at infinity on $H ^3 _R $, both those on the plane $\rho  = 0 $ and those having $\rho >0 $.

Let $\omega$ be a circle invariant connection satisfying the ASDYM equations on $ P ^1 (\mathbb{H}  ) \setminus   \mathcal{S}   $, and $\sigma$ a local section of $P / \mathbb{Z}  _2 $ satisfying (\ref{sequi}).  We introduce the notation
\begin{equation}
R \ph = \sigma ^\ast \hi, \qquad 
\Aa 
= \sigma ^\ast \omega  -R \ph \, \mathrm{d} \vartheta.
\end{equation}
By Lemma \ref{lemsec}, $R  \ph = (\sigma ^\ast \omega ) ( \partial / \partial \vartheta  )$, hence 
\begin{equation}
\sigma ^\ast \omega 
= \Aa + \ph R\,   \mathrm{d} \vartheta, \qquad  \iota _{ \partial / \partial \vartheta } \Aa = 0 .
\end{equation} 
Because of the conformal invariance of the ASDYM equations,   $\Aa $, $\ph$ satisfy the Bogomolny equations for a hyperbolic monopole,
\begin{equation}
\label{mon} 
\mathrm{d} _{ \Aa } \ph = -  \,     * _{ H ^3 _R  } \Fff,
\end{equation}  
with $\Fff $ the curvature of $\Aa $ and  $* _{H^3 _R } $  the Hodge operator with respect to the hyperbolic metric. 
Therefore, a circle-invariant instanton on $P ^1 (\mathbb{H}  ) $ can be re-interpreted as a monopole on $H ^3 _R  $.

In general the instanton number $I $, the monopole number $k$ and the weight of the circle action $\mu $ are related by the equation $ I =2k\mu  $ \cite{Atiyah:1987ua}. For our choice of the lift $\mu =1/2 $  so that $I =k =1$.


Note that above a fixed point (\ref{bcond}) becomes $[ X _f, \Phi ] = 0 $.
Therefore, the fact that the fibre of the framed moduli space $\MMMf $ is $U (1) $ instead of $Sp (1) / \mathbb{Z}   _2 $  is particularly natural from the monopole perspective: allowed infinitesimal bundle automorphisms must commute with the asymptotic value of the Higgs field.

For future reference, we would like to write down the hyperbolic monopole associated to the circle invariant  connection  $ a _\lambda \cdot \omr$ with  moduli $ \alpha = \beta =0 $ but generic $\lambda$ dependence. As expected, the section 
$s$ of the pullback bundle $\hat P$  given in \eqref{ours} does not  satisfy the intertwining condition  \eqref{sequi}.  However, with  $\vartheta $ again denoting  the polar coordinate of $q$ in the $(z,w) $-plane, the  section $\sigma$ of $\hat P/\mathbb{Z}_2$ given by 
\begin{equation}
\label{sigmasec} 
\begin{split} 
  \sigma(q) &
  = s(q)\cdot  \mathrm{e} ^{\frac{\vartheta}{2} \qi}
  = \mathrm{e} ^{\frac{\vartheta}{2} \qi}\sqrt{\frac{R}{|q|^2+R^2}}(\mathrm{e} ^{-\frac{\vartheta}{2} \qi}\, q \, \mathrm{e} ^{\frac{\vartheta}{2} \qi},R)
\end{split} 
\end{equation}
 does. 
For the connection $a _\lambda \cdot \omr $ we have, for   $\Al =s ^\ast ( a _\lambda \cdot \omr) $ as in (\ref{al}),
\begin{align}
\label{indang} 
\Alc  &=   \sigma ^\ast ( a _\lambda \cdot \omr)  \nonumber 
=  u  (s ^\ast \Al)  u^{-1} + u  \, \mathrm{d} u^{-1}\\
& = \frac{  (x \, \mathrm{d} y-y \, \mathrm{d} x - \rho  ^2 \, \mathrm{d} \cchi ) \, \qi 
+ (x \, \mathrm{d} \rho - \rho  \, \mathrm{d} x + \rho  \, y \, \mathrm{d} \cchi )\, \qj
+ (\rho  \, \mathrm{d} y-y \, \mathrm{d} \rho  + \rho \,  x \, \mathrm{d} \cchi )\, \qk
}{\lambda ^2 + |q |^2 } 
\nonumber  \\ 
& + \frac{\qi}{2} \, \mathrm{d} \cchi.
\end{align} 
Hence we can read the monopole Higgs field and gauge potential,
 \begin{align}
\label{higgs} 
 \ph _\lambda & 
 = \frac{1}{R }\left( \frac{1}{2} - \frac{\rho ^2 }{\lambda ^2+  |q |^2 } \right)\,  \qi
 + \frac{\rho }{ R }\left( \frac{1}{\lambda ^2 + |q |^2 } \right) (y\, \qj + x\,  \qk ),\\
 \label{a3} 
 \Aaa &
 = \frac{1}{\lambda ^2 + |q| ^2 } \left[ 
 (x \, \mathrm{d} y - y \, \mathrm{d} x )\, \qi
 +  (x\, \mathrm{d} \rho - \rho \, \mathrm{d} x )\, \qj
 + (\rho \, \mathrm{d} y - y \, \mathrm{d} \rho )\, \qk
 \right] .
\end{align} 
We can now see that
\begin{equation}
\lim _{ \rho \to 0 }\,  R \ph _\lambda   =  \frac{\qi}{2  },
\end{equation} 
confirming that the weight of the lifted circle action is $1/2 $.

According to Lemma \ref{lemsec},  $R \ph _\lambda$ can be also calculated as $\sigma ^\ast ( (a _\lambda \cdot \omr) ( \xi ) ) $.
In fact,  using (\ref{acto}) and (\ref{xib}), 
\begin{equation} 
\begin{split} 
( a _\lambda \cdot \omega) (\xi) &
 =\frac{1}{2}\,    \frac{ \Im \left( R ^2 \bar q ^1 \, \mathbf{i}  q ^1 + \lambda ^2  \bar q ^2 \mathbf{i}  q ^2  \right)}{R ^2 | q ^1 |^2  + \lambda ^2 |q ^2 |^2 },\\
\sigma ^\ast (( a _\lambda \cdot \omega) (\xi) ) &
 =\frac{1}{2} \,  \frac{( \lambda ^2 + x ^2 + y ^2  - \rho ^2 ) \mathbf{i} + 2 \rho (y\,  \mathbf{j} + x\,  \mathbf{k} )}{\lambda ^2 + |q| ^2 }
 = R \ph _\lambda .
 \end{split} 
 \end{equation}

\section{The $L ^2 $ metric on $\MMMf$ and its geodesics} 
\label{s3} 
In this section we first compute the $L ^2 $ metric on $\MMMf$ and examine its properties, then study geodesic motion on $\MMMf$, which approximates the adiabatic dynamics of circle-invariant 1-instantons on $S ^4 _R $.

\subsection{The tangent space  $T _{ [ \omega  ]} \MMMf$}
\label{sscomp} 
It will be convenient to think of bundle automorphisms as elements of $ \Lambda ^0 ( P ^1 (\mathbb{H}  ), \mathrm{Ad} (P) )$, where $\mathrm{Ad} (P) $ is the non-linear adjoint bundle $P \times _{ \mathrm{ad}}Sp (1) $, while we keep denoting by $\mathrm{ad} (P) $ the linear adjoint bundle $P \times _{ \mathrm{ad}}\mathfrak{ sp} (1) $.
We write $\Lambda ^p_{ \mathrm{ad} }(A, B ) $  for the ad-equivariant p-forms on  $A$ with values in $B $.
Recall that there are isomorphisms 
\begin{equation}
\label{isoiso} 
 \mathcal{G} \simeq \Lambda ^0_{ \mathrm{ad} }(P, Sp (1) ) \simeq \Lambda ^0 ( P ^1 (\mathbb{H}  ), \mathrm{Ad} (P) ).
 \end{equation}
The first isomorphism is given by   (\ref{gadeqv}), and the second by  (\ref{hf}). At the infinitesimal level (\ref{isoiso}) becomes
 \begin{equation}
 \mathfrak{ X} ^ \mathrm{v}  _ \mathrm{R} (P) \simeq \Lambda ^0_ \mathrm{ad} ( P , \mathfrak{sp}(1) ) \simeq \Lambda ^0( P ^1 (\mathbb{H}  ), \mathrm{ad} (P) ),
 \end{equation} 
 where $ \mathfrak{ X} ^ \mathrm{v}  _ \mathrm{R} (P) $ denotes the vertical right invariant vector fields on $P$. The  isomorphism $ \mathfrak{ X} ^ \mathrm{v}  _ \mathrm{R} (P) \simeq  \Lambda ^0 _ \mathrm{ad} ( P, \mathfrak{ sp } (1)   )$ is given by (\ref{assvec}), the isomorphism between $\Lambda ^0 _ \mathrm{ad} ( P, \mathfrak{ sp }(1)   ) $ and $ \Lambda ^0 ( P ^1 (\mathbb{H}  ), \mathrm{ad}  (P) ) $ is analogous to (\ref{hf}).

As customary, we make the identification
\begin{equation}
\label{ttttt} 
T _{[ \omega ]} \MM _1   
=\{ \upsilon   \in \Lambda ^1 (P ^1 (\mathbb{H}  ), \mathrm{ad} (P)   ): P _{+}\circ \mathrm{d} _{\omega }\,   \upsilon =0, \,  \mathrm{d} ^\dagger _{\omega }  \upsilon =0\} ,
\end{equation} 
where $T _{\left[ \omega \right] } \MM _1 $ is the tangent space to $\MM _1  $ at an equivalence class of connections $ [ \omega ] $, $P _+ =(\mathrm{Id}  + *)/2  $ is the  projection operator onto the space of self-dual 2-forms and 
$\mathrm{d} ^\dagger_{\omega }  = - {*} \mathrm{d}_{\omega } {*} $
 is the formal adjoint of $\mathrm{d} _{\omega } $ with respect to $ \langle \cdot  , \cdot \rangle $.
 The condition  $P _{+}\circ \mathrm{d} _{  \omega  }\,  \upsilon  =0$ is simply the linearised version of anti self-duality. The condition $\mathrm{d} ^\dagger _{ \omega }  \upsilon = 0 $ encodes orthogonality to the (tangent space to) the gauge group orbits. For more details see e.g.~\cite{Groisser:1987uq,Naber:1500641}.  To obtain $T _{[\omega ]} \MMM $ we require both $\omega$ and $\upsilon$ to be circle invariant, that is $ \mathcal{L} _\xi \omega =0 =\mathcal{L} _\xi \upsilon $. 

The framed moduli space is locally a product, so its tangent space at $[ \omega ] $ is the direct sum of $T _{ [ \omega ] } \MMM  $ and of the tangent space $T _{ [ \omega ]} \f$  to the framed directions which we now describe. 
In a neighbourhood of $\mathrm{Id}_P $, an automorphism of $P$, viewed as an element of $\Lambda ^0 ( P ^1  (\mathbb{H}  ), \mathrm{Ad} (P)  ) $,  can be written in the form $\exp \L  $, $\L\in \Lambda ^0 (P ^1  (\mathbb{H}  ), \mathrm{ad} (P) )$. Its action on a connection $\omega$ is, to the linear level, $\omega \rightarrow \omega  + \mathrm{d} _\omega \Lambda $. Quotienting out by the action of $\mathcal{G} _0 $ we have
\begin{equation}
\label{nhgytre} 
T _{ [ \omega ]} \f  = \frac{ 
\mathrm{d} _\omega\Big( \Lambda ^0_ \mathrm{c}  (P ^1 (\mathbb{H}  ) , \mathrm{ad} (P) )\Big) }{ \mathrm{d}  _\omega \Big (   \Lambda ^0 _{ \mathrm{c} 0}  (P ^1 (\mathbb{H}  ) , \mathrm{ad} (P) )\Big)}.
\end{equation} 
Here  $\Lambda ^0 _\mathrm{c}   (P ^1 (\mathbb{H}  ) , \mathrm{ad} (P) ) $ is the subspace of $\Lambda ^0 ( P ^1  (\mathbb{H}  ), \mathrm{ad}(P) ) $ consisting of elements for which the associated vector field in $ \mathfrak{ X} ^ \mathrm{v}  _ \mathrm{R} (P) $ satisfies (\ref{bcond}), and $\Lambda ^0 _{\mathrm{c} 0}  (P ^1 (\mathbb{H}  ) , \mathrm{ad} (P) ) $ is the subset of $\Lambda ^0 _\mathrm{c}   (P ^1 (\mathbb{H}  ) , \mathrm{ad} (P) ) $ consisting of sections which vanish at  the framing point $x _0 $. The equivalence relation in (\ref{nhgytre}) is 
\begin{equation}
\mathrm{d} _\omega \Lambda _1 \sim \mathrm{d} _\omega \Lambda _2 \text{ if } \mathrm{d} _\omega (\Lambda _1 - \Lambda _2) \in \mathrm{d}  _\omega \Big (   \Lambda ^0 _{ \mathrm{c} 0}  (P ^1 (\mathbb{H}  ) , \mathrm{ad} (P) )\Big) .
\end{equation}
 However, a non-trivial anti self-dual  connection on $P ^1 (\mathbb{H}  )$ is irreducible and, for irreducible anti self-dual connections, $\mathrm{d} _\omega : \Lambda ^0 (P ^1 (\mathbb{H}  ), \mathrm{ad} (P) ) \rightarrow \Lambda ^1 (P ^1 (\mathbb{H}  ), \mathrm{ad} (P) ) $ is injective \cite{Freed:105021}. Therefore, equivalently,
\begin{equation}
\label{eccc}
T _{ [ \omega ]} \f 
= \left\{  
[\mathrm{d} _{ \omega }  \L] _0  \, :\,   \L\in \Lambda ^0 _ \mathrm{c} (P ^1 (\mathbb{H}  ) , \mathrm{ad} (P))
\right\},
\end{equation} 
where the equivalence relation $[ \cdot , \cdot ] _0 $ is  
\begin{equation}
\mathrm{d} _{\omega  } \Lambda _1 \sim  \mathrm{d} _{\omega   } \Lambda _2 \text{ if } (\Lambda _1- \Lambda _2 ) (x _0 ) =0.
\end{equation}

\subsection{The $L ^2 $ metric on $\MMMf$}
The  inner product $\langle \cdot , \cdot \rangle $ on $\Lambda ^1 (P ^1 (\mathbb{H}  ), \mathrm{ad} (P))$ given in \eqref{ip} is manifestly invariant under bundle automorphisms. This, and the identification of $T _{ [ \omega ] }\MMM$  with the $L ^2 $-orthogonal complement of the tangent space $ \mathrm{d} _\omega (\Lambda  ^0 ( P ^1 (\mathbb{H}  ), \mathrm{ad}(P) )$ to the orbit of $\mathcal{G}$  through $\omega $ imply that $\langle \cdot , \cdot \rangle $  descends to a metric on $\MMM$.
For $\upsilon  _1 , \upsilon  _2 \in T _{ [ \omega ] }\MMM$, the $L ^2 $ metric is therefore simply
\begin{equation}
\label{metric} 
\gmm( \upsilon  _1 , \upsilon  _2  ) 
=\langle \upsilon  _1 , \upsilon  _2  \rangle .
\end{equation}

In extending the metric to the framed directions, there is an important difference between the case of instantons over the non-compact space $\mathbb{H}  $ and over the compact space $P ^1 (\mathbb{H}  )$. Let $M$ denote either of these spaces. 
If $M =\mathbb{H}  $ we require elements of $ T _{[\omega ]} \MMMf $ to have finite $L ^2 $ norm. For $\mathrm{d} _\omega \Lambda \in T _{ [ \omega ] }\f $ this condition implies that  $\lim _{ |q| \rightarrow \infty }\Lambda (q)$ is finite and independent of the direction. For this reason, on $\mathbb{H}  $ we can take infinity as our framing point.

To extend the metric, we  need  to select representatives of elements in $T _{ [ \omega ]} \f  $ in a way compatible with (\ref{ip}). That is, if $ \mathrm{d} _\omega \bar\Lambda $ is such a representative, it needs to be orthogonal to any trivial element  in the same equivalence class,
\begin{equation}
\label{agsd} 
0 =\langle \mathrm{d}_ \omega \bar\Lambda,  \mathrm{d} _\omega \Gamma  \rangle =\langle\mathrm{d} ^\dagger _\omega \mathrm{d}_ \omega \bar\Lambda,   \Gamma  \rangle
\end{equation} 
for all $\mathrm{d}_ \omega \Gamma \in \mathrm{d} _\omega( \Lambda ^0 _{ \mathrm{c} 0 }(M , \mathrm{ad} (P) ) )$.
Therefore, we obtain the condition
\begin{equation}
\label{orthg1} 
\mathrm{d}  ^\dagger _{\omega } \mathrm{d} _{\omega }  \bar \Lambda  =0.
\end{equation} 

For $M =\mathbb{H}$, imposing (\ref{orthg1}), also known as Gauss's law,  results in the identification of $T _{ [ \omega ] }\f $ with
the space of  elements $ \mathrm{d} _{\omega } \L  $ which are circle invariant, have finite $L ^2 $ norm  and are orthogonal to all the infinitesimal bundle automorphisms  vanishing at infinity.  We can therefore extend the $L ^2 $ metric on $T _{ [ \omega ] } \MMM $ to a metric on $T _{ [ \omega ] } \MMMf $ in a well defined manner. Note that for $u\in T _{ [ \omega ] } \MMM $, $v\in T _{ [\omega]} \f $, $\gm(u,v)=0 $  as $u$ is orthogonal to the tangent space to the gauge group orbits while $v$ is parallel. Hence $\gm $ is a product metric.

For  $M =P ^1 (\mathbb{H}  )$,  the operator  $ \mathrm{d} ^\dagger _{\omega } $ restricted to the image of $\mathrm{d} _\omega $ is injective, so that  (\ref{orthg1})  only has the  trivial solution 
$\mathrm{d} _{ \omega } \bar \Lambda  =0 $ \cite{Naber:1500641}.
Therefore,  there is no way to select representatives of the equivalence classes in (\ref{eccc}) in a way which is compatible with $\langle \cdot , \cdot \rangle $. 
Instead, we can select a non-trivial element   to serve as a basis of $T _{ [ \omega ] }\f $, and express any other element as a multiple of it.
As discussed in Section \ref{cinst}, see in particular Proposition \ref{prop8}, the Higgs field $\hi  $ is a non-trivial infinitesimal bundle automorphism, and so  $ [\mathrm{d} _\omega \hi] _0  $ is a natural candidate for a  preferred element of $T _{ [ \omega ] }\f $.
For $v _1, v _2\in T _{ [ \omega ] }\f $,  $v _i =c _i [ \mathrm{d} _\omega \hi ] _0 $, $i =1, 2 $,  we thus define
\begin{equation}
\label{ldoiepu} 
\gm(v _1 , v _2 )
= c _1 c _2\,  || \mathrm{d} _\omega \hi|| ^2 .
\end{equation}   
For the same reasons as in the non-compact case, $\gm $ is then a product metric on $\MMM$.

While our definition of  the metric in the framed directions clearly depends on picking a gauge, our choice appears to be essentially canonical in the context of    circle invariant instantons. At least, we do not see any other way of satisfying all the requirements discussed in  Section \ref{cinst}.


\subsection{Computation of the metric}
We have a 4-parameters family of anti self-dual circle-invariant 1-instantons depending on the moduli $\lambda,\alpha , \beta  , \chi $. The relation between the real parameters $\alpha $, $\beta $ and the complex parameter $m$  is given by (\ref{difpar}).
Let  $\omega$ be a circle invariant 1-instanton, $ A = s ^\ast \omega $ for some section $s$. Since a gauge potential $A$ is enough to determine a connection $\omega$ on $P$, we will work with $A$ rather than $\omega$, and denote by 
 $ \phi _A $ the pullback $  s ^\ast \hi $ of the associated Higgs field $\hi =\omega (\xi) $. For $u$ one of the moduli, write $ \delta _u A$ for $  \mathrm{d} A/ \mathrm{d} u $.
Since $ A(u)$ is anti self-dual and-circle invariant for any value of $u$,  $\delta _u A$  is also anti self-dual and circle-invariant.  
If $u$ is one of the unframed  moduli $\lambda , \alpha , \beta  $, projection orthogonally to the orbits of $\mathcal{G} $  is obtained by 
replacing $ \delta _u A $ with $ \pg{ \delta _{ u }A} = \delta _u A  +  \mathrm{d} _A \psi $, where $ \psi\in \Lambda ^0( P ^1 (\mathbb{H}  ), \mathrm{ ad }(P) ) $ is chosen so that
\begin{equation}
\label{orthg} 
\mathrm{d} ^\dagger _A (\pg{\delta _u A } )=0 .
\end{equation} 

Let $\hat A = (  \hat\phi ^{-1} _N  )^\ast A  $, $\pg{\delta_u  \hat A}= (  \hat\phi ^{-1} _N  )^\ast\pg{{\delta_u  A}} $.
Using Equation (\ref{axcvd}) 
we have
\begin{equation}
(  \hat\phi ^{-1} _N  )^\ast  (\pg{{\delta _u A}} \wedge  *\,  \pg{{\delta _v A}} )
= \frac{R ^4 }{(R ^2 + |q| ^2  )^2 }\pg{\delta _u \hat A}\wedge \hat* \, \pg{\delta _v \hat A},
\end{equation} 
where $\hat * $ denotes the Hodge operator with respect to the metric $\gh =\mathrm{d} \bar q \,  \mathrm{d} q $ on $\mathbb{H}  $, and
\begin{equation}
\begin{split} 
\langle  \pg{{\delta _u A}}, \pg{{\delta _v A}} \rangle  &
= - \frac{1}{2} \int _{ P ^1 (\mathbb{H}  ) } \!\!\!\!\!\!\!\mathrm{Tr} \left( \pg{{\delta _u A}} \wedge  * \, \pg{{\delta _v A}} \right) 
= - \frac{1}{2}\int _{ \hat\phi_N ( P ^1  (\mathbb{H}  )   )}\!\!\!\!\!\!\!\!\!\!\!\!\!\!\!\!\!\! \mathrm{Tr}\left( \phi_N ^{-1} \right)  ^\ast  \left[\pg{{\delta_u  A}} \wedge  * \,  \pg{\delta_u  A} \right] \\ &
=  \int _{  \mathbb{H}    }\fp{\pg{\delta _u \hat A}}{\pg{\delta _v \hat A}} \frac{R ^4 }{(R ^2 + |q| ^2  )^2 }
 \, \mathrm{vol}_{ \mathbb{H}   }.
\end{split} 
\end{equation} 
Similarly 
\begin{equation}
\langle \dchia, \dchia \rangle 
=  \int _{  \mathbb{H}    }\fp{\dchia }{\dchia} \frac{R ^4 }{(R ^2 + |q| ^2  )^2 }
 \, \mathrm{vol}_{ \mathbb{H}   }.
\end{equation} 

\begin{prop}
\label{propmetric} 
The $L ^2 $ metric $\gm $ is of the form
 \begin{equation}
 \label{m2} 
\gm=
 f _R ( \lambda )\,  \mathrm{d} \lambda ^2  + g _R  (\lambda) R ^2 (\mathrm{d} \alpha  ^2 + \sin ^2 \alpha \, \mathrm{d} \beta ^2 ) + h_R  (\lambda)  \, \mathrm{d} \chi  ^2 .
\end{equation} 
\end{prop}
The factor of $R ^2 $ has been inserted for future convenience.
\begin{proof} 
The metric $\gm $  has an $SO (3) \times U (1)  $ isometry group coming from space rotations and non-trivial (differing from the identity at $[R,0] $) bundle automorphisms which rotate the $U (1) $ phase --- let us call the latter phase rotations.
Invariance of  $\gm $ under the $U (1) $ factor comes directly from its invariance under bundle automorphisms.
Space rotations are generated by the left action on $S ^7 _{ \sqrt{  R }} $ of the group $\mathcal{N} \cap Sp (2) $  given by  (\ref{wash}). Invariance of $\gm $  under space rotations follows since the induced action on $P ^1 (\mathbb{H}  ) $ is an isometry and  tangent vectors to $\MMMf $ transform by pullback. 

Let us consider the action of space and phase rotations on a gauge potential $A _{ \boldsymbol{ \lambda }} $ with unframed moduli $\boldsymbol{ \lambda } $ and associated Higgs field $\phi_{A _{ \boldsymbol{ \lambda }} }$. Up to equivalence, phase rotations are of the form  $\exp( R \phi_{A _{ \boldsymbol{ \lambda }} } \,  \chi  )$ with $ \chi  \in [0, 2 \pi )$ the modulus in the framed direction.
Under phase rotations  $ \boldsymbol \lambda $ is clearly  invariant. The Higgs field gets conjugated but its asymptotic value, and hence $\chi $, stays unchanged.

Consider now space rotations. The action of an element 
\begin{equation}
\HH\ni h =\begin{pmatrix}
a & b\\- \bar b  &\bar a
\end{pmatrix}, \quad a,b\in \mathbb{C},  |a| ^2 + |b| ^2 =1
\end{equation} 
on a point $( \lambda, m\in \hat{\mathbb{C}  }) $ of the moduli space leaves $\lambda$ invariant and acts on $m$ as a rotation, 
\begin{equation}
m \xmapsto{h} \frac{R (am + Rb )}{R \bar a  -   \bar b m } = (\hat \phi _N \circ \rho_ {h } \circ \hat \phi _N  ^{-1} ) (m) ,
\end{equation} 
hence $\boldsymbol{ \lambda } $ transforms as a vector.
The modulus $\chi$ is  the ratio between the asymptotic norms of an infinitesimal bundle automorphism $\Lambda$ and of $R \, \phi _{A _{ \boldsymbol{ \lambda }} }$. As both $\Lambda$ and $\phi_{A _{ \boldsymbol{ \lambda }} }$ transform by pullback,  $ \chi$ is invariant under space rotations.

\end{proof} 

Let us now compute the unknown functions $f _R , g _R , h _R  $ appearing in (\ref{m2}).
\begin{lem}
\label{lf} 
The function $f _R $ is given by 
\begin{equation}
\label{fl} 
f _R  (\lambda) 
= \frac{4 \pi ^2  R ^4 }{(\lambda  ^2 -R ^2  )^4 } \left[\lambda ^4 + 10 \lambda ^2 R ^2  +R ^4 
 - 12 R ^2  \lambda ^2\left( \frac{  \lambda ^2 +R ^2  }{\lambda  ^2 - R ^2 } \right) \log \left( \frac{\lambda}{R} \right )\right] .
\end{equation} 
\end{lem}
\begin{proof} 
Consider the variation of the gauge potential  $\Al$, see Equation (\ref{al}),  with respect to the collective coordinate $\lambda$,
\begin{equation}
\delta _\lambda \Al
=  \frac{\mathrm{d} \Al}{\mathrm{d} \lambda }
= - \frac{2 \lambda \,  \Im (\bar q \, \mathrm{d} q)}{\left( |q |^2 + \lambda ^2  \right) ^2 }.
\end{equation} 
It already satisfies the orthogonality condition (\ref{orthg}) 
\begin{equation}
\begin{split}
\mathrm{d} ^\dagger _{\Al} \delta _\lambda \Al &
 = - \frac{2 \lambda }{(\lambda ^2 + |q| ^2) ^4} 2 x ^i  \Al (\partial _i )
  -  \frac{2 \lambda }{(\lambda ^2 + |q| ^2) ^2}  \partial _i (\Al (\partial _i ))   = 0
 \end{split}
 \end{equation} 
 as $\Al $ is in a transverse gauge and satisfies $ \partial _i (\Al (\partial _i ))   = 0$.
Its squared norm  is
\begin{equation} 
\begin{split} 
| \delta _\lambda \Al |_{ \mathbb{H}  } ^2&
=  - \frac{1}{2} \mathrm{Tr}   \fp{\delta _\lambda \Al}{\delta _\lambda \Al}  
= \frac{4 \lambda ^2  }{\left(  |q |^2  + \lambda ^2  \right) ^2 } |\Im (\bar q \, \mathrm{d} q ) | _{ \mathbb{H}  }^2   
= \frac{12\lambda ^2 |q |^2  }{\left( |q |^2 + \lambda ^2  \right) ^4} . 
\end{split} 
\end{equation} 
Multiplying by the conformal factor $R ^4 ( R ^2 + |q |^2  )^{-2}  $ and integrating over $\mathbb{H}  $ yields
\begin{equation}
f _R  (\lambda) 
=|| \delta _\lambda  \All || ^2 
= 12\lambda ^2   \int _{\mathbb{H}}  \frac{ |q |^2  }{\left( |q |^2 + \lambda ^2   \right) ^4} \frac{ R ^4 }{( R ^2 + |q |^2 )^2 } \, \mathrm{vol } _{\mathbb{H}  }.
\end{equation} 
To compute the integral switch to spherical coordinates
\begin{equation}
\begin{split} 
x& =r \sin \kappa \sin \theta \cos \phi,\quad 
y =r \sin \kappa \sin \theta \sin \phi ,\quad 
z = r \sin \kappa \cos \theta ,\quad 
w  =r \cos \kappa,
\end{split} 
\end{equation} 
with $r =|q |\in[0, \infty ) $, $\theta, \kappa \in[0, \pi ]$, $\phi \in[0, 2 \pi )$, and volume element
\begin{equation}
\mathrm{vol}  _{\mathbb{H} } =\mathrm{d} x \wedge \mathrm{d} y \wedge \mathrm{d} z \wedge \mathrm{d} w
= r ^3 \sin ^2 \kappa \sin \theta \,  \mathrm{d} r \wedge \mathrm{d} \kappa \wedge \mathrm{d} \theta \wedge \mathrm{d} \phi.
\end{equation} 
The angular integral contributes a factor $2 \pi ^2 $, the volume of  $S ^3$, therefore
\begin{equation}
\begin{split} 
f _R  (\lambda) &
=24 \pi ^2  \lambda ^2  \int _0 ^{ \infty }\frac{ R ^4 \, r ^5 }{( r ^2 + \lambda ^2   )^4 (R ^2  + r ^2 )^2 } \mathrm{d} r \\ &
= \frac{4 \pi ^2  R ^4 }{(\lambda  ^2 -R ^2  )^4 } \left[\lambda ^4 + 10 \lambda ^2 R ^2  +R ^4 
 - 12 R ^2  \lambda ^2\left( \frac{  \lambda ^2 +R ^2  }{\lambda  ^2 - R ^2 } \right) \log \left( \frac{\lambda}{R} \right )\right] . 
\end{split} 
\end{equation} 

\end{proof} 


\begin{lem}
\label{lg} 
The function $g _R $ is given by 
\begin{equation}
g_R (\lambda) 
=  \frac{\pi ^2  }{2( \lambda  ^2 - R ^2 ) ^2 }  \left[ 
\lambda  ^4 - 8 \lambda  ^2  R ^2  + R ^4  
+ \frac{24 \lambda ^4 R^4 }{\lambda ^4 - R ^4 } \log\left(\frac{\lambda }{R }\right)\right] .
\end{equation} 
\end{lem}
\begin{proof} 
Because of spherical symmetry, it is enough to consider the variation with respect to e.g.~$ \alpha $. Let $\hat A _{ \lambda , \alpha }  =g _{\alpha,0} a _{\lambda} \cdot \Ar $.
Using (\ref{atrgen}) we get
\begin{equation}
\hat A _{\lambda, \alpha }
= \frac{\Im \left[
\frac{R}{2} \left(  \frac{\lambda ^2}{ R ^2 } - 1 \right)  \sin \alpha\, \mathrm{d} q 
 + \left(\cos ^2  \frac{\alpha}{2} +  \frac{\lambda ^2 }{R ^2 }\sin ^2 \frac{\alpha}{2} \right) \bar q\, \mathrm{d} q 
 \right] }{\cos ^2  \frac{\alpha}{2} |q| ^2 + R ^2 \sin ^2 \frac{\alpha}{2}  - R\, x \sin \alpha   + \frac{\lambda ^2 }{R ^2 }  \left(R ^2  \cos ^2 \frac{\alpha }{2} + 
\sin ^2\frac{\alpha }{2} |q |^2 + R \, x \sin \alpha \right)  }.
\end{equation} 
Its variation with respect to $\alpha$, evaluated for $\alpha =0 $, is
\begin{equation}
\begin{split} 
\label{oooopl} 
{\delta _\alpha  \hat A_{\lambda  , \alpha } } &
= \left. \frac{\mathrm{d} \hat A _{ \lambda ,\alpha} }{ \mathrm{d} \alpha}  \right | _{ \alpha =0 } \!\!\!\!
=\frac{1}{2R}  \frac{\lambda ^2 - R ^2 }{ (\lambda ^2  + |q |^2 ) ^2  }  
 \left[ ( \lambda ^2 + |q |^2) \Im( \mathrm{d} q) - 2x\,  \Im\left(   \bar q \, \mathrm{d} q \right) \right] \\ &
 =\frac{1}{2R}  \frac{\lambda ^2 - R ^2 }{ (\lambda ^2  + |q |^2 ) ^2  }  \cdot \\ 
\Big[ &\left[ 
 2x(y \, \mathrm{d} x - w \, \mathrm{d} z + z \, \mathrm{d} w) + (\lambda ^2 + y ^2 + z ^2 + w ^2 - x ^2 ) \mathrm{d} y)
 \right] \qi\\ 
 +&\left [ 2x(z \, \mathrm{d} x + w \, \mathrm{d} y - y \, \mathrm{d} w) + (\lambda ^2 + y ^2 + z ^2 + w ^2 - x ^2) \mathrm{d} z \right ]\qj
\\ 
 +  &\left[2x(w \, \mathrm{d} x -z \, \mathrm{d} y + y \, \mathrm{d} z) + (\lambda ^2 + y ^2 + z ^2 + w ^2 - x ^2) \mathrm{d} w  \right] \qk \Big] . 
 \end{split} 
\end{equation} 

We now need to project $ \delta  _\alpha  \Alaa$ orthogonally to the gauge group orbits. Here we partially follow \cite{Habermann:1988wv}.
First notice that, since $\All $ is irreducible, the Laplacian $ \mathrm{d} ^\dagger _{  \All } \mathrm{d} _{\All }: \Lambda ^0  (P ^1 (\mathbb{H}  ), \mathrm{ad} (P)   ) \rightarrow \Lambda ^0  (P ^1 (\mathbb{H}  ), \mathrm{ad} (P)   ) $ is invertible. We denote by $G_{ \All} $ its inverse. Then
\begin{equation}
\label{cffp} 
\pg{(\delta _\alpha \Alaa)} 
= { \delta  _\alpha   \Alaa} - \mathrm{d} _{ \All } G_{ \All} \mathrm{d} ^\dagger _{ \All }{ \delta  _\alpha  \Alaa},
\end{equation} 
and
\begin{equation}
\label{prprpr} 
\begin{split} 
g_R (\lambda) R ^2  &=
|| \pg{(\delta _\alpha   \Alaa)} || ^2
=|| \delta  _\alpha  \Alaa || ^2  
- 2 \langle  \delta  _\alpha  \Alaa, \mathrm{d} _{  \All } G_{ \All} \mathrm{d} ^\dagger _{  \All } { \delta  _\alpha  \Alaa} \rangle \\ &
+ \langle \mathrm{d} _{\All} G_{ \All} \mathrm{d} ^\dagger _{\All } { \delta  _\alpha  \Alaa},\mathrm{d} _{\All } G_{ \All} \mathrm{d} ^\dagger _{\All } { \delta  _\alpha  \Alaa} \rangle \\ &
= || \delta  _\alpha  \Alaa || ^2
- 2 \langle  \mathrm{d} ^\dagger _{\All}\delta  _\alpha \Alaa,   G_{ \All} \mathrm{d} ^\dagger _{\All} { \delta  _\alpha \Alaa} \rangle 
+ \langle \mathrm{d} ^\dagger _{\All} { \delta  _\alpha  \Alaa}, G_{ \All} \mathrm{d} ^\dagger _{\All} { \delta  _\alpha  \Alaa} \rangle\\ &
=|| \delta  _\alpha  \Alaa || ^2
-  \langle  \mathrm{d} ^\dagger _{\All}\delta  _\alpha  \Alaa,   G_{ \All} \mathrm{d} ^\dagger _{\All} { \delta  _\alpha \Alaa}\rangle. 
\end{split} 
\end{equation} 
From (\ref{oooopl}) we have
\begin{equation}
\begin{split}
|{ \delta _\alpha  \Ala } | ^2  _{ \mathbb{H}  }
&= \frac{1}{4 R ^2 }  \left(  \frac{\lambda ^2 - R ^2 }{ (\lambda ^2  + |q |^2 ) ^2  } \right)^2  \cdot  \\ &
\left[4 x ^2  |\Im (\bar q \, \mathrm{d} q) |^2_{ \mathbb{H}  } +  (\lambda ^2 + |q| ^2  )^2 | \Im( \mathrm{d} q )|^2 _{ \mathbb{H}  }-4x ( \lambda ^2 + |q| ^2  )
\fp{  \Im\left(   \bar q \,\mathrm{d} q \right)}{ \Im( \mathrm{d} q)   } \right] \\ &
= \frac{1}{4 R ^2 }   \frac{(\lambda ^2 - R ^2) ^2  }{ (\lambda ^2  + |q |^2 ) ^4  }
\left( 12  x ^2 |q| ^2 + 3 (\lambda ^2  + |q| ^2  )^2 - 12 x ^2 (\lambda ^2 + |q| ^2 ) \right) \\ &
= \frac{3}{4 R ^2 }   \frac{(\lambda ^2 - R ^2) ^2  }{ (\lambda ^2  + |q |^2 ) ^4  }
\left(  (\lambda ^2  + |q| ^2  )^2 - 4 x ^2 \lambda ^2  \right) . 
\end{split} 
\end{equation} 
Multiplying by the conformal factor $ R ^4 /( |q| ^2 + R ^2 )^2 $ and integrating we get
\begin{equation}
\label{nmnmnm} 
\begin{split} 
&|| \delta _\alpha  \Alaa || ^2  
=\frac{3}{2}  \pi ^2R ^2   (\lambda ^2 - R ^2  )^2  
\int _0 ^\infty \left( \frac{\lambda ^4 + \lambda ^2 r ^2 + r ^4  }{ (\lambda ^2 + r ^2 )^4 } 
 \frac{r ^3 }{( R ^2 + r ^2 )^2 } \right)  \, \mathrm{d} r\\ &
 =\frac{\pi ^2 R ^2}{4( \lambda  ^2 - R ^2 )^2} 
 \left[  - 6\left( \frac{R ^2 + \lambda ^2 }{R ^2 - \lambda ^2 } \right) (R ^4 + \lambda ^4 ) \log \left(  \frac{\lambda }{R} \right)-7 R ^4  + 2 R ^2 \lambda ^2 -7 \lambda ^4 
  \right] .
\end{split} 
\end{equation} 
By using (\ref{astrf}) we have
\begin{equation}
\begin{split} 
\label{gper} 
(\hat \phi _N ^{-1} )^\ast \left( \mathrm{d}  ^\dagger _{ \All  } {\delta _\alpha  \Alaa}    \right) &
= - (\hat \phi _N ^{-1} )^\ast \left(  *  \mathrm{d} _{ \All}  *{\delta _ \alpha   \Alaa}    \right) \\ &
=  \frac{(R ^2  + |q| ^2 )^4 }{R ^8 }  \mathrm{d} _{\Al} ^\dagger   \left( \frac{R ^4 }{(R ^2  + |q| ^2 ) ^2 }  \delta_\alpha  \Ala \right) .
\end{split} 
\end{equation} 
This quantity is  computed in Appendix \ref{projjj}, see Equation (\ref{jhgjfk}). The result is
\begin{equation} 
\label{kk1} 
(\hat \phi _N ^{-1} )^\ast \mathrm{d} _{  \All } ^\dagger  {\delta _\alpha   \Alaa } 
= - \frac{2}{ R ^5 }\frac{( \lambda  ^2 - R ^2  ) ^2  (R ^2  + |q| ^2)}{(\lambda ^2 + |q| ^2  )^2  }   \, \gamma _1,
\end{equation} 
with $\gamma _1 =-(y\qi + z\qj + w\qk )$.
The quantity $(\hat \phi _N ^{-1} ) ^\ast \left[  G_{  \All }  \mathrm{d} _{  \All } ^\dagger  { \delta _\alpha  \Alaa }  \right] $ is also computed in Appendix 
\ref{projjj}, Equation (\ref{ggggf}),
\begin{equation}
\label{kk2} 
(\hat \phi _N ^{-1} ) ^\ast \left[  G_{  \All }  \mathrm{d} _{  \All } ^\dagger  { \delta _\alpha  \Alaa }  \right] 
= - \frac{1}{2R}   \frac {(\lambda  ^2 - R ^2 )^2 }{ (\lambda ^2 + R ^2 ) (  \lambda  ^2 + |q| ^2 ) }  \, \gamma _1 .
\end{equation} 
By using (\ref{kk1}), (\ref{kk2}) and (\ref{axcvd})   we have
\begin{equation}
\begin{split}
& - \frac{1}{2}  \mathrm{Tr} \left[ 
 (\hat \phi _N ^{-1} )^\ast  (G_{ \All }  \mathrm{d} _{ \All } ^\dagger   \delta _\alpha \Alaa)
 \wedge (\hat \phi _N ^{-1} )^\ast * \left( \mathrm{d} _{ \All } ^\dagger \delta _\alpha \Alaa  \right) \right] \\&
 = R ^2 \frac {(\lambda  ^2 - R ^2) ^4}{ \lambda ^2 + R ^2}\frac{ y ^2 + z ^2 + w ^2}{(\lambda ^2 + |q| ^2  )^3(R ^2  + |q| ^2) ^3 } \mathrm{vol} _{ \mathbb{H}  } .
\end{split}
\end{equation} 
Integrating we get
\begin{equation}
\label{ppprod} 
\begin{split}
&
\left\langle   G_{ \All }  \mathrm{d} _{ \All } ^\dagger  \delta _\alpha \Alaa,
 \mathrm{d} _{ \All } ^\dagger   \delta _\alpha \Alaa \right\rangle 
 =R ^2 \frac {(\lambda  ^2 - R ^2) ^4 }{ \lambda ^2 + R ^2} \cdot 2 \pi ^2 
\int _0 ^\infty \frac{3}{4} r ^2 \frac{r ^3\, \mathrm{d} r}{ (\lambda ^2 + r ^2  )^3(R ^2  + r ^2) ^3}  \\ &
= \frac{3}{2} \pi ^2 R ^2  \frac {(\lambda  ^2 - R ^2) ^4 }{\lambda ^2 + R ^2 } 
\int _0 ^\infty   \frac{r ^5 \, \mathrm{d} r}{(R ^2  + r ^2)^3 (\lambda ^2 + r ^2  )^3}  \\ &
=\frac{3}{4} \pi ^2 R^2  \left[ 
 - 3 +2 \frac{  \left(\lambda ^4+4 \lambda ^2 R^2 +R^4\right) } { \left(\lambda ^2 - R ^2\right) \left(\lambda ^2 + R^2\right)}
\log \left(\frac{\lambda }{R}\right) \right] .
\end{split}
\end{equation} 
Taking the difference of (\ref{nmnmnm}), (\ref{ppprod})   we finally obtain
\begin{equation} 
\label{gl}
\begin{split} 
g_R (\lambda) R ^2 &
= || \pg{(\delta _\alpha  \Alaa)}   ||^2 \\ &
=  \frac{\pi ^2 R^2 }{2( \lambda  ^2 - R ^2 ) ^2 }  \left[ 
\lambda  ^4 - 8 \lambda  ^2  R ^2  + R ^4  
+ \frac{24 \lambda ^4 R^4 }{\lambda ^4 - R ^4 } \log\left(\frac{\lambda }{R }\right)\right] .
\end{split} 
\end{equation} 
\end{proof} 

\begin{lem}
\label{lh} 
The function $h _R $ is given by 
\begin{equation}
h _R (\lambda)  
=(\lambda ^2 /2) f _R  (\lambda).
\end{equation} 
\end{lem}
\begin{proof} 
Working in the  gauge (\ref{indang}),  we have $\pphi =\ph _\lambda $, $\Alc = \Aaa + \ph _\lambda R \, \mathrm{d} \cchi $, so $ \mathrm{d} _{ \Alc } \ph _\lambda= \mathrm{d} _{ \Aaa  }\ph _\lambda $, which is given by
\begin{align}
\mathrm{d}  _{ \Aaa }\ph _\lambda 
&=\frac{2 \lambda ^2 \rho }{R ( \lambda ^2 + |q |^2 ) ^2 }( - \mathrm{d} \rho \,   \qi +  \mathrm{d} y \,  \qj + \mathrm{d} x\,  \qk), \\
\left |  \mathrm{d}  _{ \Aaa  }\ph _\lambda   \right | _{ \mathbb{H}  }^2  &
= \frac{12 \lambda ^4 \rho  ^2 }{R ^2 (\lambda ^2 + |q | ^2 )^4 }.
\end{align} 
Integrating over $ P ^1  (\mathbb{H}  ) $ we  get
\begin{equation}
\label{hl} 
\begin{split} 
h _R  (\lambda) &
= R ^2  ||\mathrm{d}  _{ \Aaa }\ph _\lambda    ||^2 
= ( \lambda ^2 /2) || \delta _\lambda A _\lambda  ||^2 
=(\lambda ^2 /2) f _R  (\lambda) . 
\end{split}  
\end{equation} 

\end{proof} 
Putting together Proposition \ref{propmetric} with Lemmas \ref{lf}, \ref{lg}  and \ref{lh}  we obtain the following theorem.
\begin{thm}
\label{thmmetric} 
The metric $\gm $ on $\MMMf $ is given by
\begin{equation}
\label{metricmod} 
\gm
= f _R (\lambda) (\mathrm{d} \lambda ^2 +(\lambda ^2/2) \mathrm{d} \chi  ^2 ) + g _R  (\lambda) R ^2 (\mathrm{d} \alpha ^2 + \sin ^2 \alpha  \, \mathrm{d} \beta ^2 ),
\end{equation} 
with $ \lambda \in(0, R]$, $\chi \in[0, 2 \pi )$ and
\begin{align} 
\label{ffffff} 
f _R  (\lambda) &
= \frac{4 \pi ^2   }{(\frac{\lambda  ^2}{R ^2}  - 1  )^4 } 
\left[\frac{\lambda ^4}{R ^4}  + 10 \frac{\lambda ^2}{R ^2}  + 1
 - 6 \frac{\lambda ^2}{R ^2}  \left( \frac{  \frac{\lambda ^2}{R ^2}  +1  }{\frac{\lambda  ^2}{R ^2}  - 1} \right) \log \left( \frac{\lambda ^2 }{R ^2 } \right )\right],  \\
g _R  (\lambda) &
=  \frac{\pi ^2  }{2( \frac{\lambda  ^2}{ R ^2 } - 1 ) ^2 }  \left[ 
\frac{\lambda  ^4}{ R ^4 } - 8 \frac{\lambda  ^2}{R ^2}  + 1
+12 \frac{ \frac{\lambda ^4}{ R ^4 }  }{\frac{\lambda ^4 }{ R ^4 } - 1 } \log\left(\frac{\lambda ^2 }{R ^2  }\right)\right] .
\end{align} 
\end{thm}
The metric of Theorem \ref{thmmetric} is in agreement with the $L ^2 $ metric on the unframed moduli space of 1-instantons on the 4-sphere calculated in \cite{Habermann:1988wv}.


Both $f _R $ and $g _R $ depend only on the ratio $\lambda /R $, hence we write
 $f( \lambda /R) =  f _R (\lambda ) $, $g (\lambda /R)=   g _R (\lambda) $.  
Since $f (\lambda /R) = (R/ \lambda ) ^4 f (R/ \lambda  )  $, $g (\lambda /R) =g (R/ \lambda  ) $, the metric (\ref{metricmod}) is invariant under the transformation
$\lambda/R  \rightarrow R/ \lambda  $, $m/R \rightarrow - R  \, m/ |m |^2$
which maps a  gauge potential into an equivalent one. In terms of the angles $\alpha$, $\beta$, the transformation $m/R \rightarrow - R  \, m/ |m |^2$ is simply the antipodal mapping $ \alpha \rightarrow \pi  - \alpha $,
$\beta  \rightarrow \pi +  \beta$. To ease notation,  we set 
\begin{equation}
x =  \lambda / R .
\end{equation} 

The functions $f$ and $g$ are plotted in Figure \ref{fgplot}.
\begin{figure}[htbp]
\begin{center}
\includegraphics{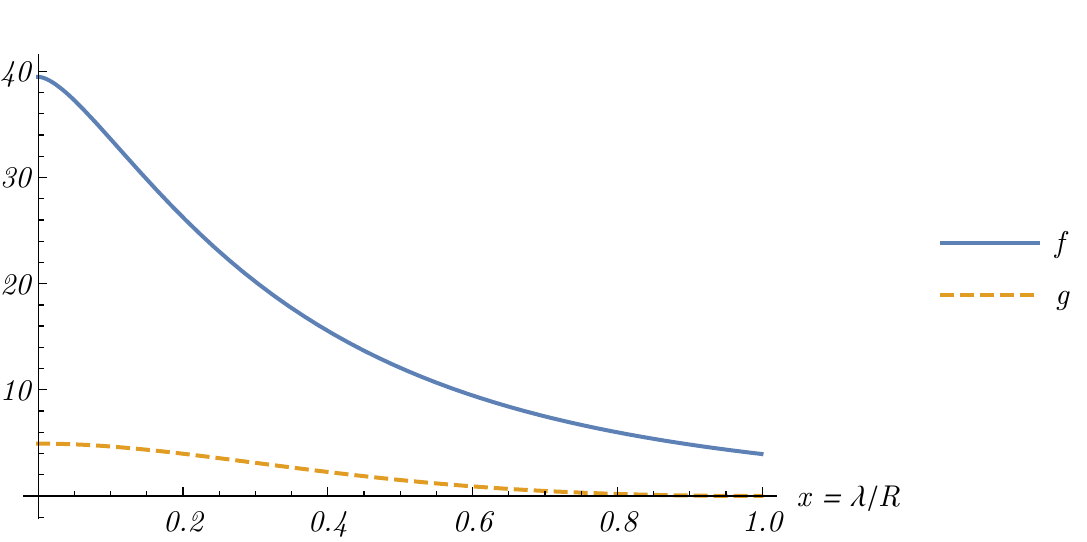}
\caption{Plot of the functions $f$ and $g$ appearing in the metric (\ref{metricmod}). While $f$ is strictly positive, $g$ vanishes at $x =1 $.}
\label{fgplot}
\end{center}
\end{figure}
The function $f$ is strictly positive and monotonously decrescent in $(0,1] $, while $g$ is non-negative and monotonously decrescent. At the endpoints we have
\begin{alignat}{2}
\label{aaaab1} 
\lim_{x \rightarrow 1}f (x)&=2 \pi ^2 / 5,& \quad
\lim _{ x \rightarrow 0 }f (x) &= 4 \pi ^2 ,\\
\label{aaaab2} 
\lim_{x \rightarrow 1}g(x)&=0 , &\quad 
\lim _{ x \rightarrow 0 } g (x) &= \pi ^2 /2 .
\end{alignat}  
While the metric is finite on the boundary  of the moduli space, a calculation shows that its scalar curvature diverges as $x \rightarrow 0 $, see Equation (\ref{divofs}).

\subsection{Properties of the metric and behaviour for $R \rightarrow \infty $}
\label{s4} 
The metric $\gm $ has an $SO(3) \times U (1) $ symmetry. 
Because of circle invariance, only the  $SO (3,1) $ subgroup of the $SO(5,1) $ symmetry group of Yang-Mills equations acts on the unframed moduli space. The $L ^2 $ metric on $\MMM  $ is not conformally invariant, hence  its symmetry group is  $SO (3)  \subset SO (3,1 ) $.
The $U (1) $ factor comes from the $Sp (1) $ structure group reduced  to the $U (1) $ subgroup commuting with the asymptotic value of the Higgs field. 

The radii of the  $SO (3) $ and $U (1) $ orbits  differ in their $\lambda$-dependence. The $SO (3) $ orbits are 2-spheres of squared radius $R ^2  g $ which monotonically increases from the centre $x =1 $ of the moduli space, the only fixed point of the $SO (3) $ action, towards the boundary $x =0 $.
The $U (1) $ orbits are circles of  squared radius $R ^2 x ^2 f $ which monotonically decreases from the centre towards the boundary, see Figure \ref{gxfplot}.
Since $ x^2  f( x) \rightarrow 0 $ as $x \rightarrow 0 $, the boundary of the moduli space is a fixed point set of the $U (1) $ action where the circle fibration collapses to zero size.
 \begin{figure}[htbp]
\begin{center}
\includegraphics{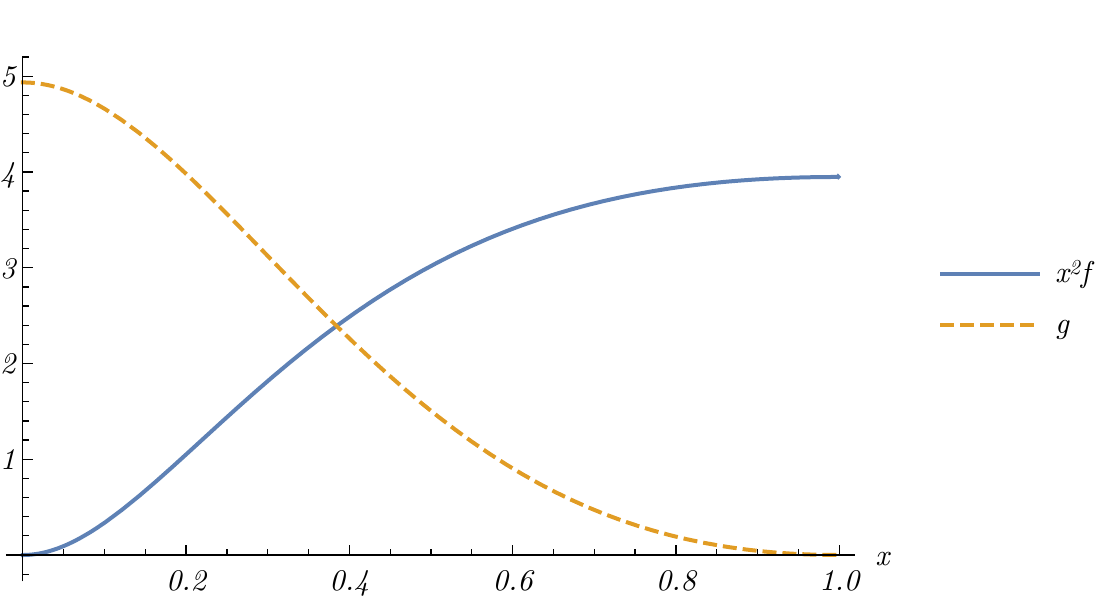}
\caption{Plot of the squared radius (for $R =1 $) of the $SO (3) $ orbits $g $,  and  of the $U (1) $ orbits $x ^2 f$. }
\label{gxfplot}
\end{center}
\end{figure}

We have calculated the scalar curvature $s$ of $\gm $ making use of the Mathematica computer algebra system, see Figure \ref{scplot}.
\begin{figure}[htbp]
\begin{center}
\includegraphics{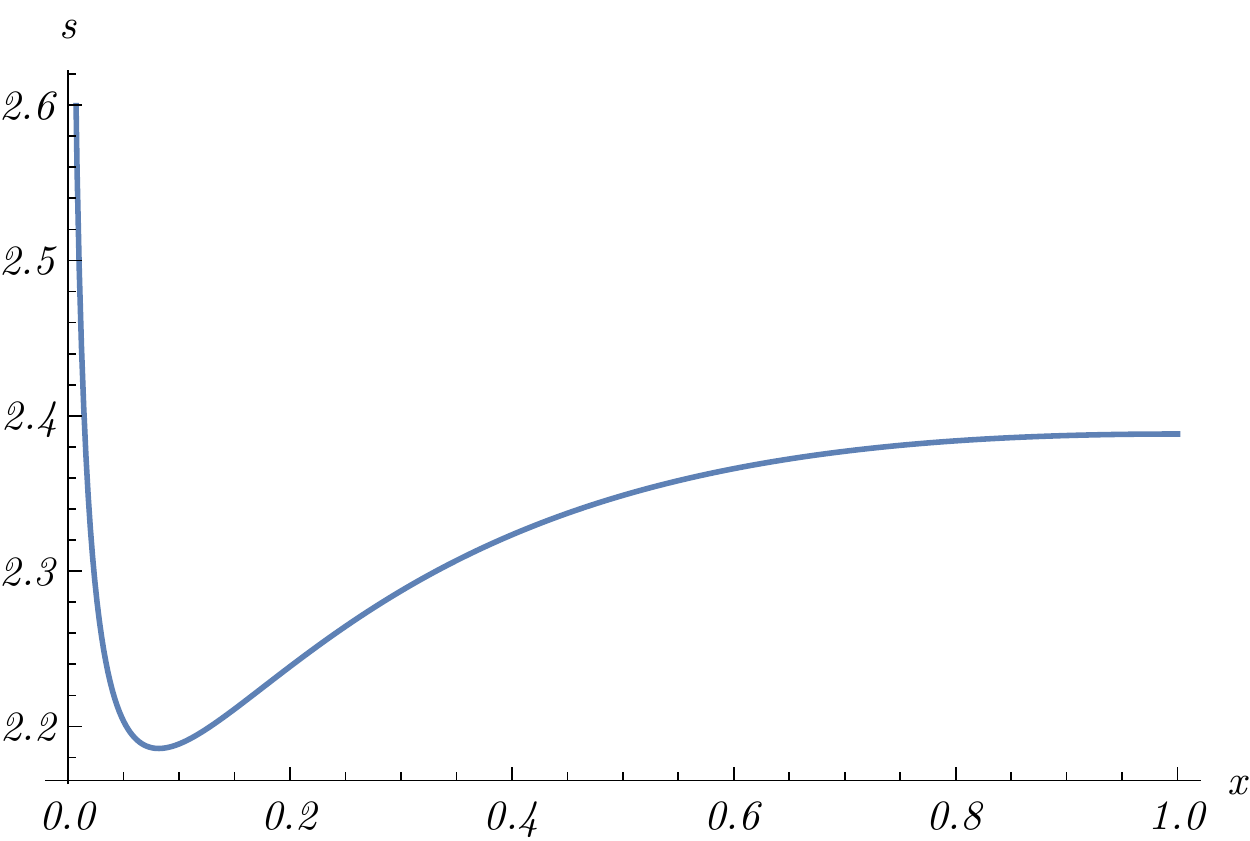}
\caption{Plot of the scalar curvature $s$ of $\gm $.  In the limit $x \rightarrow 0 $ it diverges to plus infinity.}
\label{scplot}
\end{center}
\end{figure}
The function $s (x) $ is everywhere positive  for $x\in(0,1] $ with the value at the standard instanton $x =1 $ being
\begin{equation}
\lim _{ x \rightarrow 1 }s (x) =\frac{165}{7 \pi ^2 }. 
\end{equation} 
It diverges to plus infinity for $x \rightarrow 0 $,  the leading behaviour near $x =0 $ being
\begin{equation}
\label{divofs} 
s(x)= - \frac{6  }{ \pi ^2 } \log x ^2- \frac{10}{ \pi ^2 }
\end{equation} 
plus terms vanishing in the limit $x \rightarrow 0 $.

Let us consider  the $R \rightarrow \infty $ limit of  $\gm $. 
By (\ref{difpar}) we have
\begin{equation}
R ^2 (\mathrm{d} \alpha ^2 + \sin ^2 \alpha  \, \mathrm{d} \beta ^2 )
= \frac{4 R ^4 }{(R ^2 + m _x ^2 + m _y  ^2 ) ^2 } \left(  \mathrm{d} m _x ^2 + \mathrm{d} m _y  ^2 \right) 
\end{equation} 
which in the limit $R \rightarrow \infty $ becomes four times the Euclidean metric on $\mathbb{R}  ^2 $
\begin{equation}
g _{ \mathbb{R}  ^2 }
=\mathrm{d} m _x ^2 + \mathrm{d} m _y ^2.
\end{equation} 
Hence  using (\ref{aaaab1}), (\ref{aaaab2})  we have
\begin{equation} 
\label{rinflinit} 
\gm \xrightarrow[ R \rightarrow \infty ]{} 4 \pi ^2 (\mathrm{d} \lambda ^2 + (\lambda ^2/2) \mathrm{d} \chi  ^2  )
+ 2\pi ^2  g _{ \mathbb{R}  ^2  }.
\end{equation} 

We would like to compare (\ref{rinflinit}) with the metric on the moduli space of circle-invariant instantons on $\mathbb{H}   \simeq  \mathbb{R}  ^4 $. The metric on the moduli space of instantons over flat $\mathbb{R}  ^4  $ is  the flat metric on $\mathbb{R}  ^4 \times (\mathbb{R}  ^4) ^\ast  / \mathbb{Z}  _2  $,   with $ (\mathbb{R}  ^4) ^\ast =\mathbb{R}  ^4 \setminus \{ 0 \}  $ \cite{Manton:2004tk}. The first factor parameterises the instanton centre on $\mathbb{H}  $, the second combines the instanton scale on $\mathbb{H}  $ and the $S ^3 / \mathbb{Z}  _2 $ coming from the structure group. The group  $\mathbb{Z}  _2 $  acts on the second factor as a reflection about the origin. 

For circle-invariant instantons,  the instanton centre has to lie in the plane fixed by the rotation, and bundle automorphisms have to preserve circle invariance. 
We obtain therefore the flat metric on $ \mathbb{R}  ^2 \times (\mathbb{R}  ^2 ) ^\ast  $, in agreement with (\ref{rinflinit}).\footnote{Note, however, that since the curvature of any translation-invariant self-dual connection has infinite $L ^2 $ norm, Euclidean monopoles cannot be obtained from instantons on $ \mathbb{H}  $.}
In order to check that the numerical coefficients multiplying the flat metrics on the two $\mathbb{R}  ^2 $ factors  also agree, below we compute the metric on the moduli space of circle-invariant instantons over $\mathbb{H}  $.
\begin{thm}
The $L ^2 $ metric on the moduli space of $Sp (1) $ circle invariant instantons over $\mathbb{H}  $ is
\begin{equation}
\label{dasjhg} 
4 \pi ^2( \mathrm{d} \nu ^2  + (\nu ^2/4)\,  \mathrm{d} u  ^2  ) 
+ 2 \pi ^2  g _{ \mathbb{R}  ^2 },
\end{equation} 
with $u \in[0, 2\pi )$.
\end{thm}
\begin{proof} 
As we saw in Section \ref{s1}, in the limit $R \rightarrow \infty $ the $( \nu, n) $ and $( \lambda ,m )$ parameterisations become equivalent, so we are free to use the $( \nu , n ) $ parameterisation (\ref{instnn}) which is more convenient for instantons on $\mathbb{H}  $.
The circle-invariant instanton of centre $n$ and scale $\nu$ on $\mathbb{H}  $ is 
\begin{equation}
A _{ \nu ,n  }
=\frac{\Im \left[ (\bar q - \, \bar n) \, \mathrm{d} q  \right] }{\nu ^2 + |q - n |^2 },
\end{equation} 
with $n \in \mathbb{C}  $.
Since $A _{ \nu ,0 } $ has the same form as $\hat{A} _\nu  $, for the modulus $ \nu   $ we just need to repeat the calculations that we did for $P ^1 (\mathbb{H}  )$ but without multiplying by the conformal factor $R ^4 / ( R ^2 + |q| ^2) $. We obtain
\begin{equation}
|| \delta _\nu   A _{ \nu ,0 } || _{ \mathbb{H}  }^2 
= 12 \nu  ^2   \int _{\mathbb{H}}  \frac{ |q |^2  }{\left( |q |^2 + \nu  ^2   \right) ^4} \,  \mathrm{vol } _{\mathbb{H}  }
=4 \pi ^2 .
\end{equation}

For the framed moduli we need to consider the $L ^2 $ solutions of (\ref{orthg1}) which commute with the circle action on $\mathbb{H}  $. By translational symmetry of $\mathbb{H}  $, it is enough to do so for $n =0 $.
 Denote by $\An =  A _{ \nu ,0 } $.
 It is convenient to work in singular gauge, obtained  via the gauge transformation generated by $g = \bar q/ |q| $,
\begin{equation}
\label{als} 
\Ans
=  g ^{-1} \An g + g ^{-1} \, \mathrm{d} g 
=\frac{ \nu  ^2 \Im ( q\, \mathrm{d} \bar q )}{|q| ^2 (\nu  ^2 + |q| ^2 )}.
\end{equation} 
In order for  $\exp( \bar\Lambda ) $ to commute with the circle action on $\mathbb{H}  $, we take the ansatz $\bar \Lambda  = a (r) \qi $, with $r =|q| $. Equation (\ref{orthg1}) with respect to the global gauge potential (\ref{als}) becomes then the ODE
\begin{equation}
a ^{ \prime\prime} + \frac{3}{r} a ^\prime - \frac{8 \lambda ^4\, a  }{r ^2 (\nu  ^2 + r ^2  )^2 } =0,
\end{equation} 
which has the normalisable solution
\begin{equation}
a = \frac{r ^2}{ r ^2 + \nu  ^2 }  \, u ,
\end{equation} 
where $u \in \mathbb{R}  $ is an arbitrary constant. Therefore
\begin{equation}
\label{lambdasol} 
\bar \Lambda _\nu  
= \frac{r ^2}{ r ^2 + \nu  ^2 }  u \,  \qi   .
\end{equation} 
Its squared norm is
\begin{equation}
| |\delta _{ u  } \mathrm{d} _{\Ans } \bar \Lambda _\nu  ||^2 _{ \mathbb{H}  }
=4 \pi ^2 \nu ^2.
\end{equation} 
The range of $\nu $  in the parameterisation (\ref{instnn}) is $ (0, \infty )$, while the range of $\lambda$  
 in the parameterisation (\ref{ppp2}) is $(0, R )$, but the two ranges agree in the limit $ R \rightarrow \infty $.

For the translational moduli $n_1 , n _2  $ we have e.g.~$ \delta _{ n _1 } A _{ \nu ,n  }= - \partial _1  A _{ \nu ,n  } $, which is not orthogonal to the gauge group orbits. However
\begin{equation} 
\pg{( \delta _{ n _1 }A _{ \nu ,n  }  )}
= - \partial _1  A _{ \nu ,n  } + \mathrm{d} _{  A _{ \nu ,n  }  }  (A _{ \nu ,n  } (\partial _1 ))
= - \iota _{ 1 } F _{ \nu ,n},
\end{equation} 
where $\iota _1 $ denotes the insertion of the vector field $\partial _1 $, satisfies by virtue of self-duality and of the Bianchi identity the orthogonality condition  $\mathrm{d} _{ A _{ \nu ,n  }} ^\dagger \pg{( \delta _{ n _1 } A _{ \nu ,n  }  )} =0 $.
Since
\begin{equation}
\begin{split} 
\iota _1 (F_{ \nu ,n}) &
=\frac{2\nu ^2 }{(\nu ^2  + |q| ^2  )^2 }  \Im (\mathrm{d} q ), \quad 
\iota _1 (F_{ \nu ,n})  \wedge \hat* \iota _1 (F_{ \nu ,n}) 
=3 \cdot \frac{4\nu ^4 }{(\nu ^2  + |q| ^2  )^4 }  ,
\end{split} 
\end{equation} 
 integrating over $\mathbb{H} $ we get
\begin{equation}
||\iota _1 (F_{ \nu ,n}) || ^2 _{ \mathbb{H}  }
= 12\nu ^4 \cdot 2 \pi ^2 \int _0 ^\infty \frac{r ^3 }{(\nu ^2 + r ^2 )^4 } \mathrm{d} r
=2 \pi ^2 .
\end{equation} 

Therefore,  the metric on the moduli space of circle-invariant 1-instantons over $\mathbb{H} $ is
\begin{equation}
4 \pi ^2( \mathrm{d} \nu ^2  + (\nu ^2/4)\,  \mathrm{d} u  ^2  ) 
+ 2 \pi ^2  g _{ \mathbb{R}  ^2 },
\end{equation} 
with the factor $1/4 $ arising because of the $\mathbb{Z}  _2 $ quotient.

\end{proof} 

The result is  in agreement with (\ref{rinflinit}) apart from the numerical factor of $1/4 $ (instead of $1/2 $) multiplying $\mathrm{d} u ^2 $. The fact that there is a discrepancy is not surprising, as in the limit $R \rightarrow \infty $ the norm of the Higgs field $\phi$ vanishes, so that $\phi$ is not a good infinitesimal automorphism anymore. In fact, the Euclidean limit of hyperbolic monopoles is subtle, and it is obtained by allowing the weight of the lifted circle action to also diverge \cite{Atiyah:1987ua,Jarvis:1997ws}.
However, the reason why in the $R \rightarrow \infty $ limit  $\gm $ and  (\ref{dasjhg})   differ by exactly a factor $1/2 $ in their framed parts is not clear to us.

\subsection{ Geodesic motion on $\MMMf $}
Because of spherical symmetry, we can assume without loss of generality that motion is taking place in the equatorial plane $\alpha =\pi /2 $.
The conservation laws associated to the Killing vector fields $ \partial _\chi $ and $\partial _ \beta  $ are,
with $x = \lambda /R\in(0,1] $,
\begin{equation}
\label{claws} 
f x ^2 \dot{ \chi } =C, \qquad  g \sin ^2 (\alpha) \, \dot{ \beta  }=D ,
\end{equation} 
where $\dot{}$ denotes differentiation along an affinely parameterised geodesic.

The constants $D$ and $C$ have the physical meaning of angular momentum with respect to, respectively,
space and phase rotations. The quantities $g$,  $ x ^2 f $ can be thought of as the corresponding moments of inertia. Note that $x ^2 f $ is an increasing function of $x$, while $g $ is a decreasing function, see Figure \ref{gxfplot}.
Since the size of an instanton increases with $x$, the moment of inertia with respect to phase rotations grows as customary with the size of the instanton, while the moment of inertia with respect to space rotations behaves the opposite way.

For an affinely parameterised geodesic,  the conservation laws (\ref{claws}) give
\begin{equation}
 \dot x ^2 + \frac{C  ^2 }{f ^2 x ^2 } + \frac{D ^2 }{fg} - \frac{1}{f} =0,
\end{equation} 
corresponding to a 1-dimensional motion of a particle of mass $m =2 $ with zero energy in the effective potential 
\begin{equation}
\label{ripot} 
V = \frac{1}{f} \left( \frac{C  ^2 }{f  x ^2 } + \frac{D ^2 }{g} - 1 \right).
\end{equation} 
A plot of $V$ for generic values of $C$ and $D$ can be found in Figure \ref{vcd}.

\begin{figure}[htb]
\begin{center}
\includegraphics[scale=0.6]{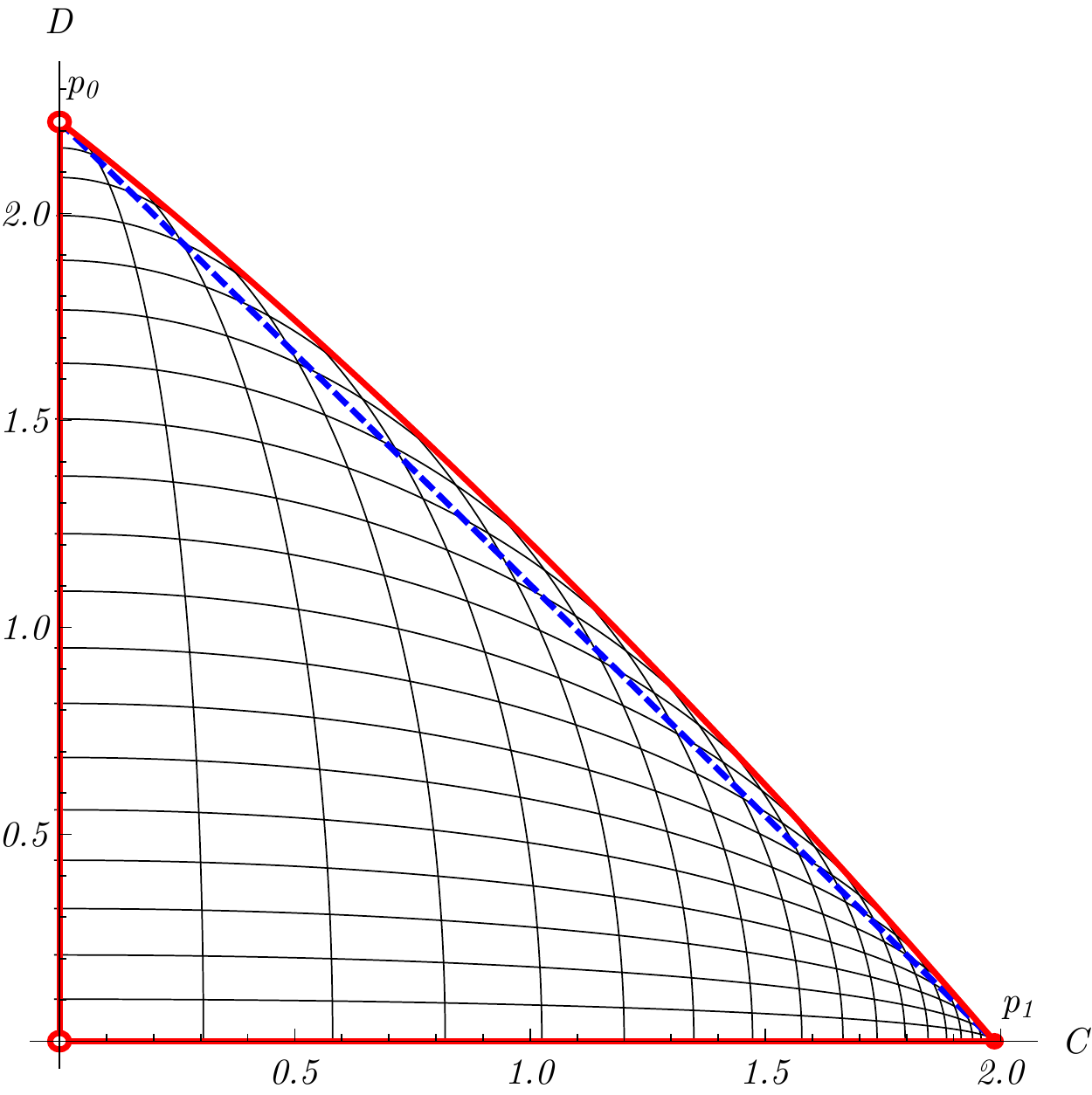}
\caption{Values of $C,D$ for which $V$ has at least one root if $x\in[0,1]$. The black curves are  plots of the ellipses 
$ C ^2 /(f  x ^2 ) + D ^2 / g=1 $ for various values of $x\in(0,1)$. Through any point $(C,D) $ not on the boundary  there pass two ellipses, corresponding to the two roots of $V$ for the given values of $C, D$.  For  values of the parameters corresponding to points on the boundary $V$ has only one root. See the text for more details.
  }
\label{mspace}
\end{center}
\end{figure}
Motion is only possible in the region where $V\leq 0 $. Points where $V$ vanishes are inversion points for the $\lambda$-motion.
The potential $V$  diverges to plus infinity for $x \rightarrow 0 $ unless $C =0 $, and for $x  \rightarrow 1 $ unless $ D = 0 $. Therefore $x =0 $ is a root of $V$ if and only if $C =0 $ and  $ D ^2  =  g(0) $. Similarly $x =1 $ is a root if and only if $ D = 0 $ and  $C ^2 = f (1) $.
Hence a non-zero angular momentum with respect to phase rotations  prevents the instanton from shrinking to zero size ($x = 0 $), while a non-zero angular momentum with respect to  space rotations prevents it from reaching maximal size ($ x =1$). 
For any fixed value of $x $, the region in the $(C,D) $-plane in which $V$ has at least one root is bounded by the ellipse with semiaxes $ x \sqrt{ f} $, $\sqrt{ g } $. A numerical approximation of   the union of these regions as $x$ varies in $(0,1) $ is depicted in Figure \ref{mspace}.
Since $C$, $D$ only appear quadratically in (\ref{ripot}) we can focus on the case $C \geq 0 $, $ D  \geq 0 $.

Consider first the case in which either $C$ or $D$ vanishes. For $C =0 $, $0<D \leq \sqrt{\mathrm{max}_{ x\in[0,1] }g(x)} =\sqrt{g(0)} =\pi / \sqrt{ 2 } \sim  2.22 $,  $V$ is an increasing function having only one root which can take any value in the interval $[0,1) $.  For $D =0 $, $0<C \leq \sqrt{\mathrm{max}_{ x\in[0,1] } x ^2 f(x)} =\sqrt{f(1)} =\sqrt{ 2/5 }\, \pi  \sim  1.99 $,  $V$ is a decreasing function with only one root which can take any value in the interval $(0,1] $. For $C = D =0 $, $x \in[0,1] $, $V$ is everywhere negative. 

\begin{figure}[H]
\begin{center}
\includegraphics[scale=0.6]{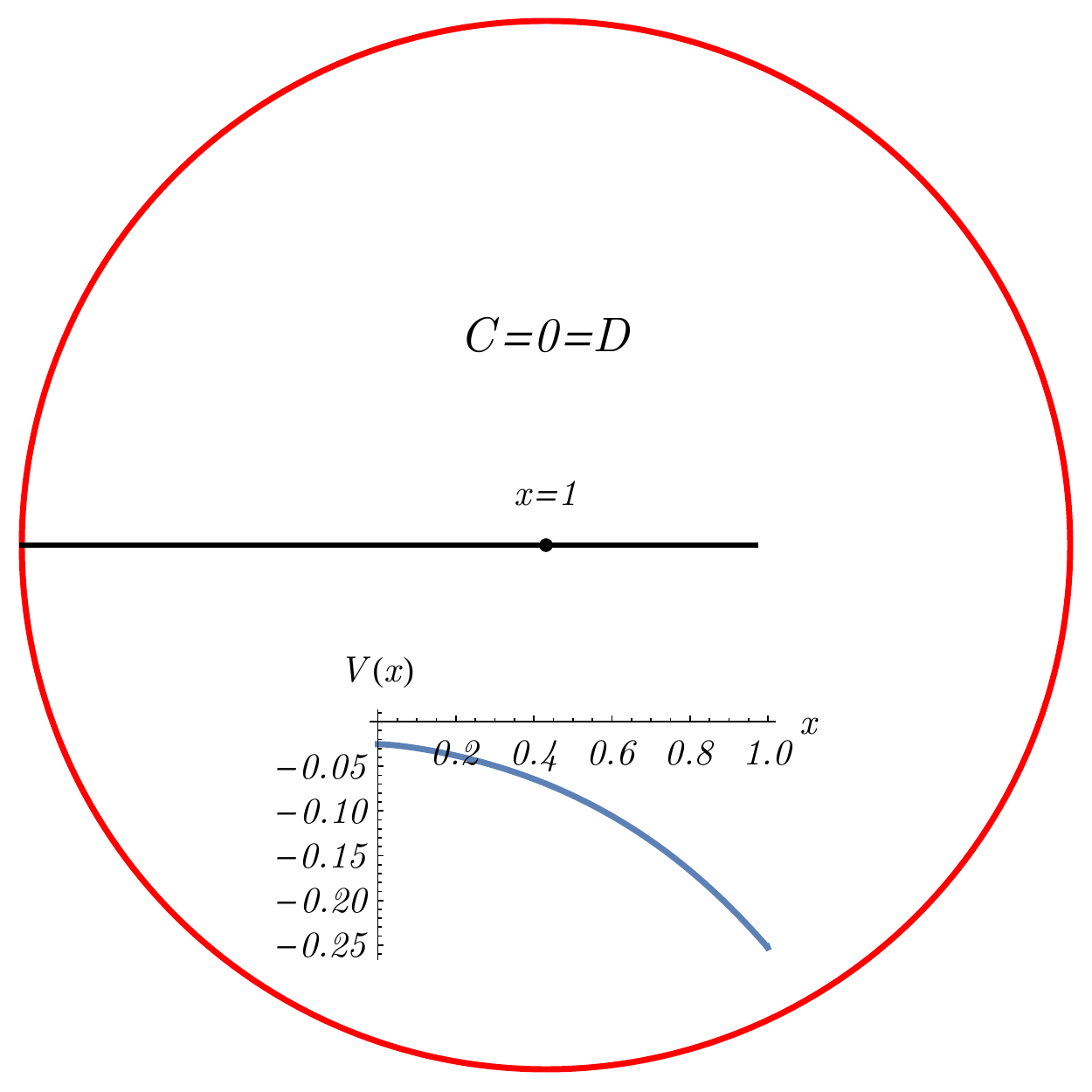}
\caption{Example of motion for $C =D =0 $. The inset shows the corresponding potential $V$. The black line shows a possible motion in which the instanton grows in size until it reaches its maximum possible size ($x=1$), and then shrinks to a singular ``small'' instanton infinitely concentrated at the point antipodal to its initial centre.  For $C =D =0 $ the general trajectory is part of a diameter with an  endpoint on the boundary circle.
}
\label{plotc0d0}
\end{center}
\end{figure}

For $C \neq 0 \neq D $, $V$ has at most two roots: The functions $1/(f \lambda ^2) $ and $1/g$ are convex,  and a linear combination with positive coefficients of convex functions is convex. Since a convex function cannot take the same value more than twice,  the quantity $ C ^2 /(f  x ^2 ) + D ^2 / g $ has at most two roots.  For generic values of $C$ and $D$, $V$ has  two roots, corresponding  in Figure \ref{mspace} to the two ellipses passing through any point not on the boundary, or no roots. The limiting case in which $V$ has two coincident roots corresponds to points $(C, D)$ lying on the envelope of the family of ellipses with semiaxes  $ x \sqrt{ f }$, $ \sqrt{ g } $ parameterised by $x\in(0,1)$. Such points solve the system of equations
\begin{equation}
\frac{C ^2 }{ x ^2 f} + \frac{D ^2 }{g} =1, \qquad
\partial _x  \left( \frac{C ^2 }{ x ^2 f} + \frac{D ^2 }{g} \right)  =0.
\end{equation} 

The boundary values of $C $, $D$  depicted in red in Figure \ref{mspace} can be described as follows.
At $p _1 = ( \sqrt{ f (1) },0) $, $V$  has root $x =1 $, which is part of the moduli space. At $p _0 = ( \sqrt{ f (1) },0) $, $V$ has root $x =0$, which is not part of the moduli space. Along the boundary curve from $p _0 $ to $p _1 $ the single root of $V$ takes all values in  $[0,1] $. For all the points on this curve  it has multiplicity two, except at the endpoints $p _0, p _1 $ where it has multiplicity one. Along the oriented line from $ p _0 $ (respectively $p _1 $) to the origin $V$  has a single root with multiplicity one which moves from $ x =0 $ at $p _0 $  (respectively $x =1 $ at $p _1 $) to arbitrarily close to $x =1 $ (respectively $x =0 $) at the origin. 
 ``Frustration'' at the origin is avoided since $V$ has no roots for $C =D =0 $. The curve between $p _0 $ and $p _1 $ is not far from being linear. For comparison, the straight line segment $\overline{ p _0 p _1} $ is shown in Figure \ref{mspace} as a blue dashed line.

\begin{figure}[H]
\begin{center}
\includegraphics[scale=0.6]{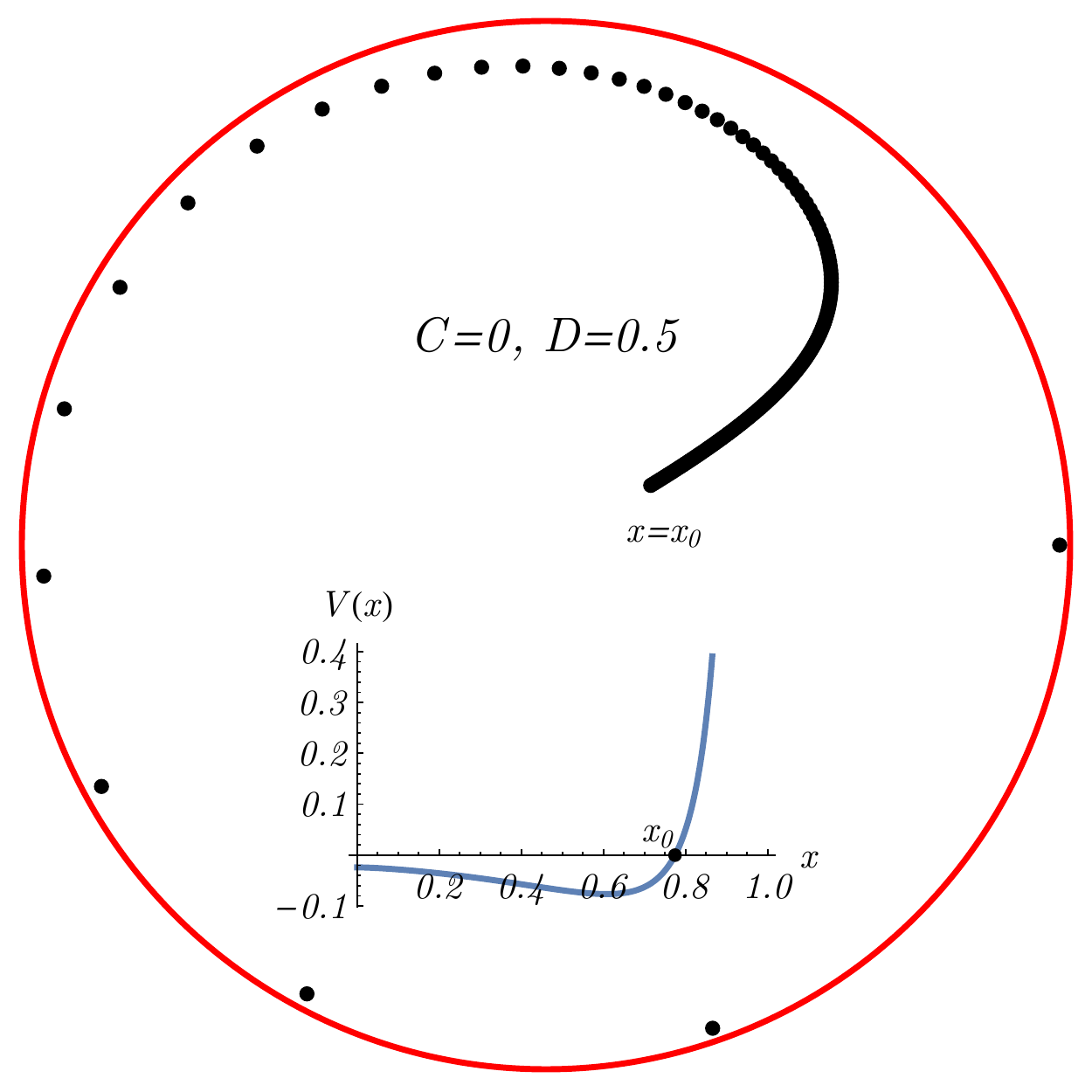}
\caption{Example of motion for $C =0$, $D =0.5$. The black dots, plotted for uniformly spaced values of $\lambda$,  show a possible motion starting at the turning point $x _0 $ and ending on the boundary of the moduli space.}
\label{plotc0}
\end{center}
\end{figure}

Next we discuss the instanton motion corresponding to these different cases.
In general, it follows from (\ref{claws}) that $\dot\beta  $ increases from the centre ($ x =1 $) to the boundary ($ x =0$) of the moduli space, while $\dot \chi $ behaves in the opposite way. Figures \ref{plotc0d0}, \ref{plotc0}, \ref{plotd0}, \ref{plotcd}
show the image of the curve
\begin{equation}
t \mapsto (1 - x(t)) ( \cos \beta (t)  , \sin \beta (t)  )
\end{equation} 
on the the equatorial section $ \alpha =\pi /2 $ of  $\MMM \simeq \mathring{B} ^3 _R $ for various kinds of geodesic motion. The $\chi$ fibration is not shown. The curve has been obtained by numerical integration. Points on the boundary $x =0 $ of $\MMM $  are represented in red.
\begin{figure}[ht]
\begin{center}
\includegraphics[scale=0.6]{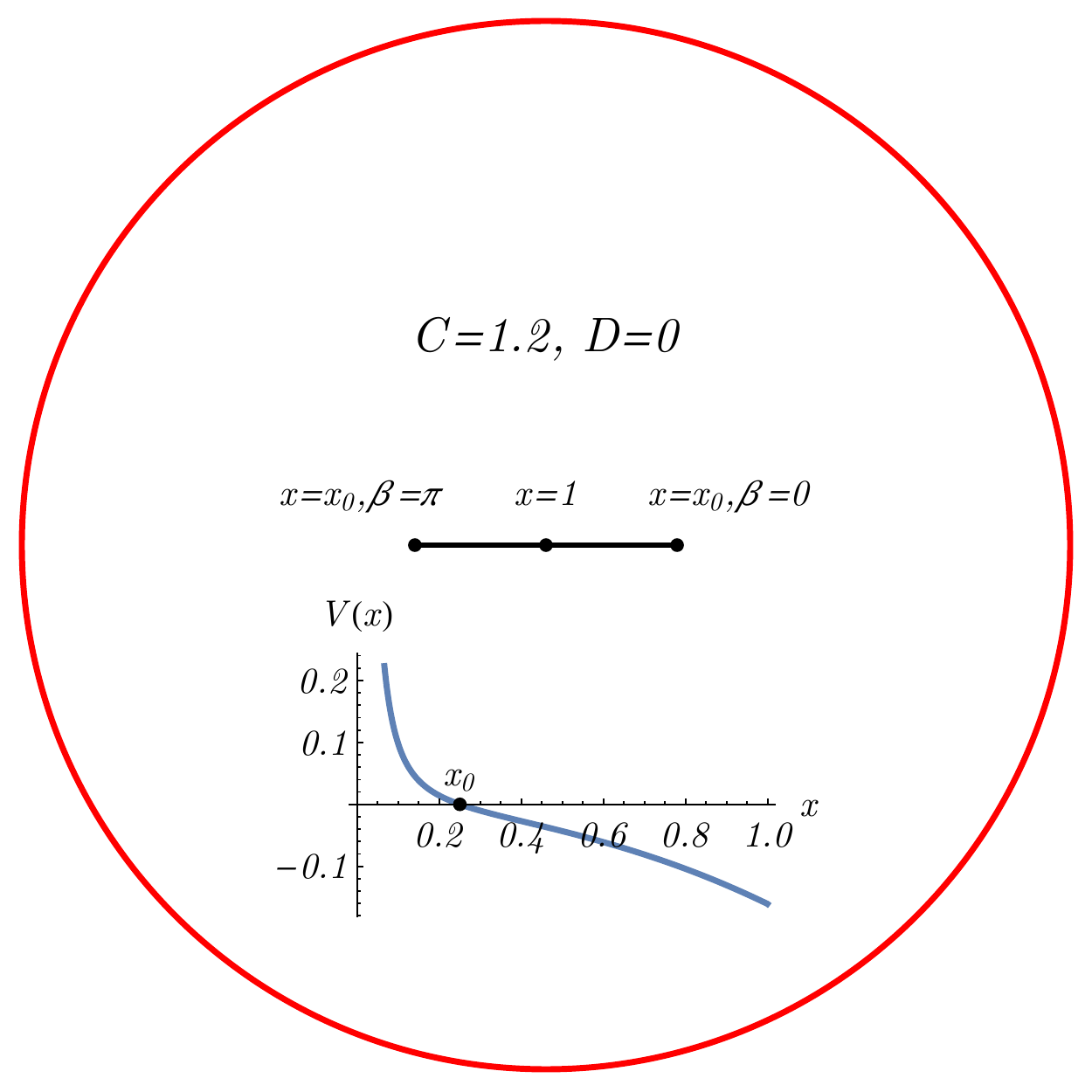}
\caption{Example of motion for $ C =1.2, D =0 $.
The instanton oscillates in size between its maximum value $x =1$ and its minimum $x =x _0 $, with its centre alternating between the point $ \alpha =\pi /2$, $\beta $ and its antipodal. The plot in the figure is for $\beta = 0 $. }
\label{plotd0}
\end{center}
\end{figure}
\begin{figure}[htbp]
\begin{center}
\includegraphics[scale=0.6]{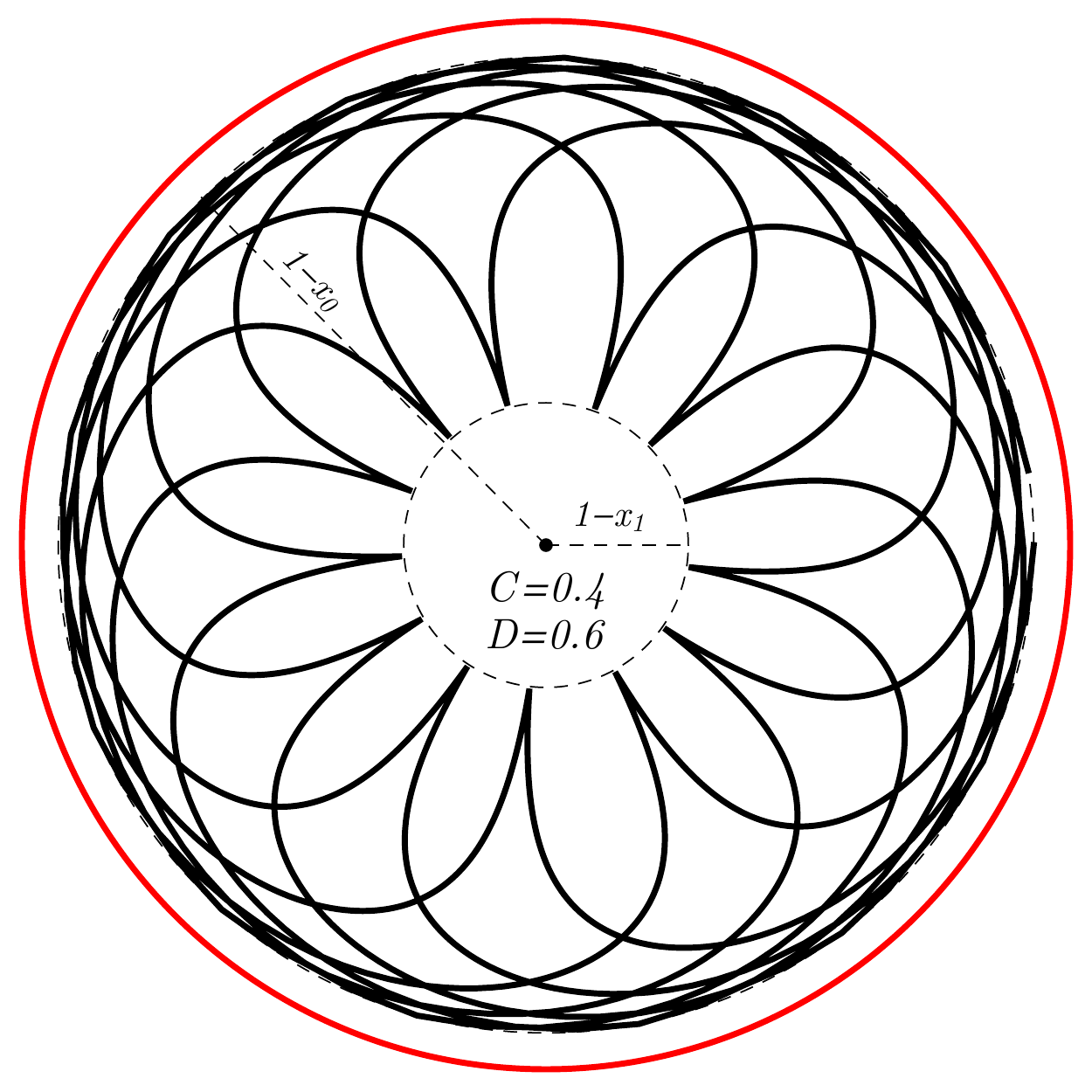}
\caption{Top: Example of motion for $ C =1.2$, $D =0.6 $.
The instanton   size oscillates between the values $x _0 $ and $x _1 $ while  $ \beta $ is monotonically increasing.
The resulting curve is bounded by the two circles of radii $1 - x _1 $ and $1 - x _0 $. It is generally not closed --- in the figure thirteen periods are shown.}
\label{plotcd}
\end{center}
\end{figure}

For $C = D =0 $, we have incomplete geodesics with the instanton shrinking to zero size ($x =0$) in finite time, see Figure \ref{plotc0d0}. The angle $\chi$ is constant for $x \in(0,1] $. It becomes undefined at the singular boundary  $ x =0 $ where the circle fibration collapses to  zero size. The angle $\beta $  is constant but for possibly changing value at $x =1 $, where $g$ vanishes. This happens if the size of the instanton is initially increasing. As $x$ reaches its maximum value $x =1 $ the instanton centre jumps to its antipodal point, i.e.~$(\alpha = \pi /2 , \beta )\rightarrow  (\pi - \alpha =\alpha , \pi + \beta )$. From that moment onwards, $x$ decreases as the instanton shrinks, eventually reaching zero size, while $\beta$ stays constant.

For $C =0$, $D \neq 0$ we also have incomplete geodesics with the instanton size $x$  decreasing from the value of the only root of $V$ to $x =0 $ in finite time, see Figure \ref{plotc0}. The angle $\chi$ is constant, while $\dot\beta $ has constant sign equal to the sign of $D$. Therefore, the instanton hits the boundary of moduli space at an angle which is neither radial nor tangential, moving along a spiral.  The limiting case in which the   only root of $V$ is $x =0 $  corresponds to an instanton with zero size, undefined  $\chi$ and $\beta $ linearly varying  in time as $\beta  (t) =\beta (0) + (D/ g (0) )\, t $, but does not belong to the moduli space. 

\begin{figure}[htbp]
\begin{center}
\includegraphics[scale=0.6]{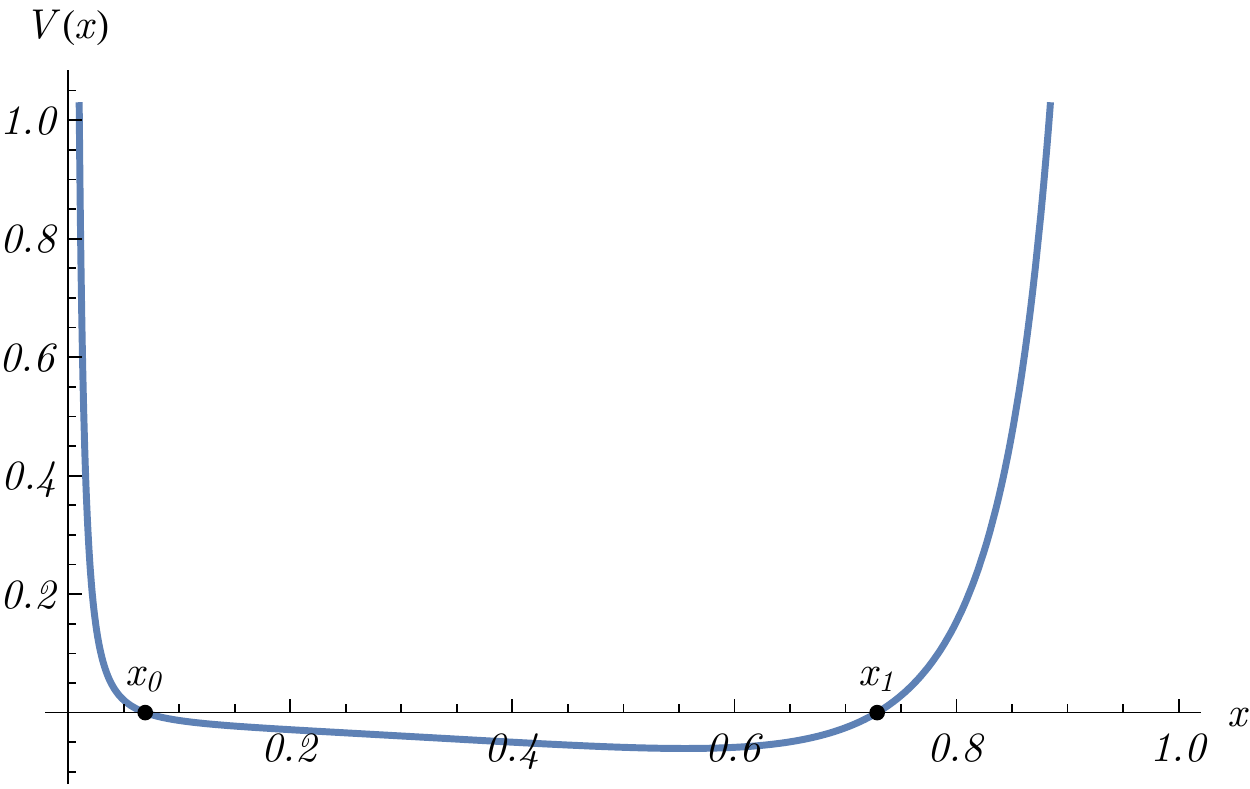}
\caption{Graph of $V$  for the values $C =0.4 $, $D =0.6 $ corresponding to Figure \ref{plotcd}. The turning points at $x _0, x _1  $ are marked by black dots.}
\label{vcd}
\end{center}
\end{figure}

For $D = 0$, $C \neq 0 $  we have bounded orbits, see Figure \ref{plotd0}. Starting from its minimum value $x  _0  $, equal to the only root of $V$, the size $x$ of the instanton increases to its maximum value $x =1 $, and then decreases again to $x _0  $, with $ \beta $ flipping by $\pi$ as $x$ reaches the value of one and staying constant otherwise. The quantity $ \dot  \chi $ has constant sign equal to the sign of $C$. In the limiting case $ x  _0  =1 $ we have an instanton of constant size,  undefined $\beta  $ (the instanton is sitting at the centre of the moduli space hence the polar angle is not defined) and $\chi $  varying linearly in time as 
$ \chi (t) = \chi (0) + (C/ f (1) ) \, t $.

For $ C \neq 0 \neq D $ we also have bounded,  generally not closed orbits, with $x$ bouncing between the values $x _0, x _1 \in(0,1)$ of the two roots of $V$, see Figure \ref{plotcd}. The frequency of the motion increases as the area of the region bounded by the graph of $V$ and the $x$-axis decreases. Both $\beta $ and $\chi$ are varying with $\dot \beta, \dot \chi $ having constant sign. In the limiting case $ x _0 =x _1 $ we have an instanton with constant $x =x _0 $, and $\beta $, $\chi$   varying linearly in time as 
$\beta  (t) =\beta (0) + (D/ g (x_0) )\, t $, $ \chi (t) = \chi (0) + C \, t/ ( x _0 ^2 f (x_0) ) $.

\section{Conclusions}
\label{s5} 
Our parameterisation of circle-invariant  1-instantons on the 4-sphere in terms of the 3-dimensional coset $SL(2,\mathbb{C})/SU(2)$ has a  simple interpretation: suitably chosen parameters on the coset  give the size of the instanton and its position on an equatorial 2-sphere which is kept fixed by the circle action. The instanton size is positive but  at most equal to the radius of the 4-sphere, with equality corresponding to a uniformly spread out instanton.  Framing adds a single internal phase degree of freedom, leading to a 4-dimensional moduli space.

The $L^2$ metric we have computed has the expected $SO(3)\times U(1)$ symmetry, but we have not been able to ascertain any other special feature, like a  K\"ahler structure. However, the symmetry is sufficient to determine  geodesics and to understand their generic features. These turn out to be remarkably similar to those found for degree one  lumps on a 2-sphere in \cite{Speight:1997ub}. 

In fact, such lumps are also characterised by a size parameter taking values in a finite interval, a position on the 2-sphere, and an internal phase. For lumps, the phase is  $SU(2)$-valued, the moduli space six dimensional, and the  isometry group is  $SO(3)\times SO(3)$. Nonetheless, the geodesic behaviour is similar to that found here, 
essentially because of similar properties of the  momenta of inertia with respect to spatial and internal rotations.  In both models, the former decreases and the latter increases with the size of the soliton. This behaviour is the basic reason why, for generic geodesics,
 the  size parameter oscillates while  the solitons spin both in space and phase. 

The circle action and associated Higgs field play an essential r\^ole in the definition of the tangent space and of the  $L ^2 $ metric in the framed directions. It is therefore not obvious how our discussion of framing can be generalised if circle invariance is dropped. Framing is well understood in the case of instantons on $\mathbb{H}  $ and one would expect there to be a compact counterpart also for  non-circle invariant instantons.  However,  using the terminology of the Introduction, we do not know a natural way of selecting generators of `physically relevant' gauge transformations in this case,  and  have left this issue for a  future investigation.

As explained in the Introduction, one motivation for considering circle-invariant instantons  is their relation to hyberbolic monopoles. The $L^2$ metric on the moduli space of hyperbolic monopoles diverges, and it is a long-standing problem to define a natural metric or other geometric structure on this space. Our metric solves this problem at some level, but it does not satisfy all the requirements that have been considered in the literature. 

In particular, there is no obvious limit where our metric tends to the $L^2$ metric on the moduli space of a single Euclidean monopole. Taking the radius  $R$ of the 4-sphere  to infinity  does not work, as we saw  in Section \ref{s4}, since this limit does not deform hyperbolic into Euclidean monopoles. Our metric also does not fit into the framework  of pluricomplex structures  which was proposed in \cite{Bielawski:Euclidean,Bielawski:Pluricomplex}
 as the appropriate  generalisation of hyperK\"ahler structures  for  the moduli spaces of hyperbolic monopoles. 

Nonetheless, the metric we computed is naturally defined in terms of a field theory on the 4-sphere, and the configurations parameterised by our moduli space are bona fide hyperbolic monopoles. The geodesics we have found therefore approximate the motion of hyperbolic monopoles according to a naturally defined flow, namely the time evolution of  Yang-Mills theory on $\mathbb{R} \times S^4 _R $.

\section*{Acknowledgements}
G.F.~thanks Lutz Habermann for useful discussions.
We both thank Sir Michael Atiyah for interesting discussions and acknowledge support through  the EPSRC grant ``Dynamics in geometric models of matter'',
EP/K00848X/1.

\appendix
\section{Quaternionic notation}
\label{qnot} 
We denote by $\mathbb{H}  $ the space of quaternions, and by $\qi, \qj, \qk $ the standard imaginary unit quaternions.
We write a quaternion in the form
\begin{equation}
q=x + y\qi + z\qj + w \qk
=  x ^1 + x ^2 \qi + x ^3 \qj + x ^4  \qk
\end{equation}
and denote quaternionic conjugation by a bar, $\bar q =x-y\qi-z\qj-w\qk $. The imaginary part of a quaternion $q$ is
\begin{equation}
\Im(q)= y \qi + z \qj + w \qk.
\end{equation} 
We   identify $\mathbb{R}  ^4 $ with $\mathbb{H}  $ via the isomorphism
\begin{equation}
\mathbb{R}  ^4 \ni(x ,y,z,w) \leftrightarrow q=x + \qi y + \qj z + \qk w\in \mathbb{H}.
\end{equation} 
We take  the metric 
\begin{equation}
\gh =\mathrm{d} \bar q \, \mathrm{d} q  
\end{equation}
on $\mathbb{H}$, corresponding to the Euclidean metric $g _{ \mathbb{R}  ^4 } = \mathrm{d} x ^2 + \mathrm{d} y ^2 + \mathrm{d} z ^2 + \mathrm{d} w ^2 $ on  $\mathbb{R}  ^4 $. 
The squared modulus of a quaternion  $q$ is therefore $\bar q q = |q| ^2 =x ^2 + y ^2 + z ^2 + w ^2 $.

The isomorphism  $\mathbb{H}  \simeq  \mathbb{C}  \oplus  \mathbb{C}  \qj $ mapping $x + y\qi + z\qj + w\qk $ to $(x + iy, (z + i w) \qj) $ allows us to identify $n \times n $ quaternionic valued matrices with $2n \times 2n $ complex valued matrices. Let $ M =A + B \qj $ be a quaternionic valued matrix. Then the corresponding complex valued matrix $\widetilde  M $ is given by
\begin{equation}
\widetilde  M =
\left( 
\begin{array}{cc}
A  & B \\ 
-\bar B &\bar A,
\end{array}
 \right) 
\end{equation} 
where $\bar A $ is the conjugate transpose of $A$.
The conjugate of a quaternion corresponds to the transpose conjugate of its associated matrix.

We will be dealing with  the following groups
\begin{align}
\label{sl2h} 
SL(2, \mathbb{H}  )&=\left\{M=\begin{pmatrix} a &b \\ c &d \end{pmatrix} : a,b,c,d\in \mathbb{H} , \mathrm{det}(\widetilde M) =1  \right\},\\
Sp(2 )&=\left\{\begin{pmatrix} a &b \\ c &d \end{pmatrix} : a,b,c,d\in \mathbb{H} :|a |^2 + |c |^2 =1= |b |^2 + |d| ^2, \bar a b +  \bar c d= 0 \right\},\\
Sp(1)  &= \left\{ u\in \mathbb{H}  : |u| ^2 =1   \right\}.
\end{align} 
The group $Sp (2)  $ is the subgroup of $SL(2, \mathbb{H}  )$  preserving the quaternionic inner product $\bar q  q   $. The group $Sp (1) $ is isomorphic to $SU (2) $,  $\widetilde{Sp} (1) =SU (2) $.

For the unit quaternions $1,\qi, \qj, \qk $ the above correspondence gives
\begin{equation}
\label{ccccr} 
\widetilde 1 = \mathrm{Id}_2   , \quad 
\widetilde\qi = i \sigma _3 , \quad 
\widetilde\qj = i \sigma _2 , \quad 
\widetilde\qk = i \sigma _1,
\end{equation} 
where $\{ \sigma _i \} $ are the Pauli matrices.
Consider a purely imaginary quaternion $q =y\qi + z\qj + w\qk $. Using  (\ref{ccccr}) we can identify it with an element $\widetilde q $ of $\mathfrak{su} (2)  $.
By taking the inner product 
\begin{equation}
\langle A, B \rangle _{ \mathfrak{su} (2) } = - \frac{1}{2}  \mathrm{Tr}(AB) ,
\end{equation} 
  we have $ \langle \bar q, q \rangle _{ \mathfrak{su} (2) } =y ^2 + z ^2 + w ^2 =\bar q q$ consistently with our previous definition.
  
  If $\alpha , \beta  \in \Lambda (\mathbb{H}  , \mathfrak{sp} (1)) $, we denote by
\begin{equation}
\begin{split} 
\fp{ \alpha}{ \beta }&
=  - \frac{1}{2} \mathrm{Tr} \left( \alpha \wedge \hat * \beta  \right), \qquad 
| \alpha | ^2 _{ \mathbb{H}  }
=  \fp{ \alpha }{ \alpha },
\end{split} 
\end{equation} 
where $\hat * $ denotes the Hodge operator with respect to the metric $\gh $.
For quaternionic valued forms $\alpha$, $\beta$  of degree $p$ and $q$ we have $\overline{ \alpha \wedge \beta }= (-1 )^{ pq }\bar \beta \wedge \bar \alpha $.

Some useful relations are
\begin{equation}
\begin{split} 
\Im (\bar q\,  \mathrm{d} q ) &=
(x \, \mathrm{d} y-y \, \mathrm{d} x + w \, \mathrm{d} z- z \, \mathrm{d} w )\, \qi + 
(x \, \mathrm{d} z - z \, \mathrm{d} x + y \, \mathrm{d} w - w \, \mathrm{d} y )\, \qj \\ &+ 
(x \, \mathrm{d} w-w \, \mathrm{d} x + z \, \mathrm{d} y - y \, \mathrm{d} z )\, \qk,\\
\Im ( q\,  \mathrm{d} \bar q ) &=
(y \, \mathrm{d} x - x \, \mathrm{d} y+ w \, \mathrm{d} z- z \, \mathrm{d} w )\, \qi + 
( z \, \mathrm{d} x-x \, \mathrm{d} z + y \, \mathrm{d} w - w \, \mathrm{d} y )\, \qj \\ &+ 
(w \, \mathrm{d} x - x \, \mathrm{d} w + z \, \mathrm{d} y - y \, \mathrm{d} z )\, \qk,\\
|\Im (\bar q\,  \mathrm{d} q ) |^2 _{ \mathbb{H}  }&=
 3 |q |^2 ,\\
 |\mathrm{d}\bar q \wedge \mathrm{d} q |_{ \mathbb{H}  }^2  &
= 4! .
\end{split} 
\end{equation}

\section{Stereographic projection}
\label{stp} 

In order to be able to take the flat space limit $R \rightarrow \infty $ of various quantities, we consider stereographic projection from a sphere of arbitrary radius $R$.
Embed $S ^4 _R  $ in $ \mathbb{H} \times \mathbb{R} $ as $S ^4 _R =\{( q , v) : |q |^2 + v^2 =R ^2 \}  $. 
Stereographic projection from the north pole  $p ^N=(0,R) $ of $ S ^4 _R $  is then given by the map
\begin{equation}
\label{st1} 
\begin{split} 
\phi _N : S ^4 _R\setminus\{p _N \}  \rightarrow \mathbb{H}  , \quad
(q , v ) &\mapsto  \frac{R }{R - v } \, q, \quad \text{$v \neq R $},
  \end{split} 
\end{equation} 
with inverse
\begin{equation}
\begin{split}
\phi _N ^{-1} :\mathbb{H}  \rightarrow S ^4 _R\setminus\{p _N \} , \quad 
q &\mapsto \frac{R}{R ^2 + |q |^2 } (2 R q, |q |^2 - R ^2 ).
\end{split}
\end{equation} 
The map $\phi _N $ can be extended to a conformal isometry from $S ^4 _R $ to $\hat{\mathbb{H}  } $, the one-point compactification of $\mathbb{H}  $,
by setting $p ^N =(0,R)  \xmapsto{ \phi _N } \infty $.

We identify $S ^4 _R $ with the quaternionic projective space $P ^1 (\mathbb{H}  ) $,
 the quotient space of $ S ^7 _{ \sqrt{ R } }=\{ (q ^1 , q ^2)\in \mathbb{H}  ^2 : |q ^1 |^2 + |q ^2 |^2 =R \} \subset \mathbb{H}  ^2 $ under right multiplication by unit quaternions,
\begin{equation}
\begin{split} 
&P ^1 (\mathbb{H}  ) 
=\{ [ q ^1 , q ^2 ]:(q ^1 , q ^2 )\in  S ^7 _{ \sqrt{ R } }\}, \\
& [q ^1 , q ^2 ]= [ r ^1 , r ^2 ] \Leftrightarrow ( r ^1 , r ^2 ) =( q ^1 h, q ^2 h ) \text{ for some } h\in \mathbb{H} ,\,  |h |=1.
 \end{split} 
\end{equation} 
The Hopf projection is given by
\begin{equation}
\pi _H : S ^7 _{ \sqrt{ R }} \subset\mathbb{H}  ^2  \rightarrow {S}  ^4 _R\subset \mathbb{H}  \times \mathbb{R}, \quad 
( q ^1 , q ^2 )\mapsto (2 q ^1 \bar q ^2,  |q ^1 |^2 - |q ^2 |^2).
\end{equation} 
Note that for the target space of $\pi _H $ to be the 4-sphere of radius $R$, we need to take its domain to be the 7-sphere of radius $\sqrt{ R }$.
The Hopf projection descends to an isomorphism $\hat\pi _H : P ^1 (\mathbb{H}  )\rightarrow   S ^4 _R$ which we use to identify $S ^4 _R $ with $P ^1 (\mathbb{H}  )$. Its inverse is 
\begin{equation}
\hat \pi _H ^{-1} : S ^4 _R  \rightarrow P ^1 ( \mathbb{H}  ), \quad (q, v) \mapsto \frac{1}{\sqrt{ 2(R-v)}}[q,R-v].
\end{equation} 

In terms of $P ^1 (\mathbb{H}  )$, $p _N $ corresponds to the point $[R,0] $, and stereographic projection from $p _N $   to the map
\begin{equation}
\label{st2} 
\begin{split} 
\hat\phi _N : P ^1 (\mathbb{H}  )\setminus\{[R,0]\} \rightarrow  \mathbb{H}, \quad 
[q ^1 , q ^2 ] & \mapsto Rq ^1 (q ^2 )^{-1},
\end{split} 
\end{equation} 
with inverse
\begin{equation}
\label{invstp} 
\begin{split}
\hat\phi _N ^{-1} :  \mathbb{H} \rightarrow P ^1 (\mathbb{H}  ) \setminus\{[R,0]\}, \quad
q &\mapsto  \sqrt{\frac{R}{ R ^2  + |q |^2 } }[q,R].
\end{split}
\end{equation} 
The map $\hat\phi _N $ can be extended to a conformal isometry from $P ^1 (\mathbb{H}  ) $ to $\hat{\mathbb{H}  } $, by setting $[R,0] \xmapsto{\hat\phi _N } \infty. $
Equations (\ref{st1}), (\ref{st2}) are related by    
\begin{equation} 
  \phi _N\circ \hat\pi _H =\hat\phi _N .
  \end{equation}
The projection $S ^7 _{ \sqrt{ R } } \xrightarrow{ \pi _H } S ^4 _R \xrightarrow{ \phi _N } \mathbb{H}   $  is given by $(q ^1 , q ^2 ) \mapsto  R q ^1  (q ^2 )^{-1} $.

Let $g _{ S ^4 _R }$ be the round metric on $S ^4 _R $. By pulling it back to $ \mathbb{H}  $ via $ \phi _N ^{-1} $ we obtain
\begin{equation}
(\phi _N ^{-1} ) ^\ast g _{ S ^4 _R }
= \frac{4 R ^4 }{( R ^2 + |q |^2) ^2  } \,  g _{ \mathbb{R}  ^4 }.
\end{equation} 

\section{Behaviour of some quantities under conformal transformations}
\label{confb} 
Let $(M,  g) $, $(\hat M, \hat  g) $ be Riemannian $n$-manifolds. For any smooth map $\psi: \hat M \rightarrow M $, $p $-form $\zeta \in \ell ^p ( M) $,  gauge potential $ A $ on $ M $, we have, with $ * =  * _{ g} $, $\hat * =  * _{\hat g} $, $ \hat \zeta =  \psi ^\ast \zeta $, $\hat A = \psi ^\ast A $,
\begin{equation}
\begin{split} 
\psi ^\ast ( *  \zeta )
&=* _{ \psi ^\ast  g }  \hat\zeta , \qquad 
\psi ^\ast( \mathrm{d} _{ A} \zeta )
=\mathrm{d} _{  \hat  A } \hat\zeta .
\end{split} 
\end{equation} 
In particular if $\psi$ is a conformal isometry then  $ \psi ^\ast  g = \ell ^2 \hat g $ and,
taking into account that $ \mathrm{d} ^\dagger _A = (-1)^{ np+n+1} * \mathrm{d} * $, 
\begin{align} 
\label{axcvd} 
\psi ^\ast  (*  \zeta ) &
= \ell ^{ n-2p} \, \hat * \,\hat  \zeta ,\\
\label{astrf} 
\psi ^\ast ( *  \mathrm{d}  _{A}*\zeta  )&
= \ell ^{2p-n-2 } \, \hat * \,  \mathrm{d} _{ \hat A} \,  \hat *  \,  (\ell ^{n-2p} \,  \hat\zeta ).
\end{align} 
For $p =0 $, defining $\triangle _{  A } = \mathrm{d} ^\dagger _{A} \mathrm{d} _{A }$,
\begin{equation}
\label{lc} 
\psi ^\ast\triangle _{ A } \zeta 
=( -1) ^{ n+1} \ell ^{ - n}\, \hat * \, \big ( \mathrm{d} _{ \hat A } \, \hat* \,   (\ell ^{ n - 2 } \mathrm{d} _{ \hat  A } \hat \zeta )\big).
\end{equation}

\section{Computation of  $(\hat \phi _N ^{-1} )^\ast \mathrm{d} _{  \All } ^\dagger  {\delta _\alpha   \Alaa }  $ and $ (\hat \phi _N ^{-1} )^\ast (  G _{ \All}  \mathrm{d} _{  \All } ^\dagger  { \delta _\alpha  \Alaa } )$}
\label{projjj} 
In this appendix we compute some quantities needed to calculate $ \pg{(\delta _\alpha \Alaa )}$.
We make use of the equations of Appendix \ref{confb} for $n =4 $, $M = P ^1 (\mathbb{H}  ) $ with metric 
\begin{equation} 
\gph =  \frac{R ^4}{( R ^2  + |q| ^2 )^2} \, \gh,
\end{equation}
 $\hat M = \mathbb{H}  $ with metric $\gh =\mathrm{d} \bar q \, \mathrm{d} q $.
The conformal isometry $\psi$  is inverse stereographic projection from $[R,0] $,   given by $\hat\phi _N ^{-1} $, see Equation (\ref{invstp}),
and has conformal factor  $\ell ^2 =R ^4 /( R ^2  + |q| ^2 )^2 $.
If $  \zeta  $ is some quantity defined on $P$, we denote by $ \hat \zeta   = ( \hat\phi _N ^{-1}  )^\ast  \zeta    $,
the corresponding quantity on $\hat P $. Vice versa, if $ \hat \upsilon  $ is some quantity defined on $\hat P$, we denote by 
$\upsilon  =( \hat\phi _N   )^\ast \hat \upsilon  $ the corresponding quantity on $P$.

Let us first calculate 
\begin{equation}
\begin{split} 
(\hat \phi _N ^{-1} )^\ast \mathrm{d} _{  \All } ^\dagger  {\delta _\alpha   \Alaa }
= - \frac{(R ^2  + |q| ^2 )^4 }{R ^8 } \, \hat *\,   \mathrm{d} _{\Al}  \hat *    \left( \frac{R ^4 }{(R ^2  + |q| ^2 ) ^2 }  \delta_\alpha  \Ala \right) .
\end{split} 
\end{equation} 
 Write
$\Im(\bar q \, \mathrm{d} q ) =\gamma _k \, \mathrm{d} x ^k $, with 
\begin{equation}
\label{gammas} 
\gamma _1 =-(y\qi + z\qj + w\qk ), \quad 
\gamma _2 =x \qi -w\qj + z\qk, \quad 
\gamma _3 =w\qi + x\qj -y\qk, \quad 
\gamma _4 =-z\qi + y\qj + x\qk.
\end{equation} 
Equation (\ref{oooopl}) becomes
\begin{equation}
\delta _\alpha  \Ala 
= - \frac{1}{2R}  \left(  \frac{\lambda ^2 - R ^2 }{\lambda ^2  + |q |^2 } \right) 
\left[ \left( \frac{2x}{\lambda ^2 + |q| ^2 } \right)  \gamma _k \, \mathrm{d} x ^k +  \mathrm{d} \gamma _1 \right].
\end{equation} 
Write
$\hat u= (R ^4/(R ^2 + |q| ^2 )^2) \delta _\alpha  \Ala$, then
\begin{equation}
\label{vgvhgj} 
\hat * \, \mathrm{d} _{\Al}\,   \hat * \, \hat u 
 =  \mathrm{d} _{\Al}   (\hat u ( \partial _j )) ( \partial _j )
 =\partial _j  (\hat u (\partial _j )) + [\Al (\partial _j ), \hat u( \partial _j )] .
 \end{equation} 
By using $ x ^j \gamma _j  =0$ and (\ref{gammas}) we get
\begin{equation} 
\begin{split} 
\partial _j  (\hat u (\partial _j )) &
=  \frac{2 R ^3  (\lambda ^2 - R ^2  ) \,  \gamma _1 }{(\lambda ^2 + |q| ^2 ) (R ^2 + |q| ^2  )^3  }.
\end{split} 
\end{equation} 
By using $\left[  \gamma _j ,  \partial _j   \gamma _1 \right]  =4 \gamma _1 $  we get
\begin{equation}
\begin{split}
[\Al(\partial _j ) , \hat u (\partial _j ) ] &
= -  \frac{2R ^3 (\lambda ^2 - R ^2 ) \,  \gamma _1 }{(\lambda ^2  + |q| ^2  )^2  ( R ^2 + |q| ^2  )^2 } .
\end{split}
\end{equation} 
Hence
\begin{equation}
\label{jhgjfk} 
\begin{split}
(\hat \phi _N ^{-1} )^\ast \mathrm{d} _{  \All } ^\dagger  {\delta _\alpha   \Alaa }  &
= - \frac{2}{ R ^5 }\frac{( \lambda  ^2 - R ^2  ) ^2(R ^2  + |q| ^2)}{(\lambda ^2 + |q| ^2  )^2  }   \, \gamma _1. 
\end{split}
\end{equation} 

In order to compute $ (\hat \phi _N ^{-1} )^\ast (  G _{ \All}  \mathrm{d} _{  \All } ^\dagger  { \delta _\alpha  \Alaa } )$
it is convenient to first  calculate  $((\hat \phi _N ^{-1} ) ^\ast  \circ   \triangle _{ \All} \circ \phi _N ^\ast )( \gamma _1 /(\lambda  ^2 + |q| ^2))$. Applying (\ref{lc}) gives
\begin{equation}
\begin{split}
\label{nbmcnx} 
&((\hat \phi _N ^{-1} ) ^\ast  \circ   \triangle _{ \All} \circ \phi _N ^\ast )\left(  \frac{\gamma _1}{\lambda  ^2 + |q| ^2} \right)  \\ &
=\frac{ (R ^2 + |q| ^2 ) ^2 }{R ^4 } \hat \triangle _{\Al}  \left( \frac{\gamma _1 }{\lambda ^2 + |q| ^2 } \right)\\ &
- \frac{( R ^2 + |q| ^2 ) ^4 }{R ^8 } \, \hat * \,   \left[ \mathrm{d}  \left( \frac{R ^4 }{(R ^2 + |q| ^2) ^2  } \right)   \wedge \, \hat * \,  \mathrm{d} _{ \Al }\left( \frac{\gamma _1 }{\lambda ^2 + |q| ^2 } \right) \right] \\ &
=\frac{ (R ^2 + |q| ^2 ) ^2 }{R ^4 } \hat \triangle _{\Al}  \left( \frac{\gamma _1 }{\lambda ^2 + |q| ^2 } \right)
+ \frac{4}{R ^4 } \left( \frac{ (R ^2 + |q| ^2)( \lambda ^2 - |q| ^2  )  }{(\lambda ^2 + |q| ^2) ^2 } \right) \, \gamma _1 . 
\end{split}
\end{equation} 
The first term in (\ref{nbmcnx}) is
\begin{equation}
 \begin{split}
 &\hat \triangle _{\Al} \left( \frac{\gamma _1 }{\lambda ^2 + |q| ^2 } \right) 
 =- \mathrm{d} _{\Al} \left[  \mathrm{d} _{\Al} \left( \frac{\gamma _1 }{\lambda ^2 + |q| ^2 }  \right) (\partial _j ) \right] (\partial _j )\\ &
 = - \partial ^2 _j \left( \frac{\gamma _1 }{\lambda ^2 + |q| ^2 }  \right)
  - \partial _j \left( \left[ \Al (\partial _j ) , \frac{\gamma _1 }{\lambda ^2 + |q| ^2 }  \right]  \right) \\ &
  - \left[ \Al (\partial _j ) , \partial _j  \left( \frac{\gamma _1 }{\lambda ^2 + |q| ^2 }  \right) \right] 
  - \left[ \Al (\partial _j ) , \left[\Al (\partial _j ), \frac{\gamma _1 }{\lambda ^2 + |q| ^2 } \right] \right] .
 \end{split}
 \end{equation} 
We compute the various terms as follows,
\begin{align}  
  \partial ^2 _j \left( \frac{\gamma _1 }{\lambda ^2 + |q| ^2 }  \right) &
  = - 4 \gamma  _1  \frac{(3 \lambda ^2 + |q| ^2) }{(\lambda ^2 + |q| ^2 ) ^3 },\\
   \partial _j \left( \left[ \Al (\partial _j ) , \frac{\gamma _1 }{\lambda ^2 + |q| ^2 }  \right]  \right) &
  = \frac{4 \gamma _1 }{(\lambda ^2 + |q| ^2 )^2 }=
  \left[ \Al (\partial _j ) , \partial _j  \left( \frac{\gamma _1 }{\lambda ^2 + |q| ^2 }  \right) \right] ,
\end{align} 
  \begin{equation} 
  \begin{split} 
  \left[ \Al (\partial _j ) , \left[\Al (\partial _j ), \frac{\gamma _1 }{\lambda ^2 + |q| ^2 } \right] \right] &
  = - 8 \gamma _1 \frac{|q| ^2 }{(\lambda ^2 + |q| ^2  )^3 },
  \end{split} 
\end{equation} 
since
 \begin{equation}
 [ \gamma _j ,[ \gamma _j , \gamma _1 ]]
 = \gamma _j \gamma _j \gamma _1 + \gamma _1 \gamma _j \gamma _j -2 \gamma _j \gamma _1 \gamma _j 
 =-3 |q| ^2 \gamma _1 - 3 |q| ^2 \gamma _1 -2 |q| ^2 \gamma _1 
 =-8 |q| ^2 \gamma _1 .
 \end{equation} 
Hence
\begin{equation}
\label{hjkfld} 
\hat\triangle _{\Al} \left( \frac{\gamma _1 }{\lambda ^2 + |q| ^2 }\right) 
= \frac{4 \gamma _1 }{(\lambda ^2 + |q| ^2 ) ^2 },
\end{equation} 
Inserting (\ref{hjkfld})  in (\ref{nbmcnx}) we obtain
\begin{equation}
\begin{split} 
&((\hat \phi _N ^{-1} ) ^\ast  \circ   \triangle _{ \All} \circ \phi _N ^\ast )\left( \frac{\gamma _1 }{\lambda ^2 + |q| ^2 } \right)
= \frac{4}{R ^4 } \frac{ (R ^2 + |q| ^2  ) (R ^2 + \lambda ^2 )}{(\lambda ^2 + |q| ^2 )^2 } \, \gamma _1 .
\end{split} 
\end{equation} 

Let us finally come to $(\hat \phi _N ^{-1} )^\ast (G_{ \All}  \mathrm{d} _{  \All } ^\dagger  { \delta _\alpha  \Alaa }  )$.
Using (\ref{jhgjfk}), we have
\begin{equation}
\label{ggggf} 
\begin{split} 
&(\hat \phi _N ^{-1} )^\ast \left(   G_{ \All}  \mathrm{d} _{  \All } ^\dagger  { \delta _\alpha  \Alaa } \right)  
= (\hat \phi _N ^{-1} )^\ast G_{ \All} \hat \phi _N  ^\ast  (\hat \phi _N ^{-1} ) ^\ast \left(  \mathrm{d} _{ \All } ^\dagger  { \delta _\alpha  \Alaa } \right)  \\ &
= (\hat \phi _N ^{-1} )^\ast G_{ \All}  \hat \phi _N  ^\ast
\left(  - \frac{2}{ R ^5 }( \lambda  ^2 - R ^2  )^2 \frac{  (R ^2  + |q| ^2 )   }{(\lambda ^2 + |q| ^2  )^2  }\, \gamma _1  \right) \\ &
= \left(  (\hat \phi _N ^{-1} )^\ast \circ G_{ \All} \circ   \hat \phi _N  ^\ast  \circ  (\hat \phi ^{-1} _N) ^\ast \circ  \triangle_{\All} \circ \hat\phi _N  ^\ast \right) 
\left( - \frac{1}{2R} \frac{( \lambda  ^2 - R ^2  )^2 }{(\lambda  ^2 + R ^2 )}\frac{\gamma _1 }{\lambda ^2 + |q| ^2 } \right)    \\ &
= - \frac{1}{2R}   \frac {(\lambda  ^2 - R ^2) ^2 }{( \lambda ^2 + R ^2)(\lambda  ^2 + |q| ^2) } \,  \gamma _1 .
\end{split} 
\end{equation}

\bibliographystyle{abbrv}
\bibliography{circle_inv}

\begin{thebibliography}{10}

\bibitem{Atiyah:1987ua}
M.~F. Atiyah.
\newblock Magnetic monopoles in hyperbolic space.
\newblock In {\em Collected works, Volume 5: Gauge Theories}, pages 577--611.
  Oxford University Press, Oxford, 1988.

\bibitem{Atiyah:EHEgiu5y}
M.~F. Atiyah, N.~J. Hitchin, and I.~M. Singer.
\newblock Self-duality in four-dimensional {R}iemannian geometry.
\newblock {\em Proc. R. Soc. Lond. Ser. A}, 362:425--461, 1978.

\bibitem{Belavin:428654}
A.~A. Belavin, A.~M. Polyakov, A.~S. Schwartz, and Y.~S. Tyupkin.
\newblock Pseudoparticle solutions of the {Y}ang-{M}ills equations.
\newblock {\em Phys. Lett. B}, 59:85--87, 1975.

\bibitem{Bielawski:Euclidean}
R.~Bielawski and L.~Schwachh\"ofer.
\newblock Hypercomplex limits of pluricomplex structures and the {E}uclidean
  limit of hyperbolic monopoles.
\newblock {\em Ann. Global Anal. Geom.}, 44:245--256, 2013.

\bibitem{Bielawski:Pluricomplex}
R.~Bielawski and L.~Schwachh\"ofer.
\newblock Pluricomplex geometry and hyperbolic monopoles.
\newblock {\em Commun. Math. Phys.}, 323:1--34, 2013.

\bibitem{Freed:105021}
D.~S. Freed and K.~K. Uhlenbeck.
\newblock {\em Instantons and Four-Manifolds}.
\newblock Springer, New York, 1991.

\bibitem{Groisser:1987uq}
D.~Groisser and T.~H. Parker.
\newblock The {R}iemannian geometry of the {Y}ang-{M}ills moduli space.
\newblock {\em Commun. Math. Phys.}, 112:663--689, 1987.

\bibitem{Groisser:1989vl}
D.~Groisser and T.~H. Parker.
\newblock The geometry of the {Y}ang-{M}ills moduli space for definite
  manifolds.
\newblock {\em J. Diff. Geom.}, 29:499--544, 1989.

\bibitem{Habermann:1988wv}
L.~Habermann.
\newblock On the geometry of the space of ${S}p (1)$-instantons with
  {P}ontrjagin index 1 on the 4-sphere.
\newblock {\em Ann. Global Anal. Geom.}, 6:3--29, 1988.

\bibitem{Habermann:1993ud}
L.~Habermann.
\newblock The ${L}^2$-metric on the moduli space of ${SU} (2)$-instantons with
  instanton number 1 over the {E}uclidean 4-space.
\newblock {\em Ann. Global Anal. Geom.}, 11:311--322, 1993.

\bibitem{Jarvis:1997ws}
S.~Jarvis and P.~Norbury.
\newblock Zero and infinite curvature limits of hyperbolic monopoles.
\newblock {\em Bull. London Math. Soc.}, 29:737--744, 1997.

\bibitem{Manton:2004tk}
N.~Manton and P.~M. Sutcliffe.
\newblock {\em Topological Solitons}.
\newblock Cambridge University Press, Cambridge, 2004.

\bibitem{Naber:1500641}
G.~L. Naber.
\newblock {\em Topology, Geometry and Gauge Fields: Foundations}.
\newblock Springer, New York, 2011.

\bibitem{Speight:1997ub}
J.~M. Speight.
\newblock Low energy dynamics of a ${CP}^ 1$ lump on the sphere.
\newblock {\em J. Math. Phys.}, 36:796 -- 813, 1995.

\bibitem{Uhlenbeck:1982bk}
K.~K. Uhlenbeck.
\newblock Removable singularities in {Y}ang-{M}ills fields.
\newblock {\em Commun. Math. Phys.}, 83:11--29, 1982.

\bibitem{ward:1985}
R.~S. Ward.
\newblock Slowly-moving lumps in the ${CP}^1 $ model in $(2 + 1)$ dimensions.
\newblock {\em Phys. Lett. B}, 158:424 -- 428, 1985.

\end{thebibliography}
\end{document}